\documentclass[letterpaper,11pt]{article}
\usepackage[margin=0.94in]{geometry}

\usepackage{cmap} 
\usepackage[utf8]{inputenc}
\usepackage[english]{babel}
\usepackage[T1]{fontenc}
\usepackage{amsmath}
\usepackage{amssymb}
\usepackage{amsfonts}
\usepackage{bm}
\usepackage{bbm}
\usepackage{xcolor}
\definecolor{blueviolet}{rgb}{0.2, 0.2, 0.6}
\definecolor{webgreen}{rgb}{0,.5,0}
\definecolor{webbrown}{rgb}{.6,0,0}
\usepackage{setspace}
\usepackage[pdftex,
	bookmarks=false,
	colorlinks=true, 
	urlcolor=webbrown, 
	linkcolor=blueviolet, 
	citecolor=webgreen,
	pdfstartpage=1,
	pdfstartview={FitH},  
	bookmarksopen=false
	]{hyperref}
\allowdisplaybreaks
\usepackage{tikz}
\usepackage{braket}
\usepackage[numbers,sort&compress]{natbib}
\usepackage{amsthm}
\usepackage{dsfont}
\usepackage{listings}
\usepackage[capitalize]{cleveref}
\usepackage{appendix}

\usepackage{authblk}

\numberwithin{equation}{section}

\newtheorem{theorem}{Theorem}
\newtheorem{corollary}{Corollary}

\newtheorem{definition}{Definition}

\newtheorem{lemma}{Lemma}
\newtheorem{proposition}{Proposition}
\newtheorem{fact}{Fact}

\theoremstyle{definition}

\newcommand{\vertiii}[1]{{\left\vert\kern-0.25ex\left\vert\kern-0.25ex\left\vert #1 \right\vert\kern-0.25ex\right\vert\kern-0.25ex\right\vert}}

\newcommand{\undersetbrace}[2]{ \underset{#1}{\underbrace{#2}}}
\newcommand{\rom}[1]{\mathtt{\uppercase\expandafter{\romannumeral #1\relax}}}

\DeclareMathOperator{\Tr}{tr}
\DeclareMathOperator*{\E}{{\mathbb{E}}}

\DeclareMathOperator{\poly}{poly}

\newcommand{\ketbra}[2]{\lvert #1 \rangle \! \langle #2 \rvert}

\newcommand{\norm}[1]{\left\lVert#1\right\rVert}

\usepackage{xargs}
\usepackage[colorinlistoftodos,prependcaption,textsize=tiny]{todonotes}
\newcommandx{\lz}[2][1=]{\todo[linecolor=red,backgroundcolor=red!10,bordercolor=red,#1]{LZ: #2}}

 \usepackage{graphicx}
\usepackage{bm}
\usepackage{amsmath}
\usepackage{physics}
\usepackage{amssymb}
\usepackage{amsfonts}
\usepackage{amsthm}
\usepackage{bbm}
\usepackage{mathtools}
\usepackage{hyperref}
\usepackage{braket}
\usepackage[normalem]{ulem}
\usepackage{wrapfig}
\usepackage{tikz}
\usepackage{dsfont}
\usepackage{comment}
\usepackage{thmtools,thm-restate}
\usepackage[export]{adjustbox}

\newtheorem*{theorem*}{Theorem}

\newtheorem*{task*}{Task}
\newtheorem*{proposition*}{Proposition}

\newcommand{\ee}{\end{equation}}

\newcommand{\rmi}{\mathrm{i}}

\DeclareMathOperator{\Wg}{Wg}

\def\:={\,\raisebox{0.85pt}{.}\hspace{-2.78pt}\raisebox{2.85pt}{.}\!\!=\,}
\def\=:{\,=\!\!\raisebox{0.85pt}{.}\hspace{-2.78pt}\raisebox{2.85pt}{.}\,}

\usepackage{authblk}

\begin{document}

\title{Random unitaries in extremely low depth}

\author[1]{Thomas Schuster}
\author[2]{Jonas Haferkamp}
\author[3, 1, 4]{Hsin-Yuan Huang}

\affil[1]{California Institute of Technology}
\affil[2]{Harvard University}
\affil[3]{Google Quantum AI}
\affil[4]{Massachusetts Institute of Technology}
\date{\today}

\maketitle

\begin{abstract}\normalsize
We prove that random quantum circuits on any geometry, including a 1D line, can form approximate unitary designs over $n$ qubits in $\log n$ depth.
In a similar manner, we construct pseudorandom unitaries (PRUs) in 1D circuits in $\poly \log n $ depth, and in all-to-all-connected circuits in $\poly \log \log n $ depth.
In all three cases, the $n$ dependence is optimal and improves exponentially over known results.
These shallow quantum circuits have low complexity and create only short-range entanglement, yet are indistinguishable from unitaries with exponential complexity.
Our construction glues local random unitaries on $\log n$-sized or $\text{poly} \log n$-sized patches of qubits to form a global random unitary on all $n$ qubits.
In the case of designs, the local unitaries are drawn from existing constructions of approximate unitary $k$-designs, and hence also inherit an optimal scaling in $k$.
In the case of PRUs, the local unitaries are drawn from existing PRU constructions.
Applications of our results include proving that classical shadows with 1D log-depth Clifford circuits are as powerful as those with deep circuits, demonstrating superpolynomial quantum advantage in learning low-complexity physical systems, and establishing quantum hardness for recognizing phases of matter with topological order.
\end{abstract}

\pagenumbering{arabic} 
\setcounter{page}{1}

\addtocontents{toc}{\protect\setcounter{tocdepth}{0}}

\section{Introduction}

Random processes are central to computing technologies~\cite{PRNG84,Park1988RandomNG,blum1986simple,rubinstein2016simulation,Santha1986GeneratingQS, goldreich2008computational, haastad1999pseudorandom} and our understanding of the natural world~\cite{einstein1905molecularkinetics,bartlett1949some,allen2010introduction,doob1984classical,wigner1967random,weibull1997evolutionary, hofbauer1998evolutionary}.
In quantum systems, the analog of a random process is a Haar-random unitary operation.
Random unitaries form the backbone of numerous components of quantum technologies, including quantum device benchmarking~\cite{emerson2005scalable,knill2008randomized, elben2020mixed, elben2023randomized}, efficient observable estimation~\cite{brydges2019probing, huang2020predicting, zhao2021fermionic, huang2022learning}, quantum supremacy demonstrations~\cite{arute2019quantum,movassagh2023hardness,bouland2019complexity,morvan2023phase}, and quantum cryptography~\cite{ji2018pseudorandom, ananth2022cryptography, kretschmer2023quantum}.
They also serve as indispensable toy models for complex processes in quantum many-body physics, underlying recent breakthroughs in quantum chaos~\cite{fisher2023random,nahum2017entgrowth,nahum2018operator,cotler2022fluctuations,choi2023preparing,cotler2023emergent}, quantum machine learning~\cite{mcclean2018barren,anschuetz2022quantum,larocca2024review}, and quantum gravity~\cite{sekino2008fast,hayden2007black,yoshida2017efficient,schuster2022many,brown2023quantum,akers2022black}.

\begin{figure}[t]
\centering
\includegraphics[width=0.98\textwidth]{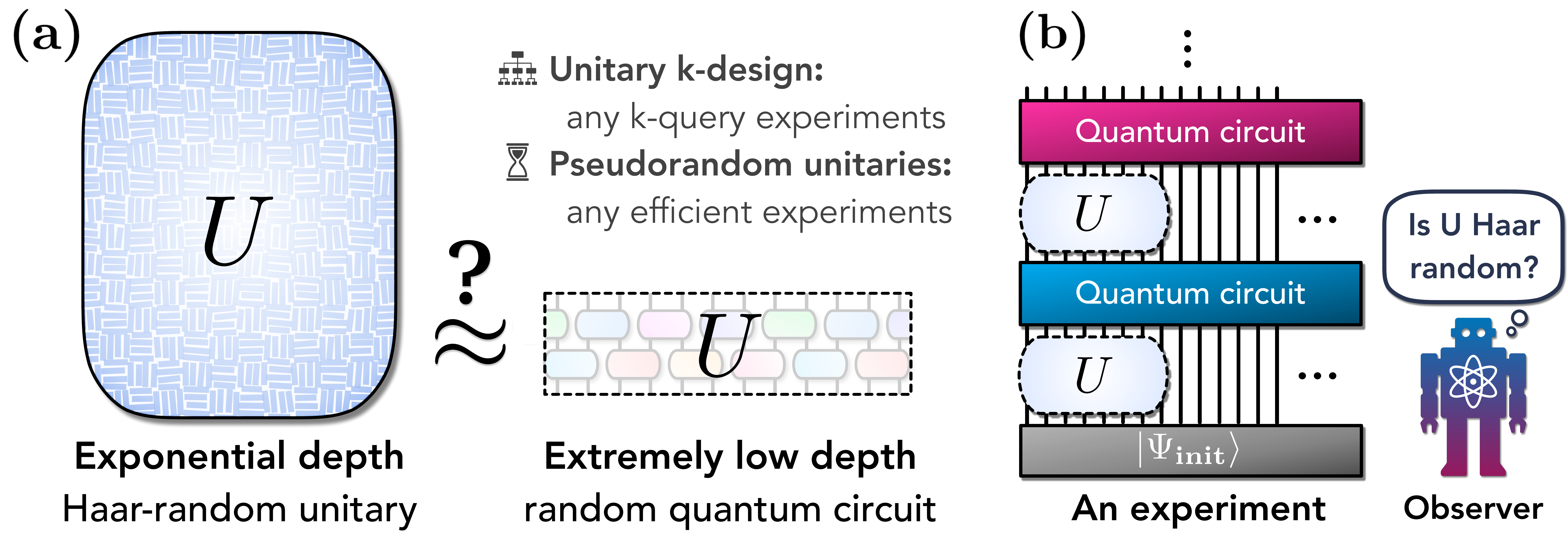}
\caption{\textbf{(a)} The central question we seek to answer is: How shallow can a random quantum circuit be while replicating the behavior of a Haar-random unitary? A Haar-random unitary over $n$ qubits requires a circuit depth that grows exponentially in $n$.
Approximate unitary $k$-designs replicate the behavior of Haar-random unitaries within any quantum experiment that queries the unitary $k$ times.
Pseudorandom unitaries replicate the behavior of Haar-random unitaries within any efficient quantum experiment.
\textbf{(b)} Any quantum experiment can be represented as follows: an observer prepares an initial state $\ket{\Psi_{\mathrm{init}}}$, applies the unitary $U$ many times, interleaved with many quantum circuits for quantum information processing, and concludes by performing a measurement (not shown).
}
\label{fig: setup}
\end{figure}

In all of these applications, a crucial consideration is in what \emph{circuit depth} a random unitary can be generated.
Since a Haar-random unitary requires a depth exponential in the number of qubits, any efficient construction of random unitaries requires a notion of approximation.
To this end, \emph{approximate unitary designs}~\cite{emerson2003pseudo,gross2007evenly,dankert2005efficient,dankert2009exact} and \emph{pseudorandom unitaries}~\cite{PRS2018,PRU2024,SPRU2024} are defined to approximate the action of a Haar-random unitary within any quantum experiment that makes at most $k$ queries to the unitary $U$ (for unitary $k$-designs), or any efficient quantum experiment (for pseudorandom unitaries); see Fig.~\ref{fig: setup}.
Enormous effort has gone into constructing approximate unitary designs and pseudorandom unitaries in as low a depth as possible~\cite{ELL05,harrow2009random,brown2010convergence,brandao2016local,Nakata16,hunter2019unitary,haferkamp2020homeopathy,haferkamp2022random,liu2022estimating,harrow2023approximate,haferkamp2021improved,ho2022exact,jian2023linear,chen2024efficient,metger2024simple,haah2024efficient,chen2024incompressibility,PRS2018,lu2023quantum,chamon2024fast,PRU2024,SPRU2024}.
To date, for both unitary designs and pseudorandom unitaries, all known constructions require a depth polynomial in the number of qubits $n$.

In this work, we show that, in fact, local quantum circuits can form random unitaries in \emph{exponentially lower circuit depths} on any circuit geometry including a 1D line.
We do so by providing a simple construction, which glues together small random unitaries on local patches of $\log n$ or $\poly \log n$ qubits to create an approximate unitary design or pseudorandom unitary on $n$ qubits (Fig.~\ref{fig: construction}).
The small random unitaries are drawn from existing constructions of approximate unitary designs~\cite{chen2024incompressibility} acting on $\log n$ qubits, or pseudorandom unitaries~\cite{metger2024simple,PRU2024,SPRU2024} acting on $\poly \log n$ qubits.
Using the former~\cite{chen2024incompressibility}, we construct approximate unitary $k$-designs with relative error $\varepsilon$ and depth $\tilde{\mathcal{O}}(k) \cdot \log(n/\varepsilon)$ on any circuit geometry\footnote{Recall that the  $\tilde{\mathcal{O}}$ notation denotes $\tilde{\mathcal{O}}(k) = k \poly\log k$. Here we consider a geometry to be given by any connected bounded-degree graph, which includes 1D lines,  2D lattices, a torus, 3D lattices, a binary tree, etc.}.
Using the latter~\cite{metger2024simple,PRU2024,SPRU2024}, we construct pseudorandom unitaries with $\text{poly} \log n$ depth on any geometry, and $\text{poly} \log \log n$ depth in all-to-all-connected circuits.
In all three cases, we show that our achieved scaling in the number of qubits $n$ is optimal.

Our results have wide-ranging applications, owing to the ubiquity of random unitaries across quantum science.
In classical shadow tomography~\cite{huang2020predicting,akhtar2023scalable,bertoni2022shallow,hu2024demonstration,ippoliti2023operator,chen2020robust,Struchalin2021Shadows,hu2023classical,koh2022classical}, our approximate unitary designs enable fidelity estimation using log-depth Clifford circuits instead of linear-depth circuits, with equivalent accuracy guarantees.
This drastically reduces the experimental resources  for classical shadows, opening the door to near-term implementations on many qubits.
In many-body physics, our pseudorandom construction allows us to rigorously establish that recognizing the topological order of a quantum state~\cite{chen2010local, wen2017colloquium, zeng2019quantum,nayak2008non,norman2016colloquium,broholm2020quantum,clark2020observation,semeghini2021probing,satzinger2021realizing,leonard2023realization,iqbal2024topological,iqbal2024non,jiang2012identifying,kim2023universal,huang2022provably} is super-polynomially hard for any quantum experiment.
We describe a number of additional applications, including novel quantum advantages for learning low-complexity dynamics and improved hardness results for random circuit sampling, within the main text and Appendix~\ref{app: applications}.

\begin{figure}[t]
\centering
\includegraphics[width=\textwidth]{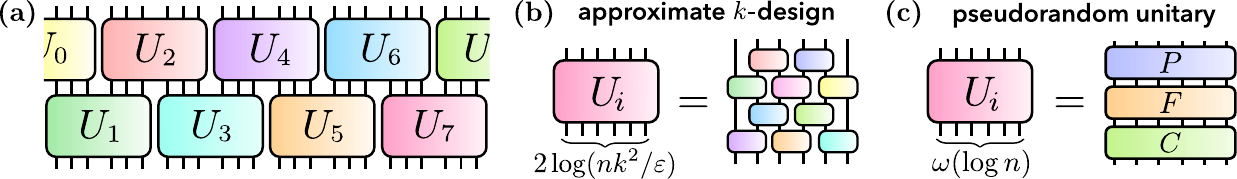}
\caption{\textbf{(a)} Our random unitary ensemble corresponds to a two-layer brickwork circuit, where each small unitary acts on $2\xi$ qubits in each layer.
\textbf{(b)} To generate $\varepsilon$-approximate unitary $k$-designs in $\log n$ depth, we draw each small unitary from an approximate unitary $k$-design on $2\xi = 2\log(nk^2/\varepsilon)$ qubits.
\textbf{(c)} To generate pseudorandom unitaries in $\text{poly} \log n$ depth, we draw each small unitary from a pseudorandom unitary ensemble, such as the $PFC$ ensemble \cite{metger2024simple, PRU2024, SPRU2024}, on $2\xi = \omega(\log n)$ qubits.
}
\label{fig: construction}
\end{figure}

At first glance, the fact that one can create random unitaries in such low depth may seem surprising.
Indeed, many properties of 1D circuits, such as the light-cone volume and entanglement entropy, require at least linear depth to reach their Haar-random values.
The key insight to reconcile this with our results is that such properties \emph{cannot be efficiently measured} in any quantum experiment involving the random unitary $U$~\cite{fn1}.
For this reason, these properties do not form obstacles to realizing pseudorandom unitaries or approximate unitary designs.

\section{Main Results}

We now introduce our random circuit construction (Fig.~\ref{fig: construction}) and present our main results characterizing the circuit depth.
For simplicity, we focus at first on the simplest possible circuit geometry, a 1D line.
Later, we show how to extend our construction to any geometry using graph-theoretic techniques.

Our construction organizes the $n$ qubits along a 1D line into $m$ local patches of $\xi  = n/m$ qubits each.
Our random unitary ensemble $\mathcal{E}$ corresponds to a two-layer circuit, in which each small random unitary acts on two neighboring patches of qubits, and the small random unitaries are arranged in a brickwork fashion between the two layers.
When the small random unitaries over local patches have depth $d$ in terms of two-qubit gates, our proposed random circuit construction has circuit depth $2d$.

\subsection{Random unitary designs}

We can quantify how close the two-layer brickwork ensemble $\mathcal{E}$ is to an $n$-qubit Haar-random unitary using the notion of approximate unitary designs.
An ensemble $\mathcal{E}$ forms an approximate unitary $k$-design if it approximates the Haar ensemble $H$ up to the $k$-th moment.
The gold standard for quantifying the approximation error is given in~\cite{brandao2016local}: A random unitary ensemble $\mathcal{E}$ forms an $\varepsilon$-approximate unitary $k$-design if the following holds,
\begin{align} \label{eq: def unitary design}
    (1-\varepsilon) \, \Phi_H \, \preceq \, \Phi_{\mathcal{E}} \, \preceq \, (1+\varepsilon) \, \Phi_H,
\end{align}
where the quantum channel $\Phi_{\mathcal{E}}(\cdot)$ is defined via
\begin{equation} 
	\Phi_{\mathcal{E}}(A) \coloneqq \E_{U \sim \mathcal{E}} \left[ U^{\otimes k} A U^{\dagger, \otimes k} \right],
\end{equation}
and similarly for the Haar ensemble. 
Here, $\Phi \preceq \Phi'$ denotes that $\Phi'-\Phi$ is a completely-positive map.
The error $\varepsilon$ is commonly known as the relative error or multiplicative error.
Operationally, the relative error guarantees that any quantum experiment that involves $k$ queries to a unitary $U$ sampled from $\mathcal{E}$, produces an output state that is $2\varepsilon$-close in trace distance to the output state when $U$ is sampled from the Haar ensemble; see Appendix~\ref{app: approximate}.

Let us assume that each small random unitary in the two-layer brickwork ensemble $\mathcal{E}$ is drawn randomly and independently from an $\frac{\varepsilon}{n}$-approximate unitary $k$-design on $2\xi$ qubits.
Our main result is that $\mathcal{E}$ forms an $\varepsilon$-approximate unitary $k$-design whenever the number $\xi$ of qubits in each local patch is at least logarithmic in $n, k$ and $\varepsilon^{-1}$.
\begin{theorem}[Gluing small random unitary designs] \label{thm:main-design}
    Given any approximation error $\varepsilon \leq 1$.
    Suppose each small random unitary in the two-layer brickwork ensemble $\mathcal{E}$ is drawn from an $\frac{\varepsilon}{n}$-approximate unitary $k$-design on $2\xi$ qubits with circuit depth $d$.
    Then $\mathcal{E}$ forms an $\varepsilon$-approximate unitary $k$-design on $n$ qubits with depth $2 d$, whenever the local patch size is at least $\xi \geq \log_2(nk^2/\varepsilon)$.
\end{theorem}
\noindent We describe the main ideas and technical lemmas behind the theorem in Section~\ref{sec:proof-overview} and Fig.~\ref{fig: proof overview}.
The proof details are provided in Appendix~\ref{app: unitary designs}.

By utilizing existing constructions of random unitary designs to instantiate each small random unitary, Theorem~\ref{thm:main-design} immediately allows us to construct designs in very low depth.
\begin{corollary}[Low-depth random unitary designs] \label{cor: upper bound design}
    Random quantum circuits over $n$ qubits can form $\varepsilon$-approximate unitary $k$-designs in circuit depth
    \begin{itemize}
    \item $d = \mathcal{O} \big( \! \log\left( n / \varepsilon \right) \cdot k \, \mathrm{poly}\log(k) \big)$, for 1D circuits without ancilla qubits,
    \item $d = \mathcal{O} \big( \! \log \log (n / \varepsilon) \big)$, for all-to-all circuits with $\mathcal{O}( n \log (n / \varepsilon) )$ ancilla qubits and $k \leq 3$.
    \end{itemize}
\end{corollary}
\noindent 
For general $k$, we take each small unitary to be a 1D local random circuit on $2\xi$ qubits, which form $\frac{\varepsilon}{n}$-approximate $k$-designs in depth $d = \mathcal{O}( ( k \xi + \log(n/\varepsilon) )\, \text{poly} \log k)$~\cite{chen2024incompressibility}\footnote{This $k$ dependence is optimal up to $\text{poly} \log k$ factors~\cite{chen2024incompressibility}. Hence, our $k$ dependence is similarly optimal.}.
For $k \leq 3$, we take each small unitary to be a random Clifford unitary~\cite{ webb2015clifford, zhu2017multiqubit}, which can be implemented in depth $d = \mathcal{O}(\log \xi)$ using ancilla qubits and non-local two-qubit gates~\cite{moore2001parallel,jiang2020optimal}.
In both cases, our result exponentially improves the system size $n$ dependence over all known constructions.
We provide a detailed discussion and additional constructions in Appendix~\ref{app: design-depth}.

Finally, we confirm that the system size $n$ dependence of our approximate unitary designs is optimal for both 1D circuits and general all-to-all circuit architectures.

\begin{proposition} \label{prop: lower bound design}
    {\emph{(Depth lower bound for unitary designs)}}
    Any quantum circuit ensemble over $n$ qubits that forms an approximate unitary $2$-design requires circuit depth
    \begin{itemize}
    \item $d = \Omega \big( \! \log n \big)$, for 1D circuits with any number of ancilla qubits,
    \item $d = \Omega \big( \! \log \log n \big)$, for all-to-all circuits with any number of ancilla qubits.
    \end{itemize}
\end{proposition}
\noindent The proposition follows by analyzing the output distribution when a state $U \ket{0^n}$ is measured in a random product basis.
When $U$ has too low of a depth, the output distribution features large fluctuations in its low-weight marginals that differ from those of a Haar-random unitary (Appendix~\ref{app: lower bounds}).
We provide an analogous lower bound in Appendix~\ref{app: lower bounds eps} showing that the $\varepsilon$ dependence of our approximate designs is also optimal, for both 1D circuits and all-to-all circuit architectures.

\subsection{Pseudorandom unitaries}

We can also quantify how close $\mathcal{E}$ is to a Haar-random unitary using the concept of a pseudorandom unitary (PRU) \cite{PRS2018, metger2024simple, chen2024efficient}.
Pseudorandom unitaries are random unitary ensembles that are indistinguishable from the Haar-random ensemble by any efficient quantum algorithm that can query $U$ for any number of times.
In more detail, an $n$-qubit pseudorandom unitary is secure against a $t(n)$-time adversary if it is indistinguishable from a Haar-random unitary by all $t(n)$-time quantum algorithms.
A formal introduction to pseudorandom unitaries is provided in Appendix~\ref{app: PRUs}.

While several constructions of pseudorandom unitaries have been proposed~\cite{PRS2018, metger2024simple, chen2024efficient}, thus far their security has only been shown for non-adaptive quantum algorithms.
Thus, the existence of pseudorandom unitaries has remained a conjecture.
This conjecture is recently resolved in Ref.~\cite{PRU2024} using the so-called $P F C$ construction proposed in Ref.~\cite{metger2024simple}.
Here, $P$ is a quantum-secure pseudorandom permutation, $F$ is a quantum-secure pseudorandom function, and $C$ is a random Clifford unitary\footnote{In more detail, the unitary $P = \sum_{x \in \{ 0,1\}^n} \dyad{\pi(x)}{x}$ implements a pseudorandom permutation, $\pi \in S_{2^n}$, on the computational basis states. The unitary $F = \sum_{x \in \{ 0,1\}^n} (-1)^{f(x)} \dyad{x}{x}$ applies a pseudorandom phase, $f: \{0,1\}^n \rightarrow \{0,1\}$, on the computational basis states. The pseudorandom unitary $U$ is obtained via the composition, $U = PFC$.}.
Assuming no subexponential-time quantum algorithm can solve the Learning With Errors (LWE) problem~\cite{regev2009lattices}, one can efficiently construct~$P$ and~$F$ such that they are indistinguishable from a truly random permutation and function by any subexponential-time quantum adversary~\cite{zhandry2021PRF, zhandry2016note}.
The analysis in Refs.~\cite{PRU2024} then shows that $PFC$ is a pseudorandom unitary with security against any subexponential-time quantum adversary.
This $n$-qubit PRU can be implemented in circuit depth $\poly(n)$ in 1D circuits.
However, the $n$-qubit PFC construction still requires $\poly(n)$ depth in all-to-all-connected circuits due to the circuit depth required in known construction of pseudorandom permutations secure against quantum attacks \cite{zhandry2016note}.
In Ref.~\cite{SPRU2024}, the authors present a construction involving random Clifford circuits and a constant number of pseudorandom functions to achieve $\poly \log(n)$ depth in all-to-all-connected circuits.

To construct pseudorandom unitaries with even lower circuit depth, let us draw each small random unitary in the two-layer brickwork ensemble $\mathcal{E}$ from a PRU ensemble on $2\xi$ qubits, and set $\xi = \omega(\log n)$.
We assume each small unitary is secure against $\poly(n)$-time quantum adversaries.
Since $\xi = \omega(\log n)$, a $\poly(n)$-time adversary is an $\exp(o(\xi))$-time adversary; hence, this is automatically satisfied by drawing each small unitary from a PRU ensemble with subexponential security, as above.
Our main finding is that the resulting ensemble $\mathcal{E}$ is an $n$-qubit pseudorandom unitary ensemble.

\begin{theorem}[Gluing small pseudorandom unitaries] \label{thm:main-PRUs}
    Let $n$ be the number of qubits in the whole system and $\xi = \omega(\log n)$ be the number of qubits in each local patch.
    Suppose each small random unitary in the two-layer brickwork ensemble $\mathcal{E}$ is a $2\xi$-qubit pseudorandom unitary secure against $\poly(n)$-time adversaries.
    Then $\mathcal{E}$ forms an $n$-qubit pseudorandom unitary secure against $\poly(n)$-time adversaries.
\end{theorem}
\noindent An overview of our proof is provided in Section~\ref{sec:proof-overview}, and details are given in Appendix~\ref{app: proof of main PRUs}.

Using the $PFC$ construction~\cite{metger2024simple, PRU2024} or the permutation-free construction~\cite{SPRU2024} to instantiate each small random unitary, we obtain $n$-qubit pseudorandom unitaries in the following low circuit depths.
A detailed proof is given in Appendix~\ref{app: PRU depth scaling}.
\begin{corollary}[Low-depth pseudorandom unitaries] \label{cor:pseudorandom unitaries}
Under the conjecture that no subexponential-time quantum algorithm can solve LWE, random quantum circuits over $n$ qubits can form pseudorandom unitaries secure against any polynomial-time quantum adversary in circuit depth
\begin{itemize}
\item $d = \poly \log n$, for 1D circuits,
\item $d = \poly \log \log n$, for all-to-all circuits.
\end{itemize}
\end{corollary}
\noindent Our depth scaling improves exponentially over all known proposals for pseudorandom unitaries \cite{PRS2018, metger2024simple, chen2024efficient, PRU2024, SPRU2024}, which require $\mathrm{poly}(n)$-depth for 1D circuits and $\mathrm{poly}\log(n)$-depth for general circuits. 

As in the previous section, our scaling of the depth is in fact optimal.
This follows from recent work on shallow quantum circuits, which provides a polynomial-time algorithm to learn any 1D circuit of depth $\mathcal{O}(\log n)$, and any general circuit of depth $\mathcal{O}(\log \log n)$~\cite{huang2024learning}.
If one can learn a circuit, one can trivially distinguish it from a Haar-random unitary.
Thus, 1D circuits require $\mathrm{poly} \log n$ depth to form PRUs, and general circuits require $\mathrm{poly} \log \log n$ depth.
We remark that the precise polynomial degree in Corollary~\ref{cor:pseudorandom unitaries} depends on the specific LWE problem one conjectures to be hard.

\begin{figure}
    \centering
    \includegraphics[width=1.0\textwidth]{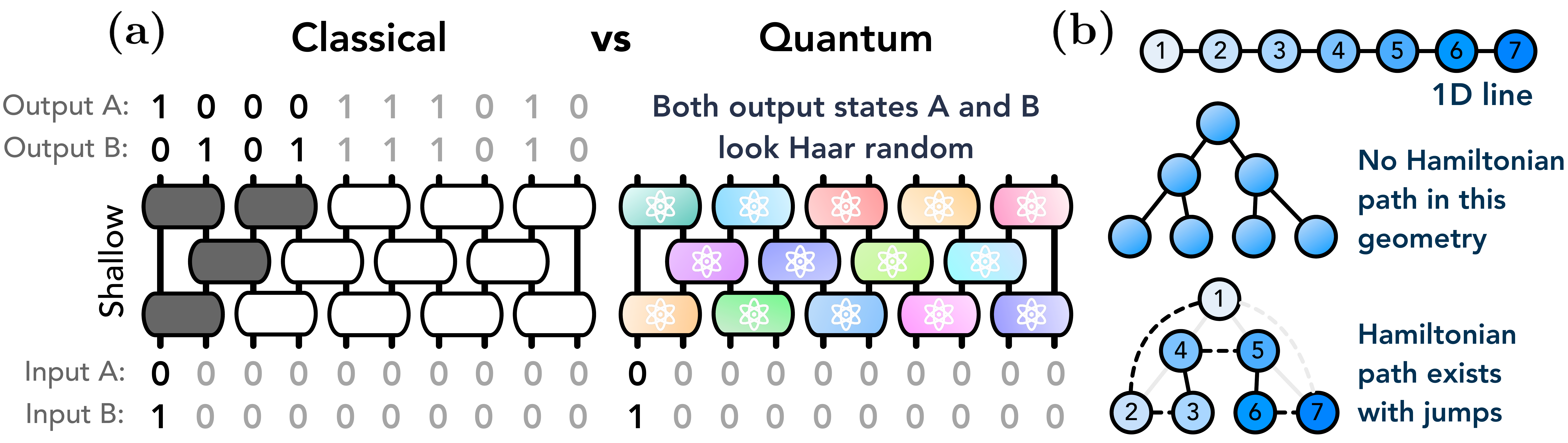}
    \caption{\textbf{(a)} Shallow random classical circuits cannot look uniformly random. In contrast, our results show that shallow random quantum circuits can already look Haar-random.
    \textbf{(b)} To create random unitaries on any circuit geometry, we implement a 1D random circuit along a Hamiltonian path of the geometry. While Hamiltonian paths do not exist in any geometry, when jumping to constant-distance neighbor is allowed, they always exist and are efficient to construct. }
    \label{fig:CvQ-geo}
\end{figure}

\subsection{Comparison between quantum and classical reversible circuits}

To gain a better intuition for our findings, it is helpful to contrast our results on random quantum circuits with the behavior of random \emph{classical} circuits [Fig.~\ref{fig:CvQ-geo}(a)].
A classical reversible circuit over $n$ bits corresponds to a permutation on the $2^n$ bitstrings, $\{0, 1\}^n$.
To determine how quickly a classical circuit can resemble a random permutation, let us consider the output of such a circuit when applied to two input bitstrings, $b_0, b_1 \in \{0, 1\}^n$, that differ only in their first bit.
For these inputs, a random permutation will output two bitstrings with almost no correlation between them.
In contrast, any classical circuit with depth less than $n / 4$ in 1D, or less than $\log_2(n / 4)$ in all-to-all circuits, will produce two output strings that are identical on $3n / 4$ bits.
This follows because the light-cone of the first bit has size at most $n / 4$. 
Hence, to match even the second moment of a random permutation, reversible classical circuits require depth $d = \Omega(n)$ in 1D, and $d = \Omega(\log n)$ in all-to-all circuits.
These classical circuit depths are exponentially larger than our obtained quantum circuit depths.

Physically, this exponential reduction in the quantum circuit depth is made possible by the abundance of non-commuting observables in quantum mechanics.
To distinguish a classical or quantum circuit from a random permutation or unitary, an observer must eventually measure the state of the system in some chosen basis.
In a classical circuit, all observables commute.
Thus, information about input bits far from the first bit can only scramble into observables that commute with the measurement basis (i.e.~the computational basis).
This causes the information to impact the measurement results, which allows the observer to easily distinguish whether two input strings have the same values far from the first bit.

In contrast, quantum circuits are able to locally hide information into non-commuting observables. 
When applied to a quantum state within an experiment, a random unitary from our  ensemble will scramble information about the quantum state over local regions of $4\xi$ qubits.
There are a large number, $e^{\mathcal{O}(\xi)}$, of observables on these qubits that information can scramble into; moreover, most of these observables do not commute with one another.
This means that it is exceedingly unlikely that the randomly scrambled information will be contained in observables that commute with the observer's measurement basis.
This causes the measurement outcome to become nearly independent of details about the initial quantum state, in such a way that the unitary appears Haar-random.
In this manner, non-commuting observables allow quantum circuits to appear quantumly random exponentially faster than classical circuits can appear classically random.

\subsection{Comparison between unitary and orthogonal circuits}

Perhaps surprisingly, even orthogonal (i.e.~real) quantum circuits cannot generate approximate \emph{orthogonal} $2$-designs in less than linear depth. 
That is, complex numbers are a necessary ingredient for low-depth random unitaries.
Fundamentally, this arises from the fact that the EPR state is stabilized by the tensor square of any orthogonal matrix, $(O \otimes O) \ket{\Psi_{\text{EPR}}} = (OO^T \otimes \mathbbm{1}) \ket{\Psi_{\text{EPR}}} = \ket{\Psi_{\text{EPR}}}$, since $OO^T = \mathbbm{1}$ and $(\mathbbm{1} \otimes A) \ket{\Psi_{\text{EPR}}} = (A^T \otimes \mathbbm{1}) \ket{\Psi_{\text{EPR}}}$ for any operator $A$.
To see how this lower bounds the depth of any orthogonal 2-design, consider applying $O \otimes O$ to the perturbed EPR state, $(Z_1 \otimes \mathbbm{1}) \ket{\Psi_{\text{EPR}}}$, where $Z_1$ is a Pauli operator on the first qubit.
The resulting state is equal to $(O Z_1 O^\dagger \otimes \mathbbm{1}) \ket{\Psi_{\text{EPR}}}$ from the above.
When $O$ is a Haar-random orthogonal matrix, the operator $O Z_1 O^\dagger$ acts non-trivially on all $n$ qubits of the system and the resulting state is highly-entangled. 
On the other hand, when $O$ is a low-depth quantum circuit, $O Z_1 O^\dagger$ must act trivially on any qubit outside the light-cone of the first qubit. 
Therefore, such outer qubits  remain locally in the EPR state, which can be easily detected.
This shows that random orthogonal circuits must have depth $d=\Omega(n)$ in 1D, and $d=\Omega(\log n)$ in all-to-all circuits, to match even the second moment of a Haar-random orthogonal matrix.

\subsection{Creating random unitaries on any geometry}

We provide two methods to extend our construction from 1D circuits to any circuit geometry (see Appendix~\ref{app: any geometry} for complete details).
We consider a geometry to be any connected bounded-degree graph, where each qubit is a node on the graph and two nodes are connected by an edge if one can implement a two-qubit gate between the two qubits~\cite{haah2024learning, huang2023learningb}.
This graph-theoretic definition includes all common physical geometries, such as a 1D circle, a 3D plane, a torus, a binary tree, a hyperbolic space, and a highly connected expander graph.

Our first method shows that a depth-$d$ quantum circuit on a 1D line can be implemented on any geometry in circuit depth $\mathcal{O}(d)$.
We do so by efficiently constructing a Hamiltonian path that goes through every node in the graph of the geometry  exactly once.
Although a Hamiltonian path does not always exist and is generally hard to find, we show that when one allows \emph{jumps} to a constant-distance neighbor on the graph, a Hamiltonian path always exists and can be found efficiently [Fig.~\ref{fig:CvQ-geo}(b)].
The two-qubit gates between constant-distance neighbors can then be implemented using a carefully-designed swap network.
Our second method extends Theorem~\ref{thm:main-design} to general two-layer brickwork circuits.
This allows one to glue together small random unitaries on a wide variety of geometries of interest, such as a 2D circuit consisting of many overlapping squares [Fig.~\ref{fig: proof overview}(c)].
Both methods apply both to our construction of low-depth unitary designs and low-depth PRUs.

\section{Applications}
\label{sec: applications}

Let us now turn to applications of our results.
We summarize a handful of the most prominent applications below, and provide full details in Appendix~\ref{app: applications}.

\vspace{3.5mm}
\noindent \textbf{Provably-efficient shallow classical shadows:} 
Classical shadow estimation utilizes random measurements to achieve rapid estimations of many non-commuting observables~\cite{huang2020predicting}. 
Traditionally, these measurements utilize random Clifford unitaries on $n$ qubits, which require a linear circuit depth to implement.
In Appendix~\ref{app: classical shadows}, we show that these deep Clifford unitaries can be replaced by Clifford circuits with $\log n$ depth from our construction, while retaining essentially the same guarantees on the protocol's sample complexity.
This depth scaling confirms prior conjectures in Refs.~\cite{bertoni2022shallow,ippoliti2023operator}.

A key motivation for these \emph{shallow shadow} protocols is to address experimental limitations due to noise in quantum devices.
To this end, we provide a rough estimate for the number of qubits that can be reached with our shallow circuit construction as opposed to the traditional approach.
For leading current noise rates $\gamma \approx 0.5\%$, one can perform roughly $1/\gamma \approx 200$ circuit gates to within good many-body fidelity.
With linear-depth Clifford circuits, this limits shadow tomography to small numbers of qubits, $n  \approx 15$ (i.e.~$n^2 \approx 200$).
On the other hand, our approach opens the door to high-fidelity shadow estimation on up to $n \approx 40$ qubits (i.e.~$n \log_2 n \approx 200$).
Due to its favorable scaling, the advantage of our approach will become even more stark as noise rates improve.

\vspace{3.5mm}
\noindent \textbf{Quantum hardness of recognizing topological order: }The detection of topologically-ordered phases of matter has remained a notoriously difficult challenge across both materials and atomic, molecular, and optical experiments~\cite{nayak2008non,norman2016colloquium,broholm2020quantum,clark2020observation,semeghini2021probing,satzinger2021realizing,leonard2023realization,iqbal2024topological,iqbal2024non,jiang2012identifying,kim2023universal,huang2022provably}.
From a quantum information perspective, one of the defining features of topological order is its invariance under the application of any low-depth local unitary circuit~\cite{chen2010local, wen2017colloquium, zeng2019quantum}. 
From this defining property, we apply Corollary~\ref{cor:pseudorandom unitaries} to prove that recognizing topological order is, in fact, quantumly hard at any poly-logarithmic depth (Appendix~\ref{app: topological order}):
\begin{corollary}[Hardness of recognizing topological order] \label{cor: topological}
Consider any definition of topological order such that (i) the product state has trivial order and the toric code state has non-trivial topological order, and (ii) the topological order of these states is preserved under any depth-$\ell$ geometrically-local circuit. 
Then, recognizing topological order is quantum computationally hard for any $\ell = \Omega( \mathrm{poly} \log n)$.
\end{corollary}
\noindent  The criteria of the corollary apply to nearly every existent definition of topological order~\cite{fn2}.

\begin{figure}
    \centering
    \includegraphics[width=1.0\textwidth]{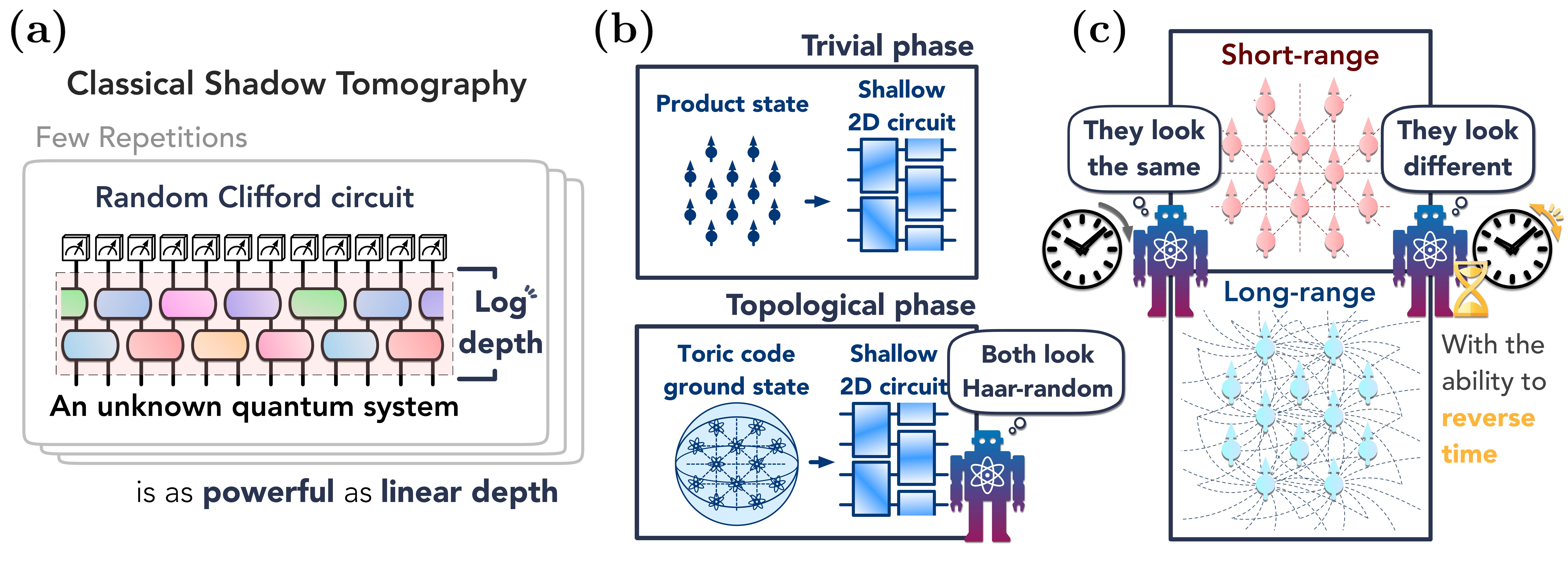}
    \caption{
    The low-depth unitary designs and pseudorandom unitaries that we construct have broad applications, several of which are depicted here.
    \textbf{(a)} \emph{Log-depth classical shadows:} Our shallow unitary 3-designs enable provably-efficient classical shadow tomography using $\log(n)$-depth random Clifford circuits instead of linear depth.
    \textbf{(b)} \emph{Quantum hardness of recognizing topological order:} By applying our shallow pseudorandom unitaries to product and toric code states, we generate pseudorandom states with trivial and topological order, respectively. This demonstrates that recognizing topological order in an unknown state is quantumly hard.
    \textbf{(c)} \emph{Power of time-reversal in learning:} We establish that quantum experiments capable of reversing time-evolution can exhibit super-polynomial advantages over conventional experiments. We prove this in a simple example where one wishes to detect whether long-range couplings are present in a strongly-interacting dynamical quantum system.
    }
    \label{fig:applications}
\end{figure}

\vspace{3.5mm}
\noindent \textbf{Quantum advantage for learning low-complexity quantum systems: } 
Our results immediately imply that several well-known quantum learning advantages~\cite{aharonov2021quantum, huang2021information, huang2021quantum, chen2021exponential} that so far only apply to highly-complex systems also hold in low-complexity systems.
A particularly relevant example concerns the task of distinguishing a random unitary process from a fully depolarizing channel~\cite{huang2021quantum, chen2021exponential}.
This can be solved efficiently by an observer with quantum access to the process of interest, but requires super-polynomially many queries for any classical observer~\cite{huang2021quantum, chen2021exponential}.
Thus far, this advantage has only been known for Haar-random unitaries, which require $\exp(\mathcal{O}(n))$ circuit depth and are thus poor models for quantum processes encountered in the physical world.
Our results show that this separation holds computationally for circuits with double-exponentially smaller depth, $\text{poly} \log n$.
Moreover, we show a similar quantum-classical separation for learning the entanglement structure of states generated by shallow quantum circuits.
We refer to Appendix~\ref{app: learning advantages} for additional details.

\vspace{3.5mm}
\noindent \textbf{The power of time-reversal in learning: } Leveraging an assortment of interaction engineering techniques, many modern quantum experiments have the ability to time-reverse their dynamics~\cite{baum1985multiple,choi2017dynamical,garttner2017measuring,davis2019photon,mi2021information,braumuller2022probing,sanchez2021emergent,dominguez2021decoherence,colombo2022time,li2023improving}.
One surprising application of our results is to show that such experiments allow one to learn properties of quantum dynamics exponentially more efficiently than conventional experiments without time-reversal~\cite{cotler2023information,schuster2023learning}.
We demonstrate this in a simple physically-motivated example.
Suppose one wishes to detect whether a quantum circuit contains only local interactions (e.g.~in 2D), or instead, a combination of local interactions and small long-range couplings.
From Corollary~\ref{cor:pseudorandom unitaries}, this task is immediately  hard for circuits with poly-logarithmic depth, since at such depths one cannot distinguish either circuit from a Haar-random unitary.
On the other hand, in Appendix~\ref{app: time-reversal}, we show that this task is \emph{easy} for experiments that can implement both the circuit $U$ and its time-reverse~$U^\dagger$.
This follows from standard measurement protocols for so-called out-of-time-order correlors~\cite{shenker2014black,swingle2016measuring,yao2016interferometric,vermersch2019probing,xu2024scrambling}.
We remark that this separation is intrinsically quantum mechanical, by similar light-cone arguments as below  Theorem~\ref{thm:main-design}.

\vspace{3.5mm}
\noindent \textbf{Output distributions of random quantum circuits: }
Random circuit sampling (RCS) is a leading current candidate for quantum computational supremacy \cite{arute2019quantum}.
Hardness results on RCS typically rely on two ingredients: worst-case hardness, and anti-concentration of the random circuits' output distributions~\cite{aaronson2011computational,movassagh2023hardness,bouland2019complexity}. 
Thus far, these two ingredients have only coincided in relatively deep, $d = \Omega(\sqrt{n})$, 2D random circuits~\cite{dalzell2022random,harrow2023approximate}.
Our construction of approximate unitary 2-designs [Fig.~\ref{fig: proof overview}(c)] yields a 2D random circuit ensemble with anti-concentration and worst-case hardness in depth $\log n$.
We can also prove several stronger statements about the output distributions of random quantum circuits by building upon our approximate $k$-designs for larger~$k$.
Namely, we show that, at depth $\Omega(\log n)$, the output distributions of random circuits are far-from-uniform with probability \emph{close to one}, and that at depth $\mathrm{poly}\log(n)$, the output distributions are both far-from-uniform and yet computationally-indistinguishable from the uniform distribution; see Appendix~\ref{app: output distributions}.

\section{Proof overview}
\label{sec:proof-overview}

In this section, we overview the key ideas that lead to our proofs of Theorems~\ref{thm:main-design},~\ref{thm:main-PRUs} and Proposition~\ref{prop: lower bound design}.
We refer to the Appendix for full details of each proof.

\subsection{Gluing small random unitary designs (Theorem~\ref{thm:main-design})}

Our proof proceeds in four steps, each of which may be of interest for future analyses and applications of random unitaries.
A common theme in our analysis is the decomposition of the Haar twirl as a sum over permutation operators that act on $\mathcal{H}^{\otimes k}$,
\begin{equation} \label{eq: exact Haar twirl}
    \Phi_H(A) = \sum_{\sigma,\tau \in S_k} \Tr(A \sigma^{-1}) \Wg_{\sigma,\tau}(D) \tau.
\end{equation}
Here, $\text{Wg}_{\sigma, \tau}( D )$ are the elements of the $k! \times k!$ \emph{Weingarten matrix}, which is defined as the inverse of the \emph{Gram matrix},  $G_{\sigma,\tau}(D) = \text{tr}( \sigma \tau^{-1} )$, given by the inner products of the permutation operators.
Fundamentally, our results below follow from a simple fact: When the Hilbert space dimension $D = 2^n$ is large, the permutation operators are approximately orthogonal to one another~\cite{harrow2023approximate}. 
This implies that the Gram and Weingarten matrices are nearly proportional to the identity, $G_{\sigma,\tau}(D) \approx D^k \cdot \mathbbm{1}_{\sigma,\tau}$ and $\text{Wg}_{\sigma, \tau}( D ) \approx D^{-k} \cdot \mathbbm{1}_{\sigma,\tau}$, up to small corrections of order $\mathcal{O}(k^2/D)$.


As our first step, we utilize the approximate orthogonality of the permutations to provide a simple yet accurate approximation for the twirl over a Haar random unitary.
\begin{lemma}[Approximate Haar twirl] \label{lemma: approx Haar twirl}
    For any $k^2 \leq 2^n$, the Haar twirl is approximated by, 
    \begin{equation}
        \Phi_a(A) \equiv \frac{1}{2^{nk}} \sum_{\sigma \in S_k} \Tr(A \sigma^{-1}) \sigma,
    \end{equation}
    up to relative error, $(1-\varepsilon) \Phi_a \preceq \Phi_H \preceq (1+\varepsilon) \Phi_a$, with $\varepsilon = k^2/2^n$.
\end{lemma}
\noindent This proposition dramatically simplifies our analysis, since it replaces the Weingarten matrix elements in Eq.~(\ref{eq: exact Haar twirl}), which can be difficult to analyze, with a simple delta function, $\delta_{\sigma,\tau} / D^k$.

As our second step, we provide a key technical lemma, which bounds the \emph{relative} error $\varepsilon$ of any approximate unitary $k$-design [Eq.~(\ref{eq: def unitary design})], in terms of its \emph{additive} error (with respect to the approximate Haar twirl) when applied to the EPR state.
%
%
\begin{lemma}[Unitary designs from EPR states]\label{lemma: relative to additive}
    A random unitary ensemble $\mathcal{E}$ acting on an $n$-qubit Hilbert space $\mathcal{H}$ forms an $\varepsilon$-approximate unitary $k$-design with error
    \begin{equation} \label{eq: relative error lemma}
        \varepsilon = \frac{4^{nk}}{k!} \cdot \left( 1 + \frac{k^2}{2^{n}} \right) \cdot \big\lVert \, [ (\Phi_\mathcal{E} - \Phi_a) \otimes \mathbbm{1} ] (  P^{\mathrm{EPR}}  ) \, \big\rVert_\infty + \frac{k^2}{2^n},
    \end{equation}
    for any $k^2 \leq 2^n$, where $P^{\mathrm{EPR}}$ is the projector onto the EPR state on $\mathcal{H}^{\otimes k} \otimes \mathcal{H}^{\otimes k}$.
\end{lemma}
\noindent The lemma is useful because bounding the additive error is, in most cases, substantially easier than bounding the relative error.
%
%
We note that the lemma improves upon related existing bounds~\cite{brandao2016local} by a factor of $k!$.
This factor of $k!$ improvement is essential for our construction of low-depth pseudorandom unitaries, in which we take $k$ to be superpolynomial in $n$.

\begin{figure}[t]
\centering
\includegraphics[width=\textwidth]{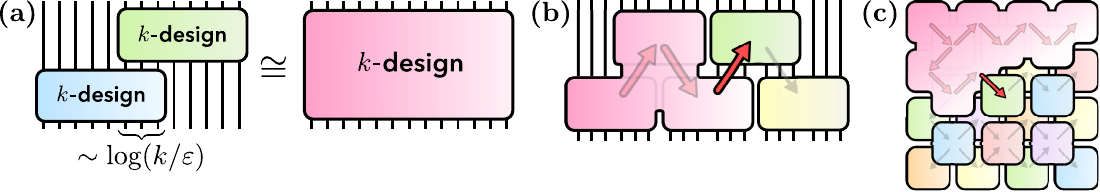}
\caption{\textbf{(a)} A key step in our proof is to show that one can ``glue'' two approximate unitary $k$-designs into a larger $\varepsilon$-approximate unitary $k$-design.
This holds whenever the two unitaries are applied in sequence and overlap on at least $\sim \! \text{log}(k/\varepsilon)$ qubits.
\textbf{(b)} To prove Theorems~\ref{thm:main-design},~\ref{thm:main-PRUs}, we apply this gluing lemma $n/\xi$ times as shown (arrows).
At each application, we glue one additional small unitary (green) into a larger unitary design on all qubits to its left (pink).
\textbf{(c)} The gluing lemma also enables us to immediately extend our theorems to 2D circuits, by applying the lemma in the order shown (arrows), as well as other geometries.
}
\label{fig: proof overview}
\end{figure}

The third step in our proof arrives at the critical difference between quantum  and classical circuits.
We show that one can ``glue'' two unitary designs together to form a larger unitary design, as long as the two unitaries overlap on a relatively small number of qubits; see Fig.~\ref{fig: proof overview}(a). 

\begin{lemma}[Gluing two random unitaries; informal] \label{lemma: AB BC to ABC}
Let $A$, $B$, $C$ be three disjoint subsystems. 
Consider a random unitary given by $V_{ABC} = U_{AB} U_{BC}$, where $U_{AB}$ and $U_{B C}$ are drawn from $\varepsilon_{1}$ and $\varepsilon_{2}$-approximate unitary $k$-designs, respectively. 
Then $V_{ABC}$ is an $\varepsilon$-approximate unitary $k$-design with
\begin{equation} \label{eq: ABC error simplified}
    1+\varepsilon \leq (1+\varepsilon_{1})(1 + \varepsilon_{2}) \left(1 + 2 k^2 / 2^{|B|} \right),
\end{equation}
as long as the number of qubits in $B$ satisfies $|B| \geq \log_2(2k^2)$. \end{lemma}
\noindent In classical circuits, such a lemma cannot exist because any two input bitstrings that differ only on $A$ would have outputs that do not differ on $C$.
In quantum circuits, this is not an issue, because information on $C$ is scrambled into random non-commuting observables.
As long as $|B| \gtrsim \log(k/\varepsilon)$, this information is effectively hidden from any quantum experiment~\cite{fn3}.

Our proof of Lemma~\ref{lemma: AB BC to ABC} follows relatively quickly from the approximate orthogonality of the permutation operators.
Leveraging Lemma~\ref{lemma: approx Haar twirl}, we have (see Appendix~\ref{app: gluing} for full details)
\begin{equation} \label{eq: phiE on EPR} 
    \big[  \Phi_{\mathcal{E}} \otimes \mathbbm{1} \big] ( P_{\text{EPR}} ) \approx 
        \frac{1}{2^{2nk}} \sum_{\sigma,\tau \in S_k} \frac{1}{2^{|B| k}} G_{\sigma,\tau}(2^{|B|})  \cdot \tau_{AB} \sigma_{C} \otimes \tau_{A} \sigma_{BC},
\end{equation}
up to small relative error.
Here, $\sigma$ and $\tau$ arise from the twirl over $U_{BC}$ and $U_{AB}$, respectively, and the factor $G_{\sigma,\tau}(2^{|B|}) = \tr_B( \sigma_B^{} \tau^{-1}_B )$ arises when the output of first twirl is traced with the input of the second twirl.
Intuitively, when $2^{|B|}$ is large, we can replace $G_{\sigma,\tau}(2^{|B|})/2^{|B| k}$ above with a delta function, $\delta_{\sigma,\tau}$, since the permutations are approximately orthogonal.
This substitution yields the \emph{exact} same state as the approximate Haar twirl, $\Phi_a$, on the entire system $ABC$.
By making this intuition precise, we bound the operator norm in Lemma~\ref{lemma: relative to additive} and thus show that the ensemble forms an approximate unitary design.

Our proof of Theorem~\ref{thm:main-design} follows immediately from Lemma~\ref{lemma: AB BC to ABC}.
We proceed unitary-by-unitary from left to right as in Fig.~\ref{fig: proof overview}(b).
At each step, we apply Lemma~\ref{lemma: AB BC to ABC} to glue the next unitary into a large design formed by all of the preceding unitaries.
After $m-2$ steps, this produces an approximate design on all $n$ qubits.
The design has relative error 
\begin{equation}
    \approx (m-1) \cdot \frac{\varepsilon}{n} + (m-2) \cdot \frac{5k^2}{2^{\xi}},
\end{equation}
since there are $m-1$ small random unitaries, each drawn from an $\frac{\varepsilon}{n}$-approximate $k$-design, and we applied Lemma~\ref{lemma: AB BC to ABC} a total of $m-2$ times.
Setting $\xi = \Omega(\log_2(n k^2 / \varepsilon))$ and recalling $\xi = n/m$ gives Theorem~\ref{thm:main-design}. We refer to Appendix~\ref{app: unitary designs} for a detailed proof including tight constant factors.

\subsection{Depth lower bound for unitary designs (Proposition~\ref{prop: lower bound design})}

To lower bound the depth of any approximate unitary 2-design, we apply a random unitary from the ensemble to the zero state, and consider the output distribution, $p_{U,v}(x) = | \bra{x} v^\dagger U \ket{0^n} |^2$, when the resulting state, $U \ket{0^n}$, is measured in a random product basis, $\{ v \ket{x} \}$. Here $v$ is a tensor product of random single-qubit unitaries, and $x \in \{ 0,1 \}^n$.
When $U$ is Haar-random, the distribution is nearly flat across the $2^n$ basis states.
This implies that the expected collision probability, $\mathbbm{E}_{U,v} [ \sum_x p_{U,v}(x)^2 ] \approx 2/2^n$, is almost maximally small.
Since the collision probability corresponds to the expectation value of a positive operator, $\mathbbm{E}_v [ \sum_x (v \dyad{x} v^\dagger)^{\otimes 2}]$, in the state $\Phi_\mathcal{E}( \dyad{0^n}^{\otimes 2} )$, this small value must be replicated by any approximate unitary 2-design.

Let us now analyze what happens when $U$ has too low a depth.
We will track how local information about the initial state $\ket{0^n}$ affects the distribution $p_{U,v}(x)$.
Consider the $n$ single-qubit stabilizers of $\ket{0^n}$.
Each stabilizer can scramble to at most $L$ qubits under $U$, where $L$ is the size of the light-cone of $U$.
Intuitively, each scrambled stabilizer will commute with the basis of the random  product measurement with probability at least $1/3^{L}$.
This leads to small bias, $\Omega(1/3^L)$, in the marginals of $p_{U,v}(x)$ that contain the stabilizer, which, we find, leads to a small increase, $\Omega( 1/3^L \cdot 1/2^n)$, in the collision probability.
Summing over the $n$ single-qubit stabilizers, we obtain a strict lower bound on the expected collision probability, $\mathbbm{E}_{U,v} [\sum_x p_{U,v}(x)^2] \geq (n / 3^{L})  / 2^{n}$  (see Appendix~\ref{app: lower bounds} for details).
For 1D circuits, the light-cone grows linearly, $L = 2d$, so one requires $d = \Omega(\log n)$ to replicate the Haar collision probability.
For all-to-all circuits, the light-cone grows exponentially, $L = 2^d$, so one requires $d = \Omega(\log \log n)$.

\subsection{Gluing small pseudorandom unitaries (Theorem~\ref{thm:main-PRUs})}

To prove that $\mathcal{E}$ forms a pseudorandom unitary ensemble, we introduce an additional unitary ensemble, $\mathcal{E}^*$, in which each small random unitary in the two-layer brickwork ensemble is drawn from the Haar ensemble on $2\xi$ qubits.
We establish Theorem~\ref{thm:main-PRUs} in two steps: (1) we show that $\mathcal{E}$ and $\mathcal{E}^*$ cannot be distinguished from one another, and (2) we show that $\mathcal{E}^*$ and the $n$-qubit Haar ensemble also cannot be distinguished.
The proof of Theorem~\ref{thm:main-PRUs} is complete after combining claims (1) and (2).

Intuitively, claim (1) follows because each small pseudorandom unitary in  $\mathcal{E}$ cannot be distinguished from a small Haar-random unitary by any $\poly(n)$-time adversary.
To establish the claim precisely, we show that the security of the individual small PRUs also implies the security of the collection of small PRUs.
We do so via a hybrid argument, where we use the fact that any quantum algorithm involving $k$ queries to a set of Haar-random unitaries can be efficiently simulated using unitary $k$-designs, despite the fact that Haar-random unitaries have exponential circuit complexity.

To establish claim (2), we apply Theorem~\ref{thm:main-design}  to compare $\mathcal{E}^*$ and the $n$-qubit Haar ensemble.
Note that each small Haar-random unitary in $\mathcal{E}^*$ is trivially a $0$-approximate unitary $\infty$-design.
Therefore, by Theorem~\ref{thm:main-design},  $\mathcal{E}^*$ is an $\varepsilon$-approximate unitary $k$-design for any $\varepsilon, k$ such that $\log_2(nk^2/\varepsilon) \leq \xi$.
To complete the proof, we note that any $\varepsilon$-approximate unitary $k$-design is indistinguishable from a Haar-random unitary by $\poly(n)$-time adversaries, as long as $\varepsilon = 1 / \exp(\omega(\log n))$ is smaller than any inverse-polynomial function in $n$, and $k = \exp(\omega(\log n))$ is larger than any polynomial function in $n$.
We achieve such $\varepsilon,k$ whenever $\xi = \omega(\log n)$.
For such $\xi$, no $\poly(n)$-time quantum algorithm can distinguish between $\mathcal{E}^*$ and the $n$-qubit Haar ensemble, establishing claim (2).
We emphasize that this claim requires taking $k$ to be superpolynomial in $n$, and hence the $\log k$ dependence in Lemma~\ref{lemma: AB BC to ABC}, achieved using Lemma~\ref{lemma: relative to additive}, instead of a $\mathrm{poly}(k)$ dependence is crucial to establish Theorem~\ref{thm:main-PRUs}.

\section{Discussion}

We have shown that random unitaries can be naturally generated in extremely low circuit depths.
Our results reveal a surprising and profound property of quantum circuits, which differs fundamentally from  classical systems.
Moreover, our construction of random unitaries is both exceptionally simple and highly versatile, from either an experimental or theoretical perspective.
We showcase these features by applying our results to a handful of prominent applications across quantum science, including: rigorous guarantees for classical shadows with log-depth Clifford circuits, strict lower bounds on the complexity of detecting many-body topological phases, quantum advantages for learning low-complexity physical systems, improved bounds on the output distributions of  quantum supremacy circuits, and an exponential separation between quantum experiments that have access to forward time-evolution, $U$, and those that have access to both $U$ and its time-reversal, $U^\dagger$.

Our results open up numerous avenues for future work. 
Most importantly, we expect that the list of applications explored in our work is far from exhaustive.
Random unitaries are a ubiquitous tool in quantum information theory and in understanding complex quantum processes.
In quantum benchmarking, efficient learning using random unitaries extends to fermionic, bosonic, and Hamiltonian systems~\cite{zhao2021fermionic, wan2023matchgate, gandhari2024precision, becker2024classical,choi2023preparing,tran2023measuring,shaw2024benchmarking}. 
Can one show that the formation of random unitaries in extremely short times applies to these systems as well?
In quantum gravity, a widespread conjecture states that black holes are the fastest scramblers in nature~\cite{sekino2008fast}.
If we consider scrambling to be the formation of random unitary designs, could the surprisingly fast formation of designs on any geometry found in our work provide new insight into quantum black holes and the AdS/CFT correspondence \cite{maldacena1999large, klebanov1999ads, bouland2019computational}?

On the mathematical side, an obvious open question concerns the optimality of our unitary design construction. 
While we have proven in Proposition~\ref{prop: lower bound design} that our $n$-dependence is optimal, and we inherit the optimal $k$-dependence of Ref.~\cite{chen2024incompressibility} up to polylogarithmic factors, the relation between these two dependencies is not known.
More precisely, we cannot yet rule out the possibility that $\varepsilon$-approximate unitary $k$-designs over $n$ qubits can be created in depth $\mathcal{O}(k) + \mathcal{O}(\log(n / \varepsilon))$, whereas our construction requires a depth of $\tilde{\mathcal{O}}(k) \times \mathcal{O}(\log(n / \varepsilon))$.
Achieving a matching upper and lower bound on the circuit depth with respect to all three parameters $n, k, \varepsilon$ remains an outstanding challenge.

Lastly, it would be interesting to determine whether the same design depth applies to brickwork local random circuits with independent and identically-distributed gates.
When we insert local random circuits into each small random unitary in our construction (as in Corollary~\ref{cor: upper bound design}), we obtain a local random circuit on a brickwork architecture, but with a small subset of the gates set to the identity.
Intuitively, it seems unlikely that drawing these gates from the Haar measure on $\mathrm{SU}(4)$ would increase the depth at which random quantum circuits converge to designs, but we cannot rule this out. 
In fact, this suggests a more general open question: Is it possible for the deletion of random quantum gates to speed up the formation of unitary designs?
This question was also raised in Ref.~\cite{harrow2023approximate} relating to  so-called \emph{censoring} inequalities in Markov chains.

\vspace{0.5em}
\subsection*{Acknowledgments:}

We are grateful to Eric Anschuetz, Ryan Babbush, Christian Bertoni, Adam Bouland, Sergio Boixo, Fernando Brand\~{a}o, Xie Chen, Soonwon Choi, Jordan Cotler, Jens Eisert, Bill Fefferman, David Gosset, Soumik Goush, Patrick Hayden, Nicholas Hunter-Jones, Marios Ioannou, Matteo Ippoliti, Vedika Khemani, Isaac Kim, William Kretschmer, Dominik Kufel, Daniel Liang, Fermi Ma, Jarrod R. McClean, Tony Metger, Ramis Movassagh, Quynh Nguyen, Mehdi Soleimanifar, Nathanan Tantivasadakarn, Umesh Vazirani, and Norman Yao for valuable discussions and insights.
We would like to thank Soonwon Choi for introducing TS and JH to each other, which was vital for this work.
Part of this work was conducted while JH and HH were visiting the Simons Institute for the Theory of Computing.
TS acknowledges support from the Walter Burke Institute for Theoretical Physics at Caltech.
JH acknowledges funding from the Harvard Quantum Initiative.
HH acknowledges the visiting associate position at the Massachusetts Institute of Technology.
The Institute for Quantum Information and Matter, with which TS and HH are affiliated, is an NSF Physics Frontiers Center.
This work was conducted while JH and HH were at the Simons Institute for the Theory of Computing, supported by DOE QSA grant FP00010905.

\vspace{0.6em}
\noindent \emph{Note added:}
We extend our gratitude to Nicholas LaRacuente and Felix Leditzky for bringing their independent concurrent work on unitary designs \cite{shallow-unitary-designs2024} to our attention. Their approach, motivated by quantum communication, demonstrates that approximate $k$-designs can be created across two parties by exchanging only $\mathcal{O}(k \log k)$ qubits, independent of system size. Building on this, they showed that $n$-qubit $k$-designs can be generated in $\tilde{\mathcal{O}}(k^2) \log(n)$ depth.
Our work employs different constructions and analyses, yielding additional findings on low-depth pseudorandom unitaries that require an improved $\mathcal{O}(\log k)$ qubit scaling. We also provide circuit depth lower bounds and explore applications in classical shadows, quantum advantages, and many-body physics.

\vspace{2.5em}
\appendix

\addtocontents{toc}{\protect\setcounter{tocdepth}{2}}

\noindent 
\textbf{\LARGE{}Appendices}
\vspace{2em}

\noindent We provide a road map to help the readers navigate through our appendices. Appendix \ref{app:literature-review} presents a brief literature review of existing results relevant to this work.
Appendix \ref{app: unitary designs meta} establishes our main results regarding random unitary designs.
Appendix~\ref{app: PRUs} establishes our results regarding pseudorandom unitaries.
Appendix~\ref{app: any geometry} shows how to generalize our results from 1D circuits to any circuit geometry. 
Finally, Appendix~\ref{app: applications} presents the technical details of the applications considered in the main text.

\tableofcontents

\section{Literature review}
\label{app:literature-review}

\subsection{Unitary designs}
\label{sec:review-designs}

Unitary designs~\cite{emerson2003pseudo,gross2007evenly,dankert2005efficient,dankert2009exact} are a ubiquitous tool throughout quantum information.
Early work~\cite{harrow2009random} established that local random quantum circuits can form approximate unitary $2$-designs in depth $\mathcal{O}(n+\log(1/\varepsilon))$. 
This was shown by proving the existence of a gap in the spectrum of the moment operators, $\E[ U^{\otimes k}\otimes \overline{U}^{\otimes k} ]$, for $k=2$.
Later work~\cite{brown2010convergence} provided evidence for such a gap at higher $k$ using a mean-field analysis, in which the moment operators can be interpreted as frustration-free Hamiltonians.
Using techniques to bound spectral gaps of quantum-many body systems, Brand\~{a}o, Harrow and Horodecki~\cite{brandao2016local} proved that local random quantum circuits indeed generate approximate unitary $k$-designs in depth $\mathcal{O}(k^{9.5}(nk+\log(1/\varepsilon))$ for any $k$.

Since these seminal early works, a tremendous effort has been devoted to improving the $k$ dependence of the circuit depth.
To begin, Ref.~\cite{haferkamp2022random} showed that the bound in Ref.~\cite{brandao2016local} can be improved to $\mathcal{O}(k^{4+o(1)}(nk+\log(1/\varepsilon))$, where $o(1)$ goes to zero as $k\to \infty$. 
More recently, Ref.~\cite{haah2024efficient} constructed approximate unitary designs with relative error using circuits of depth $\mathcal{O}(\log(n)(nk^2+\log(1/\varepsilon))$ which grows only quadratically in $k$.
Instead of local random circuits, Ref.~\cite{haah2024efficient} considered the use of random Pauli rotations, $e^{\mathrm{i}\theta P}$; these rotations can then be decomposed into local circuits if desired.
Even more recently, Refs.~\cite{chen2024efficient,metger2024simple} constructed approximate unitary designs with depth \emph{linear} in $k$; however, these constructions only achieved additive error and not relative error.
(We recall that relative errors are necessary to match Haar-random behavior in any quantum experiment that adaptively queries the random unitary $k$ times; see Appendix~\ref{app: operational meaning}.)
A wide array of works suggested that approximate unitary designs with relative errors can be generated in depth $\mathcal{O}(nk)$~\cite{hunter2019unitary,brown2010convergence,jian2023linear,Nakata16}.
Finally, Ref.~\cite{chen2024incompressibility} resolved this conjecture up to poly-logarithmic factors by showing that local random quantum circuits generate $\varepsilon$-approximate unitary designs in depth $\mathcal{O}(\mathrm{poly}\log(k)(nk+\log(1/\varepsilon))$ for $k = \mathcal{O}(2^{2n/5})$.

In every case above, the circuit depth of the unitary ensemble grows linearly in $n$.
For additive error designs, Ref.~\cite{harrow2023approximate} showed that the $n$-dependence can be improved to $\mathcal{O} (n^{1/\mathcal{D}})$, if one considers local random quantum circuits on $\mathcal{D}$-dimensional lattices.
Another work that focuses on the $n$-dependence is Ref.~\cite{haferkamp2020homeopathy}, which uses circuits of linear depth  but only an $n$-independent number of non-Clifford gates, $\mathcal{O}(k^4+k\log(1/\varepsilon))$, again for additive error designs.
Finally, we recall that for the specific case of Clifford unitaries and all-to-all-connected circuit architectures, exact unitary 2- and 3-designs can be generated in $\log n$ depth using ancilla qubits~\cite{moore2001parallel,jiang2020optimal}.

Our work improves the $n$ dependence exponentially compared to every prior work above.
Fundamentally, this is because our work is the first  that does not require  circuits with extensive light-cones.

\subsection{Pseudorandom states and unitaries}
Pseudorandom states and unitaries were introduced in Ref.~\cite{PRS2018}.
Ref.~\cite{PRS2018} also provides the first examples of pseudorandom states, which assume only the existence of quantum-secure one-way functions, and provides a potential construction of pseudorandom unitaries.
By assuming the quantum hardness of the Learning With Error problems~\cite{regev2009lattices}, one can establish the existence of quantum-secure one-way functions.
These quantum-secure one-way functions can then be used to construction quantum-secure pseudorandom functions~\cite{zhandry2021PRF} and quantum-secure pseudorandom permutations~\cite{zhandry2016note}, which are important building blocks for most proposals for pseudorandom states and unitaries.
Pseudorandom states and unitaries have since become a crucial concept in quantum learning theory~\cite{huang2021quantum, zhao2023learning}, cryptography~\cite{ananth2022cryptography,morimae2022quantum} and the AdS/CFT correspondence~\cite{bouland2019computational}.

Towards constructing pseudorandom unitaries, Ref.~\cite{lu2023quantum} proved the existence of \emph{quantum pseudorandom scramblers}, i.e. an ensemble of unitaries which, applied to any input state, generates an ensemble of pseudorandom states.
Refs.~\cite{chen2024efficient,metger2024simple} constructed ensembles of unitaries that are pseudorandom with non-adaptive security. 
That is, such unitary ensembles are indistinguishable from the Haar ensemble by any non-adaptive polynomial-time quantum algorithm that can query $\mathrm{poly}(n)$ copies of $U$ in parallel, but only in a single round.
Pseudorandom unitaries with non-adaptive security can be seen as a stronger version of quantum pseudorandom scramblers, but they still do not resolve the conjecture regarding the existence of pseudorandom unitaries.
Similar to pseudorandom states \cite{ji2018pseudorandom}, these works \cite{lu2023quantum,chen2024efficient,metger2024simple} rely on the cryptographic assumption that quantum-secure one-way functions exist in order to instantiate quantum-secure pseudorandom functions and permutations.

In~\cite{PRU2024}, the authors resolve the conjecture regarding the existence of pseudorandom unitaries by proving that the $PFC$ construction proposed in Ref.~\cite{metger2024simple} forms a pseudorandom unitary.
Furthermore, Ref.~\cite{PRU2024} gives new constructions for \emph{strong} pseudorandom unitaries that are secure against any efficient quantum algorithms that can query both the unitary $U$, its time-reversal $U^\dagger$, and the controlled versions of both $U$ and $U^\dagger$.
As shown in Section~\ref{sec: applications}, obtaining time-reversal in a quantum experiment is powerful. 
Nonetheless, Ref.~\cite{PRU2024} provides an efficient construction of pseudorandom unitaries secure against access to both time-reversal and controlled operations.
In an upcoming work~\cite{SPRU2024}, the authors prove that strong PRUs can be generated in $\log n$ depth on all-to-all-connected circuits.

All known proposals for constructing pseudorandom states or unitaries require the implementation of pseudorandom functions or permutations.
When placed in a 1D circuit layout, these require $\poly(n)$ depth for $n$-qubit systems.
Our work provides the first construction of both pseudorandom states and pseudorandom unitaries that only requires $\poly \log n$ depth in 1D, which improves exponentially over all known constructions.
For all-to-all-connected circuits, there are pseudorandom states that can be constructed in $\poly \log n$ depth~\cite{SPRU2024}. 
For such circuits, our construction only requires circuit depth $\poly \log \log n$, again an exponential improvement over all known constructions of pseudorandom states and unitaries.

By building on the concept of pseudorandom states, researchers have also constructed \emph{pseudoentangled states}, which are states that have low entanglement, but are computationally indistinguishable from states with high entanglement~\cite{aaronson2022quantum}.
Because geometrically-local shallow quantum circuits only alter the entanglement structures of the input states up to area-law fluctuations, applying our low-depth PRUs to an initial state $\ket{\psi}$ creates pseudorandom states with almost the same entanglement structure as $\ket{\psi}$.
As a special case, when $\ket{\psi}$ is any product state, applying our low-depth PRUs to $\ket{\psi}$ generates pseudoentangled states~\cite{aaronson2022quantum}.
Adam Bouland has suggested to us the notion of \emph{pseudoentangling unitaries}, which are unitaries that can only create short-range entanglement, but which are indistinguishable from unitaries that create volume-law entanglement.
The low-depth pseudorandom unitaries we constructed are the first known construction of pseudoentangling unitaries.

\subsection{Classical shadows with low-depth random circuits}

Classical shadow tomography~\cite{huang2020predicting,huang2021efficient,huang2021information,aaronson2018shadow,buadescu2020improved,chen2020robust,hu2023classical,koh2022classical} for estimating non-local observables such as fidelities requires the preparation of random Clifford unitaries.
When restricted to geometrically-local circuit architectures, these random Clifford unitaries require circuits of linear depth, which limits their realization in current experiments due to the presence of noise in existing experimental platforms.
Recently, multiple works~\cite{akhtar2023scalable,bertoni2022shallow,hu2023classical,hu2024demonstration,ippoliti2023operator,farias2024robust} have considered classical shadow estimation using shallow random Clifford circuits instead of linear-depth random Clifford unitaries. These proposals are known as ``shallow shadows''.

It was suggested in Ref.~\cite{bertoni2022shallow} that the depth $d=\Theta(\log(n))$ achieves ``an ideal middle ground'' between random product measurements (depth $0$) and uniformly random Clifford unitaries (linear depth).
In particular, Ref.~\cite{bertoni2022shallow} shows that random log-depth Clifford circuits reproduce the same guarantees as global Cliffords for most states drawn from a 1-design.
They moreover conjecture that this scaling holds for all states.
Ref.~\cite{ippoliti2023operator} proposes a physical explanation for this improvement using ideas from quantum many-body dynamics.
They show that the capabilities of shallow shadows can be understood in terms of a competition between two physical processes, operator spreading and operator relaxation.
Using these ideas, Ref.~\cite{ippoliti2023operator} proved that for Frobenius-normalized Pauli observables $P / \sqrt{2^n}$, the squared shadow norm is bounded above by $1.14^n$ at any circuit depth, with an optimal depth given by $\log(n)$.
They conjecture that this upper bound can be further tightened to a constant at $\log(n)$ circuit depth.
The conjecture is supported by numerical evidence and analytic calculations under a physically-motivated mean-field assumption, which assumes that Pauli densities on different qubits are independently identically distributed (i.i.d.) binomial random variables.
However, the conjecture remains unproven due to the difficulty in formally justifying the mean-field approximation.
In our work, we show that any Frobenius-norm-bounded observable, including $P / \sqrt{2^n}$, has a squared shadow norm of a constant at $\log(n)$ circuit depth. This resolves the conjectures posed in Refs.~\cite{bertoni2022shallow,ippoliti2023operator}.

\subsection{Recognizing topological order}

There exist innumerably many works across condensed matter and many-body physics on recognizing the topological order of a quantum state~\cite{clark2020observation,semeghini2021probing, satzinger2021realizing,leonard2023realization,iqbal2024topological,iqbal2024non,jiang2012identifying,cong2019quantum,rodriguez2019identifying,cian2021many,herrmann2022realizing,huang2022provably,cian2022extracting,lake2022exact,bouland2023public,cong2024enhancing}.
In nearly all cases, the proposed approaches are heuristic and are often motivated by specific details about the quantum states that one is seeking to characterize.
For example, several prominent approaches include measuring the topological entanglement entropy~\cite{jiang2012identifying}, measuring or re-constructing the expectation values of Wilson loop operators~\cite{semeghini2021probing,satzinger2021realizing,cian2022extracting,iqbal2024topological,iqbal2024non}, and methods inspired by quantum error correction and the renormalization group~\cite{cong2019quantum,cong2024enhancing,lake2022exact}.
These approaches have all been shown to succeed in small-scale numerical studies, typically involving low-entangled quantum states that are also nearby so-called ``fixed points'' of the topological orders of interest.

Our work shows that for even moderately more complex quantum states, all of these approaches must incur a super-polynomial overhead.
In the case of the topological entanglement entropy, this is easily understood because measuring entropies in general requires an exponential sample complexity.
In the case of Wilson loop and renormalization group approaches, the difficulty stems from the lack of a simple nearby fixed-point on which to begin ones' analysis.
Looking forward, several important open questions remain.
Do our hardness lower bounds extend to recognizing topological phases of matter in the presence of a protecting symmetry group~\cite{zeng2019quantum,lake2022exact}?
And can we extend our results even closer to the topologically-ordered states that might arise in the physical world---e.g.~what is the complexity of recognizing topological order in the ground states of 2-local Hamiltonians?

\subsection{Quantum advantage in learning from experiments}

Understanding how to efficiently learn from experiments is a fundamental problem in physics and has attracted significant attention recently in the quantum world.
In particular, recent works have built on quantum information theory and learning theory \cite{mohri2018foundations} to provide rigorous theoretical foundation to understand the power and limitations in learning from quantum experiments using both classical and quantum algorithms \cite{huang2021information, huang2021quantum, huang2022foundations, huang2022provably, huang2021quantum, aharonov2021quantum, chen2021exponential, huang2020power}.
These efforts can be separated into two lines of work: (1) understanding the power of classical algorithms that can learn from data to solve challenging problems, and (2) proving separations between classical and quantum learning algorithms to yield rigorous quantum advantages in learning from experiments.

Ref.~\cite{huang2020power} showed that a problem can be easy to solve for a classical machine learning (ML) algorithm that can learn from quantum data even when a problem is classically-hard to solve. At a formal level, this claim is shown using the fact that training data can be viewed as a restricted form of advice strings in computational complexity theory \cite{arora2009computational}, which can enhance algorithms that utilize them.
Building on the computational power of data,  Ref.~\cite{huang2022provably} gave the first efficient classical ML algorithm to predict ground state properties from a new unseen Hamiltonian by training on a dataset collected from quantum experiments.
The same problem is computationally hard for classical algorithms that do not learn from data.
The sample complexity of this algorithm is $\mathcal{O}(n^{1/\epsilon})$, where $\epsilon$ denotes the prediction error.
This was subsequently improved to $\log(n) e^{\mathrm{poly}\log(1/\epsilon)}$ in Ref.~\cite{lewis2024improved, onorati2023efficient} and then to $e^{\mathrm{poly}\log(1/\epsilon)}$ in Ref.~\cite{wanner2024predicting}.
When one can measure ground states of gapped quantum many-body systems, Refs.~\cite{rouze2024learning,yu2023learning} show that one can predict a large number of ground state properties using classical data collected from $\mathrm{poly}\log n$ local measurements.
In the special case of learning states generated from shallow quantum circuits, a computationally-efficient classical algorithm to do so was developed in Ref.~\cite{huang2024learning}.
In addition to predicting ground states of quantum systems, Ref.~\cite{huang2023learning} also showed that classical ML algorithms can efficiently predict many output state properties of arbitrary quantum processes, even when the quantum process has exponentially high complexity.

While many recent works have established the power of classical algorithms for learning from quantum experiments, a large body of work suggest that the full suite of quantum technologies can provide a significant advantage over classical methods.
Refs.~\cite{huang2021quantum,chen2021exponential,chen2021hierarchy,aharonov2021quantum} provide the first set of learning problems that are are easy to solve with quantum learning agents that can receive quantum information using quantum sensors, process quantum information using quantum computation, and store quantum information in quantum memory, but are provably hard to solve with classical learning agents that can receive classical information using measurements, process classical information using classical computation, and store classical information in classical memory.
Refs.~\cite{huang2021quantum,chen2021exponential} establish an exponential separation between quantum and classical learning agents for predicting highly-incompatible observables \cite{aaronson2018shadow}, performing quantum principal component analysis \cite{lloyd2014quantum}, distinguishing between the completely depolarizing channel and Haar-random unitaries \cite{chen2021exponential}, and learning polynomial-time quantum processes to within small average gate fidelity \cite{huang2021quantum}.
In Ref.~\cite{chen2021hierarchy}, the authors construct a distinguishing task that is easy with access to entangled measurements over many copies of an unknown state, but hard with entangled measurements over at most $k$ copies.
Refs.~\cite{huang2022foundations, chen2023complexity} showed that some of these exponential separations become very large polynomial separations when the quantum learning agent is noisy.
Finally, Refs.~\cite{chen2021quantum, chen2024tight, oh2024entanglement,chen2023futility} extended such quantum advantages to learning Pauli channels in many-qubit systems and learning displacement channels in bosonic systems.

In every existing work above, the quantum advantage is established by considering physical systems with highly non-local correlations across all $n$ qubits of the system.
This is not a common condition in many physical settings, where correlation lengths are much smaller in $n$.
Furthermore, in many of the learning separations above, such as distinguishing between a completely depolarizing channel and a Haar-random unitary, existing proofs require the physical system to have an exponential circuit complexity.
This is also highly nonphysical.
In our work, we present the first superpolynomial separation between quantum and classical agents for learning geometrically-local physical systems with correlation length $\poly \log n$ and circuit complexity (i.e.~depth) $\poly \log n$.

\subsection{The power of time-reversal in learning}

The idea that time-reversal experiments can reveal fundamentally new aspects of many-body systems dates back to the foundations of thermodynamics~\cite{loschmidt1876volume}.
The formal study of this advantage in \emph{quantum} systems began in Ref.~\cite{cotler2023information}, which introduced a simple learning task that featured an exponential advantage in sample complexity for experiments with time-reversal over those without.
However, this learning task had one significant limitation: It could also be efficiently solved without time-reversal when the experiment has access to a quantum memory.
Our work significantly strengthens this separation by introducing a learning task that cannot be solved by \emph{any} quantum experiment that queries $U$, including those with arbitrarily large quantum memories.
Our task builds upon the physical intuition introduced in Ref.~\cite{schuster2023learning}, which showed that time-reversal experiments are particularly advantageous for learning properties such as the connectivity of an unknown Hamiltonian.

\section{Random unitary designs}
\label{app: unitary designs meta}

\subsection{A brief review of approximate unitary designs} \label{app: approximate}

Let $\mathcal{E}$ denote an ensemble of $n$-qubit unitaries, and $\mathcal{E}_{H}$ denote the Haar-random ensemble, which samples $U$ uniformly over all $n$-qubit unitaries.

\subsubsection{Definitions}

A random unitary ensemble $\mathcal{E}$ forms an approximate unitary $k$-design if it approximately matches Haar-random unitaries up to the $k$-th moments.
The gold standard for quantifying the approximate error in the random unitary literature~\cite{brandao2016local} is given by the following.

\begin{definition}[Approximate unitary design] \label{def: strong unitary design}
An ensemble $\mathcal{E}$ is an $\varepsilon$-approximate unitary $k$-design if
\begin{align}
    (1-\varepsilon) \, \Phi_H \, \preceq \, \Phi_{\mathcal{E}} \, \preceq \, (1+\varepsilon) \, \Phi_H,
\end{align}
where the quantum channel $\Phi_{\mathcal{E}}(\cdot)$ is defined via
\begin{equation} \label{eq:twirling-channel}
	\Phi_{\mathcal{E}}(A) \coloneqq \E_{U \sim \mathcal{E}} \left[ U^{\otimes k} A U^{\dagger, \otimes k} \right],
\end{equation}
and similarly for the Haar ensemble. Here, $\Phi \preceq \Phi'$ denotes that $\Phi'-\Phi$ is a completely-positive map.
\end{definition}

The error $\varepsilon$ in the above definition is commonly referred to as the relative error.
A small relative error has an important operational meaning, in that it implies indistinguishability from a Haar-random unitary under any $k$ queries to the random unitary $U$, which we prove in Lemma~\ref{lem:dist-relative}.
There also exists a weaker notion of approximation error, which we refer to as the additive error.
Let $\norm{\cdot}_\diamond$ denote the diamond norm. The additive and relative errors are given as follows,
\begin{align}
    &\text{additive error $\varepsilon$:} & &
    \quad \norm{\Phi_{\mathcal{E}} - \Phi_H}_\diamond \leq \varepsilon, \\
    &\text{relative error $\varepsilon$:}
    &-\varepsilon \, \Phi_H \, & \preceq \, \Phi_{\mathcal{E}} - \Phi_H \, \preceq \, \varepsilon \, \Phi_H.
\end{align}
A small relative error always implies a small additive error. However, the converse is not always true and even a superpolynomially small additive error could have a large relative error.
In general, one needs to have an exponentially small additive error to guarantee a small relative error.

\begin{lemma}[Additive vs relative error; Lemma 3 in Ref.~\cite{brandao2016local}]
Any unitary $k$-design with relative error $\varepsilon$ is a unitary $k$-design with an additive error $2 \varepsilon$.
Conversely, any unitary $k$-design with additive error $\varepsilon$ is a unitary $k$-design with a relative error $2^{2n k} \varepsilon$.
\end{lemma}
\noindent We note that our work improves the latter conversion by a factor of $k!$ through Lemma~\ref{lemma: relative to additive}.

\subsubsection{Operational meaning}\label{app: operational meaning}

As aforementioned, there is a simple operational interpretation of the definition of approximate unitary $k$-designs. 
Namely, for any quantum algorithm that applies $U$ $k$ times, the output states when $U$ is sampled from $\mathcal{E}$ and when $U$ is sampled from the Haar-random ensemble $\mathcal{E}_{H}$ are close. This is captured by the following two lemmas.
The first lemma states that a small additive error is equivalent to indistinguishability under $k$ nonadaptive/parallel queries to $U$.

\begin{lemma}[Indistinguishability under additive error] \label{lem:dist-additive}
    An approximate unitary $k$-design $\mathcal{E}$ has an additive error $\varepsilon$ if and only if for all quantum algorithms that make only one query to $U^{\otimes k}$, the output state when $U$ is sampled from $\mathcal{E}$ or from $\mathcal{E}_{H}$ is at most $\varepsilon$-far in $\norm{\cdot}_1$.
\end{lemma}

The second lemma states that a small relative error implies indistinguishability under \emph{any} $k$ queries to $U$.
These $k$ queries can depend adaptively and quantumly on all previous queries and are not restricted to querying $U$ in parallel.
Hence, a small relative error is operationally much stronger than a small additive error.

\begin{lemma}[Indistinguishability under relative error] \label{lem:dist-relative}
    An approximate unitary $k$-design $\mathcal{E}$ with a relative error $\varepsilon$ implies that for all quantum algorithms that make $k$ queries to $U$, the output state when $U$ is sampled from $\mathcal{E}$ or from $\mathcal{E}_{H}$ is at most $2 \varepsilon$-far in $\norm{\cdot}_1$.
\end{lemma}

\begin{proof}[Proof of Lemma~\ref{lem:dist-additive}]
The claim regarding additive error follows from the definition of diamond norm,
\begin{equation}
\norm{\Phi_{\mathcal{E}} - \Phi_H}_\diamond = \max_{\rho} \norm{\E_{U \sim \mathcal{E}} \left[ (U^{\otimes k} \otimes \mathbbm{1}) \rho (U^{\dagger, \otimes k} \otimes \mathbbm{1}) \right] - \E_{U \sim \mathcal{E}_{H}} \left[ (U^{\otimes k}  \otimes \mathbbm{1}) \rho (U^{\dagger, \otimes k} \otimes \mathbbm{1}) \right] }_{1},
\end{equation}
where the maximization is over all $(nk+m)$-qubit pure states $\rho$, with $m$ being an arbitrarily large integer and the identity $\mathbbm{1}$ is over $m$ qubits.
Note that the definition of diamond distance is the same when we restrict $\rho$ to be pure or when we consider any density matrix $\rho$.
For any quantum algorithm that makes only one query to $U^{\otimes k}$, we can write the output state as,
\begin{equation}
    W_2 \cdot (U^{\otimes k} \otimes \mathbbm{1}) \cdot \undersetbrace{\text{denote as $\rho$}}{W_1 \ketbra*{0^{nk + m}}{0^{nk + m}} W_1^\dagger} \cdot (U^{\dagger, \otimes k} \otimes \mathbbm{1}) \cdot W_2^\dagger,
\end{equation}
for any $(nk + m)$-qubit unitaries, $W_1$ and $W_2$.
Because $\norm{W M W^\dagger}_1 = \norm{M}_1$ for any unitary $W$ and any Hermitian matrix $M$, we obtain the equivalence between additive error and indistinguishability under parallel/nonadaptive queries to $U$.
\end{proof}

\begin{proof}[Proof of Lemma~\ref{lem:dist-relative}]
The claim regarding relative error can be shown as follows; see also Fig.~\ref{fig:parallel} for a visual depiction.
We denote the $n$-qubit system register that a single copy of $U$ acts on (see Fig.~\ref{fig: setup} of the main text) to be $A$.
We denote $m$ as the number of ancilla qubits used by the algorithm, and denote the ancillary space as $A'$.
The (pure) output state $\ket{\phi_U}$ of any quantum algorithm that makes $k$ queries to $U$ can be written as follows,
\begin{equation}
    \ket{\psi_U} \coloneqq W_{k+1} \cdot \prod_{i=1}^k ( (U \otimes \mathbbm{1}_{A'}) \cdot W_i) \ket{0^{n+m}},
\end{equation}
where $W_1, \ldots, W_{k+1}$ are $(n+m)$-qubit unitaries acting on both $A$ and $A'$, and $\mathbbm{1}_{A'}$ is the identity acting on the ancilla register $A'$.
We can write $U = \sum_{x, y \in \{0, 1\}^n} \bra{y} U \ket{x} \ketbra{y}{x}$ to obtain
\begin{equation}
    \ket{\psi_U} = \sum_{\substack{x_1, \ldots, x_k \in \{0, 1\}^n \\ y_1, \ldots, y_k \in \{0, 1\}^n}}  W_{k+1} \cdot \prod_{i=1}^k ( (\ketbra{y_i}{x_i} \otimes \mathbbm{1}_{A'}) \cdot W_i) \ket{0^{n+m}} \cdot \bra{y_1, \ldots, y_k} U^{\otimes k} \ket{x_1, \ldots, x_k}.
\end{equation}
This expanded expression allows us to rewrite $\ket*{\psi_U}$ as the partial inner product between two unnormalized pure states.

\begin{figure}[t]
    \centering
    \includegraphics[width=0.84\textwidth]{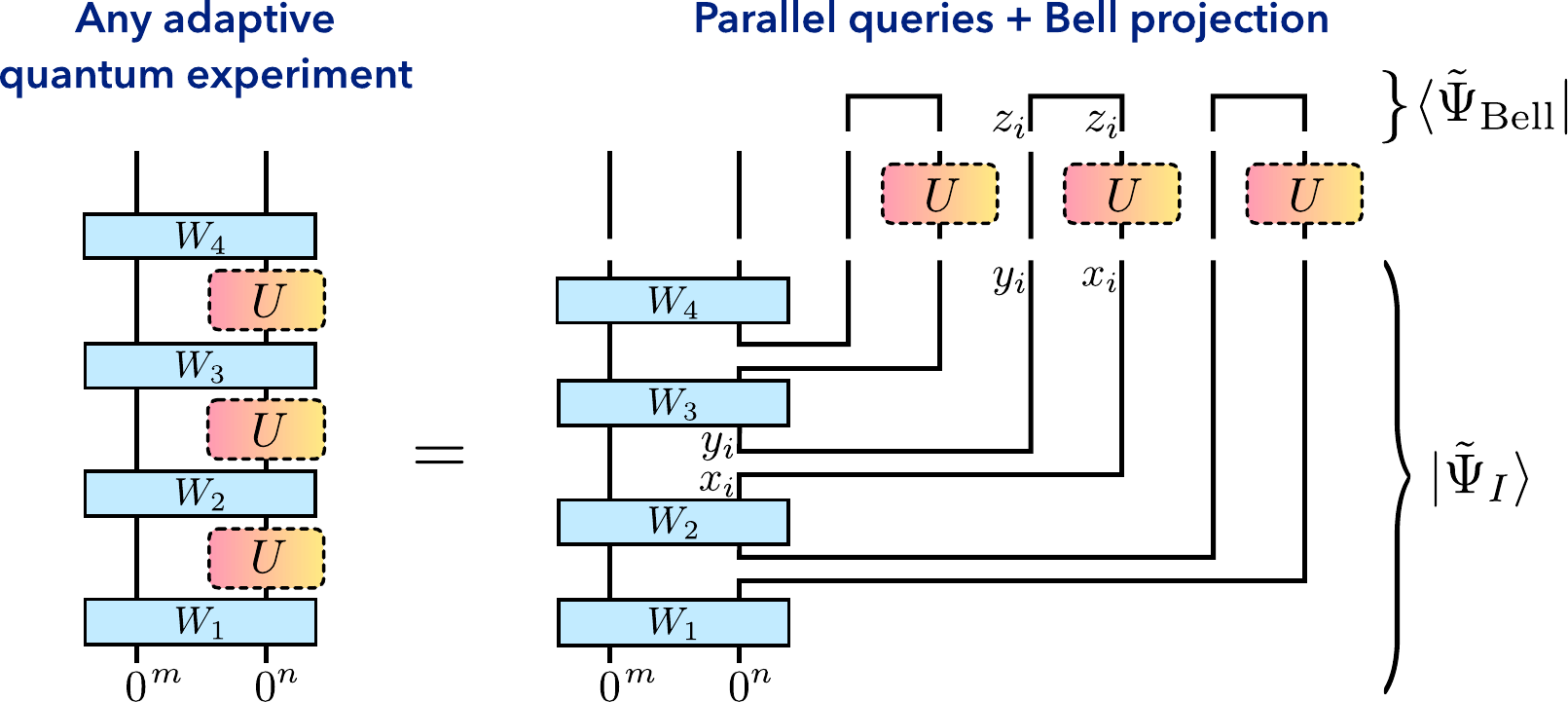}
    \caption{
    Visualization of the proof of Lemma~\ref{lem:dist-relative}, which relates the output state of any adaptive experiment involving $k$ queries to a unitary $U$ interleaved with arbitrary quantum circuits $W_1,\ldots, W_{k+1}$, to the output of an experiment with only a single parallel query to $U^{\otimes k}$ (shown for $k=3$).
    The latter experiment applies $W_1,\ldots, W_{k+1}$ to a series of Bell states (corresponding to the sums over $y_i$ in the text; shown for $i=2$) and swaps the output of each $W_i$ to an additional ancilla register (corresponding to the sums over $x_i$).
    After querying $U^{\otimes k}$, the additional ancilla registers are projected onto Bell states as shown (corresponding to the sums over $z_i$).
    By viewing each Bell state and Bell projector using tensor network notation, where each line denotes the entire Hilbert space of either a system or ancilla register, it becomes clear that the two experiments are equivalent.
    The normalization of the input state, $| \tilde{\Psi}_I \rangle$, and the Bell projector, $| \tilde{\Psi}_{\text{Bell}} \rangle$, are each $2^{nk}$.
    Hence, achieving a relative error in $U^{\otimes k}$ is crucial for capturing adaptive quantum experiments.
    }
    \label{fig:parallel}
\end{figure}

We define the following two unnormalized pure states,
\begin{align}
    \ket*{\widetilde{\Psi}_U} &\coloneqq \sum_{\substack{x_1, \ldots, x_k \in \{0, 1\}^n \\ y_1, \ldots, y_k \in \{0, 1\}^n}}  W_{k+1} \cdot \prod_{i=1}^k ( (\ketbra{y_i}{x_i} \otimes \mathbbm{1}_{A'}) \cdot W_i) \ket{0^{n+m}} \otimes U^{\otimes k} \ket{x_1, \ldots, x_k} \otimes \ket{y_1, \ldots, y_k},\\
    \ket*{\widetilde{\Psi}_{\mathrm{Bell}}} &\coloneqq \sum_{\substack{z_1, \ldots, z_k \in \{0, 1\}^n}} \ket{z_1, \ldots, z_k} \otimes \ket{z_1, \ldots, z_k}.
\end{align}
We can isolate the dependence on $U$ in the $(n + m + 2 k n)$-qubit unnormalized state $\ket*{\widetilde{\Psi}_U}$,
\begin{equation}
    \ket*{\widetilde{\Psi}_U} = ( \mathbbm{1}_{A, A'} \otimes (U^{\otimes k} \otimes \mathbbm{1}_{A^{\otimes k}}) ) \ket*{\widetilde{\Psi}_I},
\end{equation}
where $\mathbbm{1}_{A^{\otimes k}}$ acts as identity on $k$ copies of the $n$-qubit system.
The output state $\ket*{\psi_U}$ can now be written as the partial inner product between the two unnormalized states,
\begin{equation}
    \ket*{\psi_U} = (\mathbbm{1}_{A, A'} \otimes \bra*{\widetilde{\Psi}_{\mathrm{Bell}}}) \ket*{\widetilde{\Psi}_U} = (\mathbbm{1}_{A, A'} \otimes \bra*{\widetilde{\Psi}_{\mathrm{Bell}}}) ( \mathbbm{1}_{A, A'} \otimes (U^{\otimes k} \otimes \mathbbm{1}_{A^{\otimes k}}) ) \ket*{\widetilde{\Psi}_I}.
\end{equation}
Recall that $\Phi_{\mathcal{E}}(A) \coloneqq \E_{U \sim \mathcal{E}} \left[ U^{\otimes k} A U^{\dagger, \otimes k} \right]$.
Using the above identity, we can write the output state $\ketbra*{\psi_U}{\psi_U}$ averaged over random unitary $U$ sampled from the ensemble $\mathcal{E}$ as
\begin{equation}
    \E_{U \sim \mathcal{E}} \ketbra*{\psi_U}{\psi_U} = (I_{A, A'} \otimes \bra*{\widetilde{\Psi}_{\mathrm{Bell}}}) \cdot (\mathcal{I}_{A, A'} \otimes \Phi_{\mathcal{E}} \otimes \mathcal{I}_{A^{\otimes k}}) (\ketbra*{\widetilde{\Psi}_I}{\widetilde{\Psi}_I}) \cdot (I_{A, A'} \otimes \ket*{\widetilde{\Psi}_{\mathrm{Bell}}}),
\end{equation}
where $\mathcal{I}_{A, A'}$ and $\mathcal{I}_{A^{\otimes k}}$ are the identity CPTP maps acting on register $A, A'$ and the $k$-copy register $A^{\otimes k}$. Using the definition of relative error $(1-\varepsilon) \, \Phi_H \, \preceq \, \Phi_{\mathcal{E}} \, \preceq \, (1+\varepsilon) \, \Phi_H$, we have
\begin{equation}
    (1 - \varepsilon) \E_{U \sim \mathcal{E}_{H}} \ketbra*{\psi_U}{\psi_U} \preceq \E_{U \sim \mathcal{E}} \ketbra*{\psi_U}{\psi_U} \preceq (1 + \varepsilon) \E_{U \sim \mathcal{E}_{H}} \ketbra*{\psi_U}{\psi_U},
\end{equation}
where for two positive-semidefinite matrices $A, B$, the relation $A \preceq B$ denotes that $B - A$ is positive semidefinite. Therefore, using the definition of $\norm{\cdot}_1$, we have
\begin{align}
\begin{split}
    &\norm{\E_{U \sim \mathcal{E}} \ketbra*{\psi_U}{\psi_U} - \E_{U \sim \mathcal{E}_{H}} \ketbra*{\psi_U}{\psi_U}}_1\\
    &\leq \max_{0 \preceq M_+, M_- \preceq I}
    \Tr\left(M_+ \left[ \E_{U \sim \mathcal{E}} \ketbra*{\psi_U}{\psi_U} - \E_{U \sim \mathcal{E}_{H}} \ketbra*{\psi_U}{\psi_U} \right] \right)\\
    &\quad \quad \quad \quad \quad \,\, - \Tr\left(M_- \left[ \E_{U \sim \mathcal{E}} \ketbra*{\psi_U}{\psi_U} - \E_{U \sim \mathcal{E}_{H}} \ketbra*{\psi_U}{\psi_U} \right] \right)\\
    &\leq \max_{0 \preceq M_+, M_- \preceq I}
    \epsilon \Tr\left(M_+ \E_{U \sim \mathcal{E}_{H}} \ketbra*{\psi_U}{\psi_U} \right) + \epsilon \Tr\left(M_- \E_{U \sim \mathcal{E}_{H}} \ketbra*{\psi_U}{\psi_U} \right)\\
    &= 2 \varepsilon \norm{\E_{U \sim \mathcal{E}_{H}} \ketbra*{\psi_U}{\psi_U}}_1 = 2 \varepsilon.
    \end{split}
\end{align}
This concludes the proof of this lemma.
\end{proof}

\subsection{A brief review of Haar random unitaries} 
\label{app: haar-random-review}

In this appendix, we review several properties of Haar random unitaries.
We let $\mathcal{H}$ be the Hilbert space of an $n$-qubit system, with Hilbert space dimension $D = 2^n$, and $\mathcal{E}_H$ be the Haar-random unitary ensemble on $\mathcal{H}$.
We begin by reviewing the expression for the Haar twirl, $\Phi_H$ [Eq.~\eqref{eq:twirling-channel}], as a sum over permutation operators on $\mathcal{H}^{\otimes k}$.
We then review the representation theory of the unitary and symmetric groups, and write the action of the twirl in the basis of irreducible representations.

\subsubsection{Permutation operators and the Weingarten calculus} 

The $k$-fold Haar-random twirling channel can be expressed as the sum,
\begin{equation} \label{eq: Haar twirl}
\Phi_H(A) \equiv \E_{U \sim \mathcal{E}_H}[ U^{\otimes k} A \, ( U^\dagger)^{\otimes k}] = \sum_{\sigma, \tau \in S_k} \Tr( A \,  \sigma^{-1} ) \, \text{Wg}(\sigma \tau^{-1};D) \, {\tau}.
\end{equation}
Here, $\sigma, \tau \in S_k$ are permutations of $k$ elements.
In a slight abuse of notation, we use the same symbols to denote the permutations' representation on $\mathcal{H}^{\otimes k}$.
The coefficients in the sum, $\text{Wg}(\,\cdot\,;D): S_k \to \mathbb{R}$, are known as the Weingarten function~\cite{kostenberger2021weingarten}.

When $ k \leq D$, the Weingarten function can be constructed as follows.
Consider two permutations $\sigma, \tau \in S_k$ and define the function
\begin{equation}
G(\sigma \tau^{-1};D) \equiv \Tr( \sigma  \tau^{-1}) = D^{\#(\sigma \tau^{-1})} = D^{k - |\sigma \tau^{-1}|}
\end{equation}
where $\#(\sigma \tau^{-1})$ is the number of cycles in $\sigma \tau^{-1}$, and $| \sigma \tau^{-1} | \equiv k - \#(\sigma \tau^{-1})$ defines a distance measure on the permutation group.
The function $G(\cdot; D)$ can be viewed as a $k! \times k!$ symmetric Gram matrix, $\hat{G}(D)$, with matrix elements $G_{\sigma,\tau}(D) = G(\sigma \tau^{-1}; D)$.
We can similarly view the Weingarten function as a symmetric matrix, $\hat{\Wg}(D)$, with elements $\Wg_{\sigma,\tau}(D) = \Wg(\sigma \tau^{-1}; D)$.
For $k \leq D$, the Weingarten matrix is defined as the matrix inverse of $G_{\sigma,\tau}$. That is, we have
\begin{equation} \label{eq:Wg-G-matrix}
\sum_{\tau \in S^k} \Wg_{\sigma,\tau}(D) \, G_{\tau, \pi}(D) = \delta_{\sigma, \pi}\,.
\end{equation}
While $G_{\sigma, \tau}(D) \geq 0$, the Weingarten matrix elements $\Wg_{\sigma,\tau}(D)$ may take negative values.
In particular, we have $\Wg_{\sigma,\tau}(D) = (-1)^{|\sigma|}  |\!\Wg_{\sigma,\tau}(D)|  (-1)^{|\tau|}$.

A key property of the permutation operators is that they are approximately orthogonal when the Hilbert space dimension $D$ is large compared to $k^2$.
This allows one to approximate $G_{\sigma,\tau}(D) \approx D^k \delta_{\sigma,\tau}$ and $\Wg_{\sigma,\tau}(D) \approx D^{-k} \delta_{\sigma,\tau}$ in many settings.
We will not require the full statement of this orthogonality~\cite{harrow2023approximate} in our work.
Instead, we simply note that it manifests in the following bounds for sums over $G(\cdot;D)$ and $\Wg(\cdot;D)$.
For the former, we have~\cite{aharonov2021quantum},
\begin{equation} \label{eq: sum G}
	\sum_\sigma G(\sigma; D) = \frac{(k+D-1)!}{D!} \leq D^k \left( 1 + \frac{k^2}{D} \right).
\end{equation}
where the inequality holds for $k^2 \leq D$.
For the latter, we have~\cite{aharonov2021quantum}
\begin{equation}\label{eq: sum abs Wg}
	\sum_{\sigma \in S_k} |\text{\rm Wg}(\sigma; D)| = \frac{(D-k)!}{D!} \leq \frac{1}{D^k} \left( \frac{1}{1-k^2/2D} \right),
\end{equation}
where the sum is over the absolute value of $\Wg(\cdot;D)$, and the inequality holds for $k^2 \leq D$.
In each case, the leading contribution to the sum arises from the diagonal term, $G(\mathbbm{1};D) = G_{\sigma,\sigma}(D) = D^k$ and $\Wg(\mathbbm{1};D) = \Wg_{\sigma,\sigma}(D) \approx 1/D^k$, owing to the approximate orthogonality of the permutations.

\subsection{Proof of Theorem~\ref{thm:main-design}: Gluing small random unitary designs}
\label{app: unitary designs}

This appendix provides the details for proving Theorem~\ref{thm:main-design}.
We refer the readers to Section~\ref{sec:proof-overview} in the main text for an overview of the proof.

\subsubsection{Approximation for the Haar twirl} \label{app: approx Haar}

We now prove Lemma~\ref{lemma: approx Haar twirl} of the main text, which shows that the exact formula for the Haar twirl [Eq.~(\ref{eq: Haar twirl})] is approximated by the much simpler expression,
\begin{equation}
    \Phi_a( \rho ) = \frac{1}{D^{k}} \sum_\sigma \tr( \rho \sigma^{-1} ) \sigma
\end{equation}
up to a small relative error, $\varepsilon = k^2/2D/(1-k^2/2D)$, where $D=2^n$.
Our proof leverages Lemma~\ref{lemma: relative error to a}, which is introduced and proven in the following section.

Before proceeding to the proof, we remark that one can also invert the relation in Lemma~\ref{lemma: approx Haar twirl}, up to a small increase in approximation error.
From Lemma~\ref{lemma: approx Haar twirl}, we have
\begin{equation} \label{eq: inv approx Haar twirl}
    (1-\varepsilon') \Phi_a \preceq \frac{1}{1+\varepsilon} \Phi_a \preceq \Phi_H \preceq \frac{1}{1-\varepsilon} \Phi_a \preceq (1+\varepsilon') \Phi_a,
\end{equation} 
where $\varepsilon' = \varepsilon/(1-\varepsilon) = k^2/2D/(1-k^2/D)$.
For any $k^2 \leq D/2$, we have $\varepsilon' \leq k^2/D$, which is less than twice the original approximation error, $\varepsilon$.

\begin{proof}[Proof of Lemma~\ref{lemma: approx Haar twirl}]
    Let $\rho_a \equiv [\Phi_a \otimes \mathbbm{1}](P_{\text{EPR}}) = \frac{1}{D^{2k}} \sum_\sigma \sigma \otimes \sigma$, and $\rho_H \equiv [\Phi_H \otimes \mathbbm{1}](P_{\text{EPR}}) = \frac{1}{D^{k}} \sum_{\sigma, \tau} \Wg_{\sigma,\tau} \sigma \otimes \tau$, where $\Wg_{\sigma,\tau}$ is the Weingarten function.
    The triangle inequality yields,
    \begin{equation} \label{eq: H a bound}
        \lVert \rho_H - \rho_a \rVert_\infty 
        \leq \frac{1}{D^{2k}} \sum_{\sigma,\tau} \left| D^{k} \Wg_{\sigma,\tau} - \delta_{\sigma,\tau} \right|
        = \frac{1}{D^{2k}} \sum_{\sigma,\tau}  \left( D^{k} \left| \Wg_{\sigma,\tau} \right| - \delta_{\sigma,\tau} \right),
    \end{equation}
    where in the first inequality we use $\lVert \sigma \otimes \tau \rVert_\infty = 1$, and in the second we use $D^k \Wg_{\sigma,\sigma} \geq 1$~\cite{collins2017weingarten,aharonov2021quantum} to simplify the absolute value.
    The final expression is easily computed, using $\sum_{\sigma,\tau} \delta_{\sigma,\tau} = k!$ and $\sum_{\sigma,\tau} |\Wg_{\sigma,\tau} | = k! (D-k)! / D! \leq k!/(1-k^2/2D)$ from Eq.~(\ref{eq: sum abs Wg}).
    This yields $\lVert \rho_H - \rho_a \rVert_\infty \leq \frac{k!}{D^{2k}} \frac{k^2/2D}{1-k^2/2D} \leq \frac{k!}{D^{2k}} (k^2/D)$.
    Applying Lemma~\ref{lemma: relative error to a} completes our proof.
\end{proof}

\subsubsection{Bounding the relative error with respect to the approximate Haar twirl} \label{app: relative via EPR}

To bound the relative error in Lemma~\ref{lemma: approx Haar twirl} (as well as Lemma~\ref{lemma: AB BC to ABC app}, below), we prove the following general lemma.
The lemma bounds the relative error between the twirl over any ensemble $\mathcal{E}$ and the approximate Haar twirl, in terms of the additive error when the twirls are applied to the EPR state.

\begin{lemma}[Relative error from EPR states] \label{lemma: relative error to a}
    Consider a unitary ensemble $\mathcal{E}$ and its twirl $\Phi_{\mathcal{E}}$.
    The twirl is approximated by $\Phi_a$ up to relative error, $(1-\varepsilon) \Phi_a \preceq \Phi_{\mathcal{E}} \preceq (1+\varepsilon) \Phi_a$, where
    \begin{equation} \label{eq: relative error EPR}
        \varepsilon = \frac{k!}{D^{2k}} \big\lVert \big[ \delta \Phi_{\mathcal{E}}  \otimes \mathbbm{1} \big](P_{\text{EPR}})\big\rVert_\infty.
    \end{equation}
    Here, $D=2^n$, $\delta \Phi_{\mathcal{E}} \equiv \Phi_{\mathcal{E}} - \Phi_a$, and $P_{\text{EPR}}$ is the projector onto the EPR state on $\mathcal{H}^{\otimes k} \otimes \mathcal{H}^{\otimes k}$.
\end{lemma}
\noindent Our proof of Lemma~\ref{lemma: relative error to a} closely follows that of Lemma~30 in Ref.~\cite{brandao2016local}, but is simpler because we use $\Phi_a$ as a reference state instead of $\Phi_H$.
\begin{proof}
    Let $\rho_a \equiv [\Phi_a \otimes \mathbbm{1}](P_{\text{EPR}}) = \frac{1}{D^{2k}} \sum_\pi \pi \otimes \pi$, and $\rho_{\mathcal{E}} \equiv [\Phi_{\mathcal{E}} \otimes \mathbbm{1}](P_{\text{EPR}})$.
    We proceed in four steps.

    \vspace{3mm}
    \noindent (1) $\rho_a$ is equal to $k!/D^{2k}$ times the projector onto the symmetric subspace of $\mathcal{H}^{\otimes k} \otimes \mathcal{H}^{\otimes k}$~\cite{harrow2013church}.
    Hence, on the symmetric subspace, $\rho_a$ has a flat spectrum with eigenvalue  $k!/D^{2k}$.

    \vspace{3mm}
    \noindent (2) The state $\rho_{\mathcal{E}}$ has support entirely within the symmetric subspace. 
    This follows because $(\pi \otimes \pi)(U \otimes \mathbbm{1}) \ket{\Psi_{\text{EPR}}} = (U \otimes \mathbbm{1}) (\pi \otimes \pi) \ket{\Psi_{\text{EPR}}} = (U \otimes \mathbbm{1}) \ket{\Psi_{\text{EPR}}}$, for any $\pi$ and $U \sim \mathcal{E}$. Here, $\dyad{\Psi_{\text{EPR}}} \equiv P_{\text{EPR}}$.

    \vspace{3mm}
    \noindent (3) Steps (1) and (2) immediately imply that the twirl has relative error $\varepsilon$ [Eq.~(\ref{eq: relative error EPR})] on the EPR state.

    \vspace{3mm}
    \noindent (4) The relative error on the EPR state upper bounds the relative error on any state.
    This follows because we can express $\Phi(\rho) = D^{2k}\tr_2( (\mathbbm{1} \otimes \rho^T) [ \Phi \otimes \mathbbm{1} ](P_{\text{EPR}}))$ for any $\Phi$, where  the trace is over the second copy of $\mathcal{H}^{\otimes k}$.
\end{proof}

Combining Lemma~\ref{lemma: relative error to a} (which 
 bounds the relative error between $\Phi_{\mathcal{E}}$ and $\Phi_a$) with Lemma~\ref{lemma: approx Haar twirl} (which bounds the relative error between $\Phi_a$ and $\Phi_H$; see in particular Eq.~(\ref{eq: inv approx Haar twirl}))  yields Lemma~\ref{lemma: relative to additive} of the main text (which bounds the relative error between $\Phi_{\mathcal{E}}$ and $\Phi_H$, showing that $\mathcal{E}$ is a unitary design).

\subsubsection{Gluing two random unitaries} \label{app: gluing}

We now prove Lemma~\ref{lemma: AB BC to ABC} of the main text, which allows us to ``glue'' together approximate  designs on local patches of qubits to form larger approximate designs.
We first re-state the lemma with an error bound that yields tighter constants for the applications.
\begin{lemma}[Gluing two random unitaries, formal] \label{lemma: AB BC to ABC app}
Let $A$, $B$, $C$ be three disjoint subsystems. 
Consider a random unitary given by $V_{ABC} = U_{AB} U_{BC}$, where $U_{AB}$ and $U_{B C}$ are drawn from $\varepsilon_{AB}$ and $\varepsilon_{BC}$-approximate unitary $k$-designs, respectively. 
Then $V_{ABC}$ is an $\varepsilon$-approximate unitary $k$-design for 
\begin{equation} \label{eq: ABC error}
    1+\varepsilon = (1+\varepsilon_{AB})(1 + \varepsilon_{BC}) 
    \left( 1-\frac{k^2}{2 D_{AB}} \right)^{-1}
    \left( 1-\frac{k^2}{2 D_{BC}} \right)^{-1}
    e^{k^2/2D_B}
    \left( 1 + \frac{k^2}{D_{ABC}} \right),
\end{equation}
as long as $k^2 \leq D_B/2$, where $D_\alpha = 2^{|\alpha|}$ is the Hilbert space dimension of subsystem $\alpha$.
\end{lemma}


\noindent Our proof is extremely short, building upon the technical lemmas proven in the previous sections.

\begin{proof}   
    Let $\delta_{AB} = k^2/2D_{AB}/(1-k^2/2D_{AB})$ and $\delta_{BC} = k^2/2D_{BC}/(1-k^2/2D_{BC})$; elementary algebra gives $1+\delta_{AB} = (1-k^2/2D_{AB})^{-1}$ and $1+\delta_{BC} = (1-k^2/2D_{BC})^{-1}$.
    From Lemma~\ref{lemma: approx Haar twirl}, we can approximate the twirl over $V_{ABC}$ by $(\Phi_{a})_{AB} \circ (\Phi_{a})_{BC}$ up to relative error $(1+\varepsilon_{AB})(1+\varepsilon_{BC})(1+\delta_{AB})(1+\delta_{BC})-1$.
    To proceed, let $\rho_a \equiv [\Phi_a \otimes \mathbbm{1}](P_{\text{EPR}}) = \frac{1}{D^{2k}} \sum_\sigma \sigma \otimes \sigma$, and
    \begin{equation}
        \rho_{\mathcal{E},a} \equiv \big[ \big( (\Phi_{a})_{AB} \circ (\Phi_{a})_{BC} \big) \otimes \mathbbm{1} \big] ( P_{\text{EPR}} ) = 
        \frac{1}{D^{2k} D_B^{k}} \sum_{\sigma,\tau} G^B_{\sigma,\tau} \cdot \tau_{AB} \sigma_{C} \otimes \tau_{A} \sigma_{BC},
    \end{equation}
    where we denote $G^B_{\sigma,\tau} \equiv G_{\sigma,\tau}(2^{|B|}) = \Tr_B (\sigma^{}_B \tau_B^{-1})$.
    The triangle inequality and Eq.~(\ref{eq: sum G}) yield,
    \begin{equation} \label{eq: H a bound-2}
        \lVert \rho_{\mathcal{E},a} - \rho_a \rVert_\infty \leq \frac{1}{D^{2k}} \sum_{\sigma,\tau} (  D_B^{-k} G^B_{\sigma,\tau} - \delta_{\sigma,\tau} ) = \frac{k!}{D^{2k}} \left( \frac{1}{D_B^k} \frac{(k+D_B-1)!}{D_B!} - 1 \right) \leq \frac{k!}{D^{2k}} \left( e^{k^2/2D_B} - 1 \right).
    \end{equation}
    Applying Lemma~\ref{lemma: relative error to a} shows that the twirl over $V_{ABC}$ is close to $\Phi_a$ up to relative error $(1+\varepsilon_{AB})(1+\varepsilon_{BC})(1+\delta_{AB})(1+\delta_{BC}) e^{k^2/2D_B}-1$.
    Applying Eq.~(\ref{eq: inv approx Haar twirl}) for $D \rightarrow D_{ABC}$  completes the proof.
\end{proof}

\subsubsection{Completing the proof of Theorem~\ref{thm:main-design}} \label{app: thm 1 proof}

\begin{proof}[Proof of Theorem~\ref{thm:main-design}]
We apply Lemma~\ref{lemma: AB BC to ABC} (Lemma~\ref{lemma: AB BC to ABC app}) patch-by-patch as described in Section~\ref{sec:proof-overview} in the main text and Fig.~\ref{fig: proof overview}.
Recall that $m$ is the number of local patches in the $n$-qubit system, and the number of small random unitaries is equal to $m-1$.
After all $m-2$ applications of Lemma~\ref{lemma: AB BC to ABC app}, we find that the two-layer brickwork ensemble forms an approximate unitary $k$-design with error
\begin{equation} \label{eq: total error}
\begin{split}
    (1+\varepsilon/n)^{m-1} \left( 1 + f(k,q) \right)^{m-2} - 1 & \leq e^{(m-1) \varepsilon/n + (m-2) f(k,q)} -1  \\
    & \leq \frac{1}{\log(2)} \left( (m-1) \varepsilon / n + (m-2) f(k,q) \right),
\end{split}
\end{equation}
where we use the inequality $e^{x}-1 \leq x/\log(2)$ for $x \leq \log(2)$, and we abbreviate,
\begin{equation}
    f(k,q) = 2 \left( \frac{k^2}{q} + \frac{k^2}{q^2} + \frac{k^4}{q^3} + \frac{k^2/2q^2}{1-k^2/2q^2} \right) \left(1+\frac{k^2}{q^2} \right).
\end{equation}
To establish the theorem, we will show that the error, Eq.~(\ref{eq: total error}), is less than $\varepsilon$.

We take $k \geq 2$ and $n \geq 3\xi$ since otherwise the theorem holds trivially.
Combined with the assumptions $q \geq nk^2/\varepsilon$ and $\varepsilon \leq 1$, these imply that $\xi \geq 7$, i.e.~$q \geq 128$.
Observing the first term in Eq.~(\ref{eq: total error}), we have
\begin{equation}
    \frac{(m-1) \varepsilon}{ n \log(2) } \leq \frac{\varepsilon}{ 7 \log(2) },
\end{equation}
since $m \leq n / \xi \leq n/7$.
Meanwhile, applying $q \geq nk^2/\varepsilon$ to the second term in Eq.~(\ref{eq: total error}), we have
\begin{equation}
\begin{split}
    e (m-2) f(k,q) & \leq \frac{n}{7 \log(2)} \cdot 2  \left( \frac{\varepsilon}{n} + \frac{\varepsilon}{nq} + \frac{\varepsilon^2}{n^2 q} + \frac{\frac{\varepsilon}{2nq}}{1-\frac{\varepsilon}{2nq}} \right) \left( 1 + \frac{\varepsilon}{nq} \right) \leq  \varepsilon \cdot \frac{3}{7 \log(2)}.
\end{split}
\end{equation}
Taking the sum of the two terms shows that the error is less than $\varepsilon \cdot 4 / 7\log(2) \leq \varepsilon$ as desired.
\end{proof}

\subsection{Proof of Corollary~\ref{cor: upper bound design}: Low-depth random unitary designs}
\label{app: design-depth}

\noindent Our construction in Theorem~\ref{thm:main-design} enables one to exponentially improve the $n$ dependence of the depth of random unitary designs compared to existing constructions.
In this appendix, we present several examples of low-depth random unitary designs that can be realized by inserting existing constructions of unitary designs into the small random unitaries in our own construction.

\vspace{0.7em}
\noindent \textbf{Proof of Corollary~\ref{cor: upper bound design}:} 
We begin with the proof of Corollary~\ref{cor: upper bound design}.
To obtain the stated depth for 1D circuits, we replace each small random unitary in Theorem~\ref{thm:main-design} with a 1D local random circuit on $2\xi$ qubits.
From Ref.~\cite{chen2024incompressibility}, such circuits form $\frac{\varepsilon}{n}$-approximate $k$-designs in depth $\mathcal{O}(\log^{7}(k)(k\xi+\log(n/\varepsilon)))$, for $k\leq c2^{2\xi/5}$, $c = \mathcal{O}(1)$.
Setting $\xi = \log_2(nk^2/\varepsilon)$, we obtain an $\varepsilon$-approximate $k$-design in depth $d = \mathcal{O}(\log(nk/\varepsilon) \cdot k \log^7(k) ) = \mathcal{O}(\log(n/\varepsilon) \cdot k \, \text{poly} \log(k) )$, as claimed.

In the specific case of $k=2,3$, we can take the small random unitaries to be random Clifford unitaries, since the Clifford unitaries form exact 2- and 3-designs~\cite{ webb2015clifford, zhu2017multiqubit}.
We can then exploit the fact that, for circuits with all-to-all connected (i.e.~geometrically non-local) two-qubit gates, any Clifford circuit on $2\xi$ qubits can be implemented in depth $\mathcal{O}(\log \xi )$ using non-local CNOT gates and $(2\xi)^2$ ancilla qubits~\cite{moore2001parallel,jiang2020optimal}.
Setting $\xi = \log_2 ( 9 n / \varepsilon)$, we obtain $\varepsilon$-approximate 2- and 3-designs in depth $\mathcal{O}( \log \log ( n / \varepsilon) )$ using $\mathcal{O}(n \log ( n / \varepsilon) )$ ancilla qubits. 
This proves the second claim of Corollary~\ref{cor: upper bound design}.

\vspace{0.7em}
\noindent \textbf{1D log-depth Clifford circuits:}
Ref.~\cite{webb2015clifford} proved that Clifford circuits form exact $3$-designs and Refs.~\cite{aaronson2004improved, kutin2007computation, bravyi2021hadamard} showed that all Clifford circuits over $2\xi$ qubits can be implemented in 1D with depth $\mathcal{O}(\xi)$. 
Hence, using our construction, one can create $\varepsilon$-approximate unitary $3$-designs over $n$ qubits using 1D Clifford circuits of depth
\begin{equation}
    d = \mathcal{O} \left(\log(n / \varepsilon)\right).
\end{equation}
This is particularly useful in classical shadow tomography~\cite{huang2020predicting}.

\vspace{0.7em}
\noindent \textbf{Low-depth unitary designs with explicit constants:}
Ref.~\cite{haah2024efficient} constructs $\frac{\varepsilon}{n}$-approximate unitary $k$-designs on $2 \xi$ qubits with depth $\mathcal{O}(k^2 \xi^2 + k \xi \log(n / \varepsilon))$ for 1D circuits\footnote{Any $\tilde{n}$-qubit Pauli $P$ can be written as $C^\dagger Z_1 C$, where $C$ is an $\tilde{n}$-qubit Clifford circuit and $Z_1$ is the Pauli-$Z$ on the first qubit. 
Hence, $e^{\rmi \theta P} = C^\dagger e^{\rmi \theta Z_1} C$ can be implemented by a 1D quantum circuit of depth $2 \tilde{n}$ and a general circuit of depth $2\log_2(\tilde{n})$. } and $\mathcal{O}( k^2 \xi \log(\xi) + k \log(\xi) \log(n / \varepsilon))$ for general circuits. 
While the $k$-dependence is not optimal as in Ref.~\cite{chen2024incompressibility}, the design depth in Ref.~\cite{haah2024efficient} comes with very small constants, which may make the bound more practical. 
More concretely, plugging the design depth of $( 4\log(8) (2\xi) k^2+4k\log(n/\varepsilon) ) \cdot 2 \log_2(2\xi)$ for general circuits~\cite{haah2024efficient} into Theorem~\ref{thm:main-design}, we obtain $\varepsilon$-approximate unitary $k$-designs over $n$ qubits with depth
\begin{align}
    d & = 8 \log( 2 \log(nk^2/\varepsilon) )(2\log(8)k^2 \log(nk^2/\varepsilon) +k\log(n/\varepsilon)),
\end{align}
which is nearly quadratic in $k$ and nearly linear in $\log(n)$.

\subsection{Proof of Proposition~\ref{prop: lower bound design}: Lower bounds for unitary designs} \label{app: lower bounds}

We now provide our proof of Proposition~\ref{prop: lower bound design}, as described in the main text.

\begin{proof} 
We consider the output distribution when a random unitary $U$ from the ensemble of interest is applied to the zero state, and the resulting state is measured in a random product basis.
Let us denote the random basis states as $\{ v \ket{x} \}$, where $v$ is a tensor product of single-qubit unitaries and $x \in \{ 0,1 \}^n$.
Each outcome $x$ occurs with probability
\begin{equation}
    p_{U,v}(x) = \left| \bra{x} v^\dagger U \ket{0^n} \right|^2.
\end{equation}
We consider the expectation value of the collision probability over both the random unitary and the random basis, multiplied by $2^n$ for convenience,
\begin{equation}
	Z_{\mathcal{E}} = 2^{n} \E_{U \sim \mathcal{E}, v} \left[  \sum_x p_{U,v}(x)^2 \right].
\end{equation}
For a Haar-random unitary, one finds $Z_{H} = 2$~\cite{arute2019quantum, dalzell2022random}.
Hence, for $\mathcal{E}$ to form a 2-design with relative error $\varepsilon$, we require $Z_\mathcal{E} \leq 2 (1+\varepsilon)$.

We can further analyze the collision probability by performing the average over the random basis $v$.
Applying a standard formula for the twirl over a tensor product of single-qubit 2-designs~\cite{elben2023randomized}, we have
\begin{equation}
	Z_{\mathcal{E}} = 2^{n}  \E_{U \sim \mathcal{E}} \left[ \, \sum_{\ell = 0}^n   3^{-\ell} \tr \big( U \dyad{0^n} U^\dagger \cdot \mathcal{P}_{\ell} \big\{ U \dyad{0^n} U^\dagger \big\}  \big)  \right],
\end{equation}
where $\mathcal{P}_{\ell}$ is a super-operator that projects onto Pauli operators with weight $\ell$.
The contribution of each Pauli operator is damped exponentially in its weight, by $3^{-\ell}$, which is equal to the probability that the Pauli commutes with the random measurement basis.

To proceed, let us decompose the initial state $\dyad{0^n}$ as a sum of three operators,
\begin{equation} \label{eq: decompose state}
	\dyad{0^n} = \frac{1}{2^n} \mathbbm{1} + \rho_{1} + \rho_{>}.
\end{equation}
The first operator is the identity component of the state, which has weight zero.
The second operator is defined to contain all stabilizers of the state with weight one, $\rho_{1} = \frac{1}{2^n} \sum_i Z_i$.
The third operator, $\rho_{>}$, contains all of the remaining stabilizers, which have weight greater than one.
The three operators are orthogonal by definition.

After $U$ is applied, the first operator is unchanged and the second operator is evolved into an operator with weight at most $L$, where $L$ is the number of qubits in the light-cone of the unitary $U$.
In 1D circuits, we have $L = 2d$ where $d$ is the depth of $U$.
In all-to-all circuits, we have $L=2^d$.

To lower bound the collision probability in terms of the circuit depth $d$, let us consider the total support of $U \dyad{0^n} U^\dagger$ on weights between 1 and $L$, 
\begin{equation}
	P_L \equiv \sum_{\ell=1}^L \tr( U \dyad{0^n} U^\dagger \cdot \mathcal{P}_\ell \{ U \dyad{0^n} U^\dagger \} ).
\end{equation}
The restriction to $\ell \geq 1$ eliminates the first operator in Eq.~(\ref{eq: decompose state}).
Since the second and third operators are orthogonal before we apply $U$, they are also orthogonal after we apply $U$.
Moreover, they remain orthogonal even after projecting to weights $\ell \leq L$, since $\rho_1$ only has support on such weights.
Thus, we have
\begin{equation}
	P_L =  \tr( \rho_1^2 ) + 
	\sum_{\ell=0}^L \tr( U \rho_> U^\dagger \cdot \mathcal{P}_\ell \{ U \rho_> U^\dagger \} ) \geq \frac{n}{2^n}.
\end{equation}
The lower bound follows because the second term is non-negative and the first term is equal to $n/2^n$. 
This in turn lower bounds the expected collision probability, 
\begin{equation}
\begin{split}
	Z_\mathcal{E} & \geq 1 + 2^n 3^{-L} P_L \geq 1 + \frac{n}{3^L}.
\end{split}
\end{equation}
For $\mathcal{E}$ to form an $\varepsilon$-approximate 2-design, we require $n/3^L \leq 1+2\varepsilon$.
This requires depth
\begin{equation}
	d \geq \frac{1}{2} \log_3\left( \frac{n}{1+2\varepsilon} \right)
\end{equation}
in 1D circuits,
and depth
 \begin{equation}
	d \geq \log_2 \left( \log_3\left( \frac{n}{1+2\varepsilon} \right) \right),
\end{equation}
in general all-to-all circuit architectures. 
\end{proof}

\subsection{Additional lower bounds for the $\varepsilon$-dependence of low-depth unitary designs} \label{app: lower bounds eps}

Here, we provide an additional lower bound, showing that the $\varepsilon$-dependence of our approximate unitary designs is also optimal.
This question was raised in Ref.~\cite{fefferman2024anti}, which proved that a depth $d = \Omega(\log(1/\varepsilon))$ was necessary to form unitary designs in any circuit ensemble composed of independently Haar-random local gates.
Our proposition below provides an analogous lower bound for arbitrary circuit ensembles.
As our Corollary~\ref{cor: upper bound design} demonstrates, both of these lower bounds are optimal for their respective circuit ensembles.

\begin{proposition} \label{prop: lower bound eps}
    {\emph{(Depth lower bound for unitary designs; $\varepsilon$-dependence)}}
    Any quantum circuit ensemble over $n$ qubits that forms an $\varepsilon$-approximate unitary $2$-design requires circuit depth
    \begin{itemize}
    \item $d = \Omega \big( \! \log(1/\varepsilon) \big)$, for 1D circuits with depth $d \leq n/4$ and any number of ancilla qubits,
    \item $d = \Omega \big(  \! \log \log(1/\varepsilon) \big)$, for all-to-all circuits with depth $d \leq \log_2(n/2)$ and any number of ancilla qubits. We assume that every all-to-all circuit in the ensemble has the same architecture.
    \end{itemize}
\end{proposition}
\noindent The bound is specific to \emph{low-depth} unitary designs, in which the circuit light-cone is smaller than the system size.
This restriction is necessary, since the Clifford group forms an exact unitary 3-design (i.e.~$\varepsilon$ = 0), and can be compiled in depth $\mathcal{O}(n)$ in 1D and $\mathcal{O}(\log n)$ in all-to-all circuits.

\begin{proof}
Let $\mathcal{E}$ denote the circuit ensemble, $\mathfrak{L}$ denote the set of qubits in the light-cone of the first qubit under $\mathcal{E}$, and $L = | \mathfrak{L} |$ the size of the light-cone.
We consider the state obtained by preparing the first qubit in the zero state and all other qubits in the maximally mixed state, and applying $U$, 
\begin{equation}
    U \left( \dyad{0} \otimes (\mathbbm{1}/2)^{\otimes n-1} \right) U^\dagger = \frac{\mathbbm{1} + U Z_0 U^\dagger}{2^n},
\end{equation}
where $Z_0$ is the Pauli-$Z$ operator on the first qubit.
After applying $U$, we perform a SWAP test on $\mathfrak{L}$.
Since the light-cone of the first qubit is contained in $\mathfrak{L}$ by definition, this produces an expected outcome,
\begin{equation}
    \E_{U \sim \mathcal{E}} \left[ \tr( \left( \frac{\mathbbm{1} + U Z_0 U^\dagger}{2^n} \right)^{\otimes 2} \text{SWAP}_{\mathfrak{L}}  ) \right] = 2^{-L+1},
\end{equation}
where the $+1$ counts the single qubit of purity arising from the first qubit.
On the other hand, if $U$ were Haar-random, the SWAP test would produce an expected outcome,
\begin{equation}
    \E_{U \sim H} \left[ \tr( \left( \frac{\mathbbm{1} + U Z_0 U^\dagger}{2^n} \right)^{\otimes 2} \text{SWAP}_{\mathfrak{L}}  ) \right] = 2^{-L} + 2^{-L} \cdot \frac{4^{L}-1}{4^n-1},
\end{equation}
since the qubits in $\mathfrak{L}$ become maximally entangled with the rest of the system.
The ratio in the second term corresponds to the probability that a random non-identity Pauli operator on all $n$ qubits has support only on the qubits in $\mathfrak{L}$.

Observing the above, in order for $\mathcal{E}$ to form an $\varepsilon$-approximate unitary 2-design, we must have
\begin{equation}
    2^{-L+1} \leq 2^{-L} + 2^{-L} \cdot \frac{4^{L}-1}{4^n-1} + 2\varepsilon,
\end{equation}
which implies
\begin{equation} \label{eq: L bound}
    L \geq \log_2(1/2\varepsilon) - \log(1 -  \frac{4^{L}-1}{4^n-1} ).
\end{equation}
The light-cone contains at most $L \leq 2d$ qubits in 1D circuits, and  $L \leq 2^d$ in all-to-all circuits.
By assumption, we have $d \leq n/4$ in 1D and $d \leq \log_2(n/2)$ in all-to-all circuits, and hence $L \leq n/2$ in both cases. 
Therefore, the second term in Eq.~(\ref{eq: L bound}) is at most a constant, and we require $d = \Omega(\log(1/\varepsilon))$ in 1D circuits,  and $d = \Omega( \log \log(1/\varepsilon) )$ in all-to-all circuits, to form an $\varepsilon$-approximate 2-design.
\end{proof}

\section{Pseudorandom unitaries}
\label{app: PRUs}

In this appendix, we present the definition of pseudorandom unitaries, a simple construction for pseudorandom unitaries, a proof of our main theorem (Theorem~\ref{thm:main-PRUs}) for gluing small pseudorandom unitaries, and a proof of Corollary~\ref{cor:pseudorandom unitaries} for constructing low-depth pseudorandom unitaries.

\subsection{Definition of pseudorandom unitaries}
\label{app: def-PRUs-formal}

We begin with the formal definition of pseudorandom unitaries with security against $t(n)$-time adversaries,
based on the definitions proposed in~\cite{PRS2018, metger2024simple, chen2024efficient}.

\begin{definition}[Pseudorandom unitaries] \label{def: PRU-t(lambda)}
Let $n \in \mathbb{N}$ be the number of qubits.
An infinite sequence $\{ \mathcal{U}_n \}_{n \in \mathbb{N}}$ of $n$-qubit unitary ensembles $\mathcal{U}_n = \{ U_{\mathsf{key}} \}_{\mathsf{key} \in \mathcal{K}_{n}}$ for the key space $\mathcal{K}_{n}$ is a pseudorandom unitary secure against any $t(n)$-time adversary if it satisfies the following.
\begin{itemize}
    \item \textbf{Efficient computation:} There exists a $\mathrm{poly}(n)$-time quantum algorithm that implements the $n$-qubit unitary $U_{\mathsf{key}}$ for all $\mathsf{key} \in \mathcal{K}_{n}$.
    \item \textbf{Indistinguishability from Haar:} Any quantum algorithm $\mathcal{A}^{U}(1^{n})$ that runs in time $\leq t(n)$, queries an $n$-qubit unitary $U$ for any number of times, and outputs $\{0, 1\}$ satisfies
    \begin{equation} \label{eq:distinguishability-adv}
        \left| \Pr_{\mathsf{key} \sim \mathcal{K}_{n}}\left[\mathcal{A}^{U_{\mathsf{key}}}(1^{n}) = 1 \right] - \Pr_{U \sim \mathcal{E}_H}\left[ \mathcal{A}^{U}(1^{n}) = 1 \right] \right| \leq \mathrm{negl}(n),
    \end{equation}
    where $\mathcal{E}_H$ is the Haar-random unitary ensemble and $\mathrm{negl}(n)$ is the negligible function, i.e., a function that is asymptotically smaller than any inverse-polynomial function $1 / \mathrm{poly}(n)$.
\end{itemize}
A $t(n)$-time algorithm $\mathcal{A}^{U}(1^{n})$ is referred to as a $t(n)$-time adversary.
The difference between the probability under $\mathcal{U}_n$ and Haar ensemble is referred to as the advantage of $\mathcal{A}^{U}(1^{n})$.
\end{definition}

When the pseudorandom unitary is secure against a $\poly(n_{\mathrm{big}})$-time adversary, where $n_{\mathrm{big}}$ is a function of $n$ that grows faster than $n$, $\mathrm{negl}(n)$ can be strengthened to $\mathrm{negl}(n_{\mathrm{big}})$.

\begin{fact}[Strengthening the negligible function] \label{fact:strengthening-negl}
Let $n_{\mathrm{big}}$ be a function of $n$ such that $n_{\mathrm{big}} = \Omega(n)$.
For any $n$-qubit pseudorandom unitary secure against $\poly(n_{\mathrm{big}})$-time adversary, we can show that any quantum algorithm $\mathcal{A}^{U}(1^{n})$ that runs in $\poly(n_{\mathrm{big}})$ time, queries an $n$-qubit unitary $U$ for any number of times, and outputs $\{0, 1\}$ must satisfy
\begin{equation} \label{eq:subexp-negli}
    \left| \Pr_{\mathsf{key} \sim \mathcal{K}_{n}}\left[\mathcal{A}^{U_{\mathsf{key}}}(1^{n}) = 1 \right] - \Pr_{U \sim \mathcal{E}_H}\left[ \mathcal{A}^{U}(1^{n}) = 1 \right] \right| \leq \mathrm{negl}(n_{\mathrm{big}}),
\end{equation}
where $\mathrm{negl}(n_{\mathrm{big}})$ is a function asymptotically smaller than any inverse-polynomial function in $n_{\mathrm{big}}$.
\end{fact}
\begin{proof}
If the advantage of a $\poly(n_{\mathrm{big}})$-time algorithm $\mathcal{A}$ is at least $1 / \poly(n_{\mathrm{big}})$, then one could construct another $\poly(n_{\mathrm{big}})$-time algorithm that repeats $\mathcal{A}$ for $\poly(n_{\mathrm{big}})$ many times and votes on the outputs of the repetitions to achieve a constant advantage.
This contradicts with the condition Eq.~\eqref{eq:distinguishability-adv}.
Hence, the advantage of any $\poly(n_{\mathrm{big}})$-time algorithm $\mathcal{A}$ is negligible in $n_{\mathrm{big}}$.
\end{proof}

\subsection{The $PFC$ construction for pseudorandom unitaries}
\label{app: PFC construction}

A proposed construction of pseudorandom unitaries \cite{metger2024simple} is the $P F C$ ensemble, where $C$ is a random Clifford circuit over $n$ qubits, $F$ is a quantum-secure pseudorandom function (PRF) over $n$ bits given by a random real diagonal unitary with diagonal elements $\pm 1$ \cite{zhandry2021PRF}, and $P$ is a quantum-secure pseudorandom permutation (PRP) \cite{zhandry2016note} over $n$ bits.
Assuming no subexponential-time quantum algorithms can solve the Learning With Errors (LWE) problem \cite{regev2009lattices}, existing works \cite{zhandry2016note, zhandry2021PRF} have constructed quantum-secure pseudorandom functions $F$ and quantum-secure pseudorandom permutations $P$ that are secure against any subexponential-time quantum algorithms.

\begin{fact}[Circuit depth of $F$; \cite{zhandry2021PRF, banerjee2012pseudorandom}] \label{fact:depth-F}
An $n$-qubit quantum-secure pseudorandom function $F$ can be implemented using quantum circuit depth $\mathrm{poly}(n)$ on any geometry and circuit depth $\mathrm{poly}\log(n)$ on all-to-all-connected geometry. The number of ancilla qubits scales as $\mathrm{poly}(n)$.
\end{fact}
\begin{proof}
Ref.~\cite{banerjee2012pseudorandom} introduced pseudorandom functions $F$ based on the hardness of LWE, which can be constructed from classical Boolean circuits of depth $\log(n)$ with $\mathrm{poly}(n)$ ancilla bits ($\mathrm{NC}^1$). Ref.~\cite{zhandry2021PRF} subsequently proved that $F$ is quantum-secure.
Classical $\mathrm{NC}^1$ circuits may contain unbounded fan-out operations that allow a single bit to be copied to $\poly(n)$ ancilla bits in one layer. Such a classical operation can be compiled into an all-to-all-connected quantum circuit of depth $\log(n)$, using only two-qubit gates. To see this, note that we can use CNOT gates to copy a single qubit to two qubits, then to four, eight, and so on. The number of qubits grows exponentially with depth, hence we only need $\log(n)$ depth in quantum circuits consisting of two-qubit gates.

Together, any $\mathrm{NC}^1$ classical circuit can be implemented with all-to-all-connected quantum circuits of depth $\log^2(n)$, using $\mathrm{poly}(n)$ ancilla qubits. This shows that pseudorandom functions~$F$ can be implemented using all-to-all connected quantum circuits with circuit depth $\mathrm{poly}\log(n)$ and $\mathrm{poly}(n)$ ancilla qubits.
To compile the all-to-all-connected quantum circuit into any geometry, the circuit depth increases to at most $\mathrm{poly}(n)$.
This applies to any circuit connectivity represented by a connected graph.
This is because the original circuit contains at most $\mathrm{poly}(n)$ two-qubit gates. For each such gate $g$, we can use $\mathrm{poly}(n)$ swap gates to move the two qubits on which $g$ acts to neighboring positions, apply the gate, and then return them to their original positions using another $\mathrm{poly}(n)$ swap gates.
\end{proof}

\begin{fact}[Circuit depth of $P$; \cite{zhandry2016note}] \label{fact:depth-P}
An $n$-qubit quantum-secure pseudorandom permutation $P$ can be implemented using quantum circuit depth $\mathrm{poly}(n)$ on any geometry with $\mathrm{poly}(n)$ ancilla qubits.
\end{fact}
\begin{proof}
Ref.~\cite{zhandry2016note} proved that quantum-secure pseudorandom permutations $P$ can be constructed from quantum-secure pseudorandom functions. However, even in all-to-all-connected circuits, Ref.~\cite{zhandry2016note} requires the circuit depth to be $\mathrm{poly}(n)$.
Using the same construction as Fact~\ref{fact:depth-F}, we can implement $P$ with quantum circuit depth $\mathrm{poly}(n)$ on any geometry.
\end{proof}

\begin{fact}[Circuit depth of Clifford $C$; \cite{moore2001parallel,jiang2020optimal}] \label{fact:depth-C}
Any $n$-qubit Clifford circuit $C$ can be implemented in $\mathcal{O}(n)$ depth in 1D circuits without using any ancilla qubits and $\mathcal{O}(\log n)$ depth in all-to-all-connected circuits using $\mathcal{O}(n^2)$ ancilla qubits.
\end{fact}

Using the three facts above, we  arrive at the following fact.

\begin{fact}[Circuit depth of $PFC$] \label{fact:depth-PFC}
The $n$-qubit unitary $U = PFC$ can be implemented using quantum circuit depth $\mathrm{poly}(n)$ on any geometry. The number of ancilla qubits scales as $\mathrm{poly}(n)$.
\end{fact}

While establishing indistinguishability between the $PFC$ ensemble and the Haar ensemble under the hardness assumption of LWE remains a conjecture in Ref.~\cite{metger2024simple}, \cite{PRU2024} has resolved this conjecture using a surprisingly simple proof. The formal statement~\cite{PRU2024} is given below.

\begin{theorem}[$PFC$ is a pseudorandom unitary; \cite{PRU2024}] \label{thm:PFC-PRU}
Let $n$ be the number of qubits.
Consider any subexponential function $t(n) = \exp(o(n))$.
Given a quantum-secure pseudorandom function $F$ over $n$ bits secure against $t(n)$-time adversary, a quantum-secure pseudorandom permutation $P$ over $n$ bits secure against $t(n)$-time adversary, and a random $n$-qubit Clifford circuit $C$.
The $n$-qubit random unitary ensemble $U = P F C$ is a pseudorandom unitary secure against any $t(n)$-time adversary.
\end{theorem}

\noindent As a result, under the cryptographic assumption that there are no subexponential-time quantum algorithms for solving LWE, one can show that the random unitary ensemble $P F C$ is a pseudorandom unitary secure against any subexponential-time quantum algorithms, i.e., $t(n) = \exp(o(n))$ in Def.~\ref{def: PRU-t(lambda)}.

\subsection{Pseudorandom unitaries in polylog depth}

While the $PFC$ construction requires polynomial circuit depth due to the fact that known constructions for quantum-secure pseudorandom permutations \cite{zhandry2016note} require polynomial depth, Ref.~\cite{SPRU2024} provides a new PRU construction that can be implemented in $\poly \log n$ depth on all-to-all-connected circuits by removing the dependence on quantum-secure pseudorandom permutations.

\begin{theorem}[Pseudorandom unitary in polylog depth; \cite{SPRU2024}] \label{thm:QNC-PRU}
Let $n$ be the number of qubits. Assuming no subexponential-time quantum algorithm can solve LWE, there exists an $n$-qubit pseudorandom unitary that can be implemented in $\poly \log n$ depth under all-to-all-connected circuits and is secure against any subexponential-time adversary.
\end{theorem}

Using the same method as in the proof of Fact~\ref{fact:depth-F}, we can implement pseudorandom unitaries with $\poly(n)$ depth on any circuit geometry assuming hardness of LWE.

\subsection{Proof of Theorem~\ref{thm:main-PRUs}: Gluing small pseudorandom unitaries} 
\label{app: proof of main PRUs}

We re-state and prove Theorem~\ref{thm:main-PRUs}, which shows that we can use the two-layer brickwork ensemble to glue small pseudorandom unitaries over $\omega(\log n)$ local patches of qubits into a pseudorandom unitary on $n$ qubits.

\begin{theorem}[Gluing small pseudorandom unitaries]
    Let $n$ be the number of qubits in the whole system and $\xi = \omega(\log n)$ be the number of qubits in each local patch.
    Suppose each small random unitary in the two-layer brickwork ensemble $\mathcal{E}$ is a $2\xi$-qubit pseudorandom unitary secure against $\poly(n)$-time adversaries.
    Then $\mathcal{E}$ is an $n$-qubit pseudorandom unitary secure against $\poly(n)$-time adversaries.
\end{theorem}
\begin{proof}
    Let us consider a two-layer brickwork ensemble $\mathcal{E}^*$ over $n$ qubits, where each of the small random unitaries over $2 \xi$ qubits are Haar-random.
    We will prove the following two claims:
    \textbf{(1)} any $\mathrm{poly}(n)$-time quantum algorithm cannot distinguish between $\mathcal{E}$ and $\mathcal{E}^*$, \textbf{(2)} any $\mathrm{poly}(n)$-time quantum algorithm cannot distinguish between $\mathcal{E}^*$ and Haar-random ensemble.
    Hence, any $\mathrm{poly}(n)$-time quantum algorithm cannot distinguish between $\mathcal{E}$ and Haar-random ensemble.

    \vspace{0.75em}
    \noindent \textbf{Proving Claim (1).}
    We use a hybrid argument to establish the claim.
    Recall that $m-1$ is the number of small random unitaries in the two-layer brickwork\footnote{Here, $m$ is the number of local patches. Since each small random unitary acts on two local patches and there are two layers of small random unitaries, where the upper layer has one less unitary compared to the lower layer, the number of small random unitaries is equal to $m/2 + m/2 - 1 = m - 1$.}.
    We consider a sequence of $m$ ensembles $\mathcal{E}_0, \ldots, \mathcal{E}_{m-1}$, where $\mathcal{E}_j$ corresponds to having the first $j$ small random unitaries sampled from a $2\xi$-qubit pseudorandom unitary secure against $\poly(n)$-time adversaries, and the rest of the $m - 1 - j$ small random unitaries sampled from the Haar ensemble.
    By definition, we have $\mathcal{E}_0 = \mathcal{E}$ and $\mathcal{E}_{m-1} = \mathcal{E}^*$.
    The advantage of any $\mathrm{poly}(n)$-time quantum algorithm $\mathcal{A}$ to distinguish between $\mathcal{E}_{j-1}$ and $\mathcal{E}_j$ can be bounded above using a proof of contradiction.

    Using Fact~\ref{fact:strengthening-negl} with $n, n_{\mathrm{big}} \leftarrow 2\xi, n$, for any $2 \xi$-qubit pseudorandom unitary ensemble $\mathcal{V}$ secure against $\poly(n)$-time adversaries, we can bound the advantage of any algorithm $\mathcal{A}^{V}$ that runs in $\poly(n)$ time, queries a $2\xi$-qubit unitary $V \sim \mathcal{V}$ for any number of times, and outputs $\{0, 1\}$, as follows,
    \begin{equation} \label{eq:negl-n-adv-largersize}
        \left| \Pr_{V \sim \mathcal{V}}\left[\mathcal{A}^{V} = 1 \right] - \Pr_{V \sim \mathcal{E}_H}\left[ \mathcal{A}^{V} = 1 \right] \right| \leq \mathrm{negl}(n),
    \end{equation}
    where $\mathcal{E}_H$ is the Haar ensemble.
    Suppose there exists a polynomial function $p(n)$, such that
    \begin{align}
        \left| \Pr_{U \sim \mathcal{E}_{j}}\left[\mathcal{A}^{U} = 1 \right] - \Pr_{U \sim \mathcal{E}_{j-1}}\left[ \mathcal{A}^{U} = 1 \right] \right| > \frac{1}{p(n)}. \label{eq:subexp-contradiction}
    \end{align}
    The two ensembles $\mathcal{E}_{j-1}$ and $\mathcal{E}_j$ differ by a single small random unitary over $2 \xi$ qubits, where the small unitary is a $2 \xi$-qubit pseudorandom unitary secure against $\poly(n)$-time algorithm in $\mathcal{E}_{j-1}$ and is a Haar-random unitary in $\mathcal{E}_j$. We now show that Eq.~\eqref{eq:subexp-contradiction} contradicts Eq.~\eqref{eq:negl-n-adv-largersize}.

    We note that every small random unitary except for the $j$-th one is sampled from the same ensemble, either pseudorandom or Haar-random, between $\mathcal{E}_{j-1}$ and $\mathcal{E}_j$.
    The pseudorandom unitaries can be implemented in time $\mathrm{poly}(\xi)$.
    While the Haar-random unitaries require time $\mathrm{exp}(\xi)$ to implement, we can simulate their effect efficiently because any $\exp(o(\xi))$-time algorithm can make at most $\exp(o(\xi))$ queries to the Haar-random unitaries.
    Using unitary $\exp(o(\xi))$-designs~\cite{brandao2016local}, we can simulate $\exp(o(\xi))$ queries to a Haar-random unitary over $2 \xi$ qubits in $\exp(o(\xi))$ time.
    Together, we have devised a subexponential-time algorithm yielding an inverse-subexponential advantage in distinguishing between a pseudorandom unitary and a Haar-random unitary, which contradicts the assumption that the $2 \xi$-qubit pseudorandom unitary in $\mathcal{E}_{j-1}$ is secure against any $\exp(o(\xi))$-time adversary.

    By contradiction, we have the following upper bound on the advantage for any $\mathrm{poly}(n)$-time quantum algorithm $\mathcal{A}$ to distinguish between $\mathcal{E}_{j-1}$ and $\mathcal{E}_j$,
    \begin{equation}
        \left| \Pr_{U \sim \mathcal{E}_{j}}\left[\mathcal{A}^{U} = 1 \right] - \Pr_{U \sim \mathcal{E}_{j-1}}\left[ \mathcal{A}^{U} = 1 \right] \right| \leq \mathrm{negl}(n).
    \end{equation}
    Because $m = \mathcal{O}(n)$, the advantage of any $\mathrm{poly}(n)$-time quantum algorithm $\mathcal{A}$ to distinguish between $\mathcal{E}$ and $\mathcal{E}^*$ can be bounded above as follows,
    \begin{align}
        \left| \Pr_{U \sim \mathcal{E}}\left[\mathcal{A}^{U} = 1 \right] - \Pr_{U \sim \mathcal{E}^*}\left[ \mathcal{A}^{U} = 1 \right] \right| &\leq \sum_{j=1}^{m-1} \left| \Pr_{U \sim \mathcal{E}_{j}}\left[\mathcal{A}^{U} = 1 \right] - \Pr_{U \sim \mathcal{E}_{j-1}}\left[ \mathcal{A}^{U} = 1 \right] \right|\\
        &\leq (m-1) \cdot \mathrm{negl}(n) = \mathcal{O}(n) \cdot \mathrm{negl}(n) = \mathrm{negl}(n).
    \end{align}
    This concludes the proof of Claim (1) that no $\mathrm{poly}(n)$-time quantum algorithm can distinguish between the two random unitary ensembles, $\mathcal{E}$ and $\mathcal{E}^*$.

    \vspace{0.75em}
    \noindent \textbf{Proving Claim (2).} We use Theorem~\ref{thm:main-design} to prove that no $\mathrm{poly}(n)$-time quantum algorithm can distinguish between $\mathcal{E}^*$ and Haar-random ensemble.
    Because $\xi = \omega(\log n)$, we can consider $k = \exp({\omega(\log n)}) = n^{\omega(1)}$ and $\varepsilon = 1 / \exp({\omega(\log n)}) = 1 / n^{\omega(1)} = \mathrm{negl}(n)$ while satisfying $\xi \geq \log_2(n k^2 / \varepsilon)$.
    Because small Haar-random unitaries over $2 \xi$ qubits are $\varepsilon$-approximate unitary $k$ designs, the two-layer brickwork ensemble $\mathcal{E}^*$ consisting of small Haar-random unitaries forms an $\varepsilon$-approximate unitary $k$-design on $n$ qubits.
    Using the operational property of approximate unitary designs given in Lemma~\ref{lem:dist-relative}, any quantum algorithm $\mathcal{A}$ that makes a $\mathrm{poly}(n) = o(k)$ number of queries to a unitary $U$ drawn from either the ensemble $\mathcal{E}^*$ or the Haar ensemble will yield an advantage of at most
    \begin{equation}
        \left| \Pr_{U \sim \mathcal{E}^*}\left[\mathcal{A}^{U} = 1 \right] - \Pr_{U \sim \mathcal{E}_H}\left[ \mathcal{A}^{U} = 1 \right] \right| \leq 2 \varepsilon = \mathrm{negl}(n).
    \end{equation}
    As a result, any quantum algorithm $\mathcal{A}^{U}$ that runs in time $\mathrm{poly}(n)$, queries an $n$-qubit unitary $U$ for any number of times, and outputs $\{0, 1\}$ can only achieve a negligible advantage for distinguishing a unitary $U$ sampled from $\mathcal{E}^*$ from a unitary $U$ sampled from the Haar ensemble.
    This establishes Claim (2).

    \vspace{0.75em}
    By combining Claim (1) and (2), we can see that the two-layer brickwork ensemble $\mathcal{E}$ is indistinguishable from Haar ensemble by any polynomial-time quantum algorithm.
    Because all the small pseudorandom unitaries are also efficiently implementable, the two-layer brickwork ensemble $\mathcal{E}$ can be efficiently implemented.
    Together, we have established the main theorem stating that $\mathcal{E}$ forms an $n$-qubit pseudorandom unitary secure against any polynomial-time adversary.
\end{proof}

\subsection{Proof of Corollary~\ref{cor:pseudorandom unitaries}: Low-depth pseudorandom unitaries} \label{app: PRU depth scaling}

Corollary~\ref{cor:pseudorandom unitaries} follows directly from Theorem~\ref{thm:main-PRUs} on gluing small pseudorandom unitaries and Theorem~\ref{thm:QNC-PRU} on constructing pseudorandom unitaries in $\mathrm{QNC}^1_f$.

Let $n$ be the number of qubits and $\xi = (\log n)^{1 + c}$ be the local patch size for any constant $c > 0$.
Using Fact~\ref{fact:depth-PFC} and Theorem~\ref{thm:PFC-PRU}, under the cryptographic assumption that LWE cannot be solved by any subexponential-time quantum algorithm, we can construct $2\xi$-qubit pseudorandom unitaries secure against any $\exp(o(\xi))$-time adversary in circuit depth $\poly(\xi)$ in 1D and depth $\poly \log(\xi)$ in all-to-all circuits.
From the condition that $\xi = (\log n)^{1+c}$ for a constant $c > 0$, any $\poly(n)$-time algorithm is an $\exp(o(\xi))$-time algorithm. Hence, the constructed $2\xi$-qubit pseudorandom unitaries are secure against any $\poly(n)$-time adversaries.

Using the relation $\xi = (\log n)^{1 + c}$, the depth $d(n)$ of the small random unitaries  scales as $\poly \log n$ in 1D and $\poly \log \log n$ in all-to-all geometry.
Hence, from Theorem~\ref{thm:main-PRUs}, the two-layer brickwork ensemble $\mathcal{E}$ with small random unitaries drawn from independently random $PFC$ circuits on $2 \xi$ qubits forms an $n$-qubit pseudorandom unitary secure against any polynomial-time adversary in circuit depth $\poly \log n$ in 1D and $\poly \log \log n$ in all-to-all circuits.

\section{Creating random unitaries on any geometry} \label{app: any geometry}

In Appendix~\ref{subsec: general 1D to any geo}, we present a method that allows one to take a depth-$d$ random unitary on a 1D line and create a depth-$\mathcal{O}(d)$ random unitary on any geometric structure. This method can be applied to both random unitary designs and pseudorandom unitaries.

In Appendix~\ref{subsec:ext-to-general-2-layer}, we present a generalization of Theorem~\ref{thm:main-design} that allows us to glue small random unitary designs using general two-layer brickwork ensembles, such as a 2D circuit consisting of two layers of overlapping squares.
Because Theorem~\ref{thm:main-PRUs} follows from Theorem~\ref{thm:main-design}, we can immediately generalize Theorem~\ref{thm:main-PRUs} to general two-layer brickwork ensembles as well.

\subsection{From 1D lines to any geometry}
\label{subsec: general 1D to any geo}

We consider a geometry over $n$ qubits to be described by a connected graph $G = (V, E)$ with the number of vertices $|V| = n$.
For simplicity, we restrict attention to graphs of bounded degree, where the number of neighboring nodes is order one $|\mathcal{N}(v)| = \mathcal{O}(1)$ for all $v \in V$.
The graph represents the neighborhood structure of the geometry. Two qubits are neighboring if they are connected by an edge.
The geometry can be a 1D chain, a 2D plane, any finite-dimensional lattice, or any finite-dimensional manifold, such as a 1D circle, a torus, a hyperbolic space, or a highly-connected expander graph.

In Theorems~\ref{thm:main-design} and~\ref{thm:main-PRUs}, we have shown how to create random unitaries on a 1D line.
In order to create random unitaries on any geometry, we need an approach to draw a line on any geometry such that the line goes through every vertex exactly once.
This is known as a Hamiltonian path of the graph.
For a general bounded-degree graph, there may not be a Hamiltonian path. Furthermore, even if one exists, finding a Hamiltonian path is known to be NP-hard.

Luckily, when we can jump from one vertex to another vertex within a \emph{small distance}, we can always efficiently find a Hamiltonian path. 
To capture this, we define a new graph $G^{(k)}$ where edges are added between any two vertices that are not too far away, i.e.~if the distance between them is at most $k$.

\begin{definition}[Neighborhood graph]
    Given any graph $G = (V, E)$ and any integer $k \geq 1$. We define $G^{(k)} = (V, E^{(k)})$ to be the graph with the same set of vertices and an expanded set $E^{(k)}$ of edges, where
    \begin{equation}
        (v_1, v_2) \in E^{(k)}
    \end{equation}
    if the distance between the vertices $v_1$ and $v_2$ satisfies $d_G(v_1, v_2) \leq k$.
\end{definition}

\noindent We show that one can always efficiently create a Hamiltonian path on $G^{(4)}$.

\begin{lemma}[Hamiltonian path with jumps] \label{lem:Hamiltonian-path}
    For any connected graph $G = (V, E)$,
    we can efficiently find a Hamiltonian path on $G^{(4)}$.
\end{lemma}
\begin{proof}
    First, we pick any vertex $r \in V$ as the root. Then, we perform a depth-first search on the graph $G$ starting from the root $r$.
    The depth-first search creates a tree $T$ over all the vertices in $V$ with $r$ as the root.
    We prove by induction that:
    \begin{center}
        For any tree $T$ of size $\geq 2$, we can create a Hamiltonian path on $T^{(4)}$ starting at the root $r$ of the tree $T$, followed by the leftmost child node, and ending at the rightmost child node of the root~$r$.
    \end{center}
    We perform induction based on the depth of the tree $T$.
    
    \vspace{0.7em}\noindent
    \textbf{Base case (depth of $T = 2$).} Let $r$ be the root, and $c_1, \ldots, c_\ell$ be the child nodes of $r$, where $c_\ell$ is the rightmost child node. The Hamiltonian path is given by $(r, c_1, \ldots, c_\ell)$.
    
    \vspace{0.7em}\noindent
    \textbf{Inductive step (depth of $T = k$).} Assume the inductive hypothesis holds for depth $\leq k-1$.
     Let $r$ be the root of the tree $T$, and the child nodes of $r$ be $c_1, \ldots, c_\ell$, and $\mathcal{T}_1, \ldots, \mathcal{T}_\ell$ be the subtree for each of the child nodes.
    If $\mathcal{T}_l$ is of size one with only one vertex $c_l$, we consider $\mathcal{P}_l = (c_l)$.
    Otherwise, we use the inductive hypothesis to find the Hamiltonian path $\mathcal{P}_l$ for the subtree $\mathcal{T}_l$ that starts at $c_l$,  travels first to the leftmost child node of $c_l$, and ends at the rightmost child node of $c_l$.
    Given a path $\mathcal{P} = (u_1, \ldots, u_p)$. We denote $\mathcal{P}^-$ to be $(u_2, \ldots, u_p, u_1)$ to be the path with the first vertex moved to the last.
    The path given by
    \begin{equation}
        r, \mathcal{P}_{1}, \mathcal{P}_{2}, \ldots, \mathcal{P}_{\ell-1}, \mathcal{P}^-_\ell
    \end{equation}
    is a Hamiltonian path on $T^{(4)}$. The longest jump is the step from the last vertex of $\mathcal{P}_{\ell-1}$ to the first vertex of $\mathcal{P}^-_\ell$, which requires a distance at most four.
    
    \vspace{0.7em}
    Together, we have efficiently found a Hamiltonian path on $T^{(4)}$ for the tree $T$.
    Because $T$ is a subgraph of $G$, the expanded graph $T^{(4)}$ is a subgraph of $G^{(4)}$.
    Therefore, we have efficiently found a Hamiltonian path on $G^{(4)}$.
\end{proof}

The following lemma shows that if we can create an $n$-qubit unitary $U$ using a depth-$d$ quantum circuit on the distance-four neighborhood graph $G^{(4)}$, then we can implement the same unitary $U$ in $\mathcal{O}(d)$-depth on the original graph $G$.
We show this by demonstrating that each circuit layer of $U$ on $G^{(4)}$ can be implemented in constant depth on $G$.
Together, Lemma~\ref{lem:Hamiltonian-path} and~\ref{lem:jumps} allow us to implement a random quantum circuit with depth $d$ on a 1D line on any geometry in depth $\mathcal{O}(d)$.

\begin{lemma}[Implementing jumps] \label{lem:jumps}
    For any graph $G = (V, E)$ with $|\mathcal{N}(v)| = \mathcal{O}(1), \forall v \in V$, a depth-$1$ circuit on $G^{(4)}$ can be implemented by a constant-depth circuit on $G$.
\end{lemma}
\begin{proof}
    A depth-$1$ circuit on $G^{(4)}$ consists of two-qubit gates acting on disjoint pairs of vertices $\{ u_a \leftrightarrow u_b \}$, where the distance $d_G(u_a, u_b)$ between $u_a$ and $u_b$ on $G$ is at most four.
    To implement a gate between $u_a$ and $u_b$ on $G$, we can implement $d_G(u_a, u_b)-1$ swap gates moving $u_a$ to a neighbor of $u_b$, apply the desired gate, and implement $d_G(u_a, u_b)-1$ swap gates to move back to $u_a$.
    We can apply the swapping operations in parallel for a set of pairs $u_a \leftrightarrow u_b$ that have no intersection in the shortest path $p(u_a, u_b)$ between $u_a$ and $u_b$.

    For the disjoint pairs of vertices $\{ u_a \leftrightarrow u_b \}$ associated to the depth-$1$ circuit on $G^{(4)}$, we can consider a graph $G^\mathrm{pair}$, where each vertex is a pair $u_a \leftrightarrow u_b$ and an edge between $u_a \leftrightarrow u_b$ and $v_a \leftrightarrow v_b$ is added if their shortest paths intersect, i.e.~if $p(u_a, u_b) \cap p(v_a, v_b) \neq \varnothing$.
    Because $d_G(u_a, u_b) \leq 4$ for each pair $u_a \leftrightarrow u_b$, we can see that if there is an edge between $u_a \leftrightarrow u_b$ and $v_a \leftrightarrow v_b$, then $u_a, u_b, v_a, v_b$ must all be within a constant distance on $G$.
    Since $G$ is a graph with a constant maximum degree, $G^\mathrm{pair}$ is also a graph with a constant maximum degree. Therefore, we can color the graph $G^\mathrm{pair}$ with a constant number of colors.

    The last step is to consider each color on $G^\mathrm{pair}$ individually.
    For all pairs $u_a \leftrightarrow u_b$ with the same color, we can apply the swap networks in parallel to implement the gate between $u_a$ and $u_b$ in constant depth on the circuit geometry $G$.
    Because there are only a constant number of colors, implementing a depth-$1$ circuit on $G^{(4)}$ can be done in constant depth on $G$.
\end{proof}

\subsection{Extension of Theorem~\ref{thm:main-design} to general two-layer circuits}
\label{subsec:ext-to-general-2-layer}

We provide the following extension of Theorem~\ref{thm:main-design}, which applies to general two-layer circuits with gates that sufficiently overlap between the two layers.
Because the proof of Theorem~\ref{thm:main-PRUs} relies essentially on Theorem~\ref{thm:main-design}, we can immediately extend Theorem~\ref{thm:main-PRUs} to general two-layer brickwork ensembles as well.

\begin{theorem}[Gluing small unitary designs with general two-layer circuits] \label{thm:gluing-general-2layer}
    Consider any two-layer brickwork circuit consisting of small random unitaries, in which every qubit is acted on by at least one small random unitary, and each small random unitary is drawn from an $\frac{\varepsilon}{n}$-approximate unitary $k$-design on the qubits that it acts on.
    Let $G$ denote the graph obtained by associating each small random unitary $g$ with a vertex, and adding an edge between the vertex of a first-layer unitary $g_1$, and that of a second-layer unitary $g_2$, if and only if the small random unitaries $g_1$ and $g_2$ overlap on at least $\xi$ qubits.
    Then the random unitary ensemble obtained from the two-layer brickwork circuit is an $\varepsilon$-approximate unitary $k$-design as long as the graph $G$ is connected and $\xi \geq \log_2(2nk^2/\varepsilon)$, $\varepsilon \leq 1$.
\end{theorem}

\begin{proof}
    Let us first describe the manner in which we apply Lemma~\ref{lemma: AB BC to ABC}, and then analyze the error.
    Let $T = (V,E)$ denote any tree sub-graph of $G$ that contains every vertex in $G$; for example, one could obtain $T$ via a depth-first search.
    We apply Lemma~\ref{lemma: AB BC to ABC} once for every edge in $T$.
    For simplicity, we sort these applications, $\{e_1,e_2,\ldots,e_m\}$, in order of the edges' depth in the tree. (The relative order of applications between edges at the same depth does not matter.)
    Here, each edge, $e_\ell = (g_{1,\ell},g_{2,\ell})$, connects a first-layer vertex $g_{1,\ell}$ with a second-layer vertex $g_{2,\ell}$.
    
    At the first application of Lemma~\ref{lemma: AB BC to ABC}, we assign $U_{BC}$ to equal the small unitary $g_{1,1}$, and $U_{AB}$ to equal the small unitary $g_{2,1}$.
    This produces an approximate unitary design on all qubits that are acted on by the small unitaries $g_{1,1}$ and $g_{2,1}$, i.e.~all qubits in the edge $e_1$.    
    We proceed by induction.
    At the $\ell$-th application we assign $U_{BC}$ to equal the small unitary $g_{1,\ell}$ if $g_{1,\ell}$ is not contained in the edges $\{e_1,e_2,\ldots,e_{\ell-1}\}$, and $U_{AB}$ to equal the small unitary $g_{2,\ell}$ otherwise.
    We then assign the other unitary (i.e.~$U_{AB}$ in the first case, $U_{BC}$ in the second case) to equal the product of all small unitaries that are contained in the edges $\{e_1,e_2,\ldots,e_{\ell-1}\}$ traversed thus far.
    By the inductive hypothesis, the latter unitary is drawn from an approximate unitary design on the qubits contained in $\{e_1,e_2,\ldots,e_{\ell-1}\}$.
    Thus, by Lemma~\ref{lemma: AB BC to ABC}, the product of the two random unitaries corresponds to an approximate unitary design on the qubits in $\{e_1,e_2,\ldots,e_{\ell}\}$.
    After the $|E|$-th application, we obtain an approximate unitary design on all qubits.

    Our analysis of the error proceeds nearly identically to the proof of Theorem~\ref{thm:main-design}.
    Let $m = |V| = |E| + 1$.
    We have applied Lemma~\ref{lemma: AB BC to ABC} $m-1$ times to $m$ gates with error $\frac{\varepsilon}{n}$.
    This produces an approximate unitary $k$-design with error
    \begin{equation} \label{eq: total error general}
    (1+\varepsilon/n)^m \left( 1 + \tilde{f}(k,q) \right)^{m-1} - 1 \leq \frac{1}{\log(2)} \left( m \varepsilon / n + (m-1) \tilde{f}(k,q) \right),
    \end{equation}
    where we follow the proof of Theorem~\ref{thm:main-design}, and we abbreviate
    \begin{equation}
    \tilde{f}(k,q) = 2 \left( \frac{k^2}{q} + \frac{k^2}{2q} + \frac{k^4}{2q^2} + \frac{k^2/4q}{1-k^2/4q} \right) \left(1+\frac{k^2}{2q} \right).
    \end{equation}
    The quantity $\tilde{f}(k,q)$ is slightly larger than $f(k,q)$ in Theorem~\ref{thm:main-design} because we make no assumption on the size of the subsystems $A, C$. (We do assume $D_A, D_C \geq 2$, without loss of generality.)
    
    We proceed as in Theorem~\ref{thm:main-design}. We have $k \geq 2$ and $n \geq (m-1)\xi$.
    Combined with the assumptions $q \geq 2nk^2/\varepsilon$, $\varepsilon \leq 1$, these imply $\xi \geq 6$, i.e.~$q \geq 64$.
    From this, the first term in Eq.~(\ref{eq: total error general}) is less than
    \begin{equation}
        \frac{m \varepsilon}{n \log(2)} \leq \frac{\varepsilon}{6 \log(2)}.
    \end{equation}
    Meanwhile, the second term is less than
    \begin{equation}
        e (m-1) \tilde{f}(k,q) \leq \frac{n}{6 \log(2)} \cdot 2  \left( \frac{\varepsilon}{2n} + \frac{\varepsilon}{4n} + \frac{\varepsilon^2}{8n^2} + \frac{\frac{\varepsilon}{4n}}{1-\frac{\varepsilon}{8n}} \right) \left( 1 + \frac{\varepsilon}{4n} \right) \leq  \varepsilon \cdot \frac{2}{6 \log(2)} \\
    \end{equation}
    Taking the sum of the two terms, we have an error $\varepsilon \cdot 1/2\log(2) \leq \varepsilon$ as desired.
\end{proof}

\section{Applications} \label{app: applications}

In this appendix, we present additional details and results on applications of our construction.
The applications are discussed in the same order that they appear in the main text.

\subsection{Classical shadows with 1D log-depth Clifford circuits}\label{app: classical shadows}

We begin with a brief review of classical shadow estimation following Ref.~\cite{huang2020predicting}, and then turn to our results.
Classical shadow estimation seeks to estimate the expectation values of many observables from very few measurements on the quantum state of interest.
The general classical shadows protocol proceeds as follows: We draw a random  unitary $U$ from some ensemble (to be specified), apply it to the state, and measure in the computational basis.
The measurement outputs a bitstring $b$ with probability $p_U(b) = \bra{b} U \rho U^\dagger \ket{b}$.
From this bitstring, one constructs a `classical snapshot' of the density matrix,
\begin{equation} \label{eq: classical snapshot}
	\hat{\rho}_{U,b} = \mathcal{M}^{-1} ( U^\dagger \dyad{b} U ),
\end{equation}
where the linear map $\mathcal{M}$ is chosen, dependent on the ensemble of $U$, such that the expectation value, $\hat{\rho}$, of the snapshot is equal to the actual density matrix $\rho$,
\begin{equation} \label{eq: expected value classical shadow}
	\hat{\rho} = \E_{U} [ \sum_{b \in \{ 0,1 \}^n} p_U(b) \, \hat{\rho}_{U,b} \, ] = 2^n \cdot \E_{U,b} [ \, \bra{b} U \rho U^\dagger \ket{b} \cdot \mathcal{M}^{-1} ( U^\dagger \dyad{b} U ) \, ] = \rho.
\end{equation}
Here, we use $\E_{U,b}[ \cdot ]$ to denote the expectation value when $U$ is drawn from the ensemble and $b$ is drawn from uniformly from the computational basis.
To obtain accurate classical shadow estimates, the procedure above is repeated many times. 
The true density matrix $\rho$ is estimated via the sample mean of the received snapshots, from which estimates of expectation values can also be computed.

The accuracy of classical shadow estimation can be quantified by computing the variance of the classical estimates.
The variance will depend both on the ensemble of $U$ as well as the observable being estimated.
For estimating an expectation value $\tr(\rho O)$ of an observable $O$, one finds a variance
\begin{equation} \label{eq: variance shadow norm}
\begin{split}
	\text{Var} \left[ \tr ( \hat{\rho}_{U,b} O ) \right] & = 2^n   \E_{U,b} [ p_U(b) \cdot \tr ( \hat{\rho}_{U,b} O)^2 ] - 2^{2n}   \E_{U,b} [ p_U(b) \cdot \tr ( \hat{\rho}_{U,b} O) ]^2 \\
	& \leq 2^n \E_{U,b} [ p_U(b) \cdot \tr ( \hat{\rho}_{U,b} O)^2 ] \\
	& \leq \max_\sigma \left( 2^n \E_{U,b} \left[ \bra{b} U \sigma U^\dagger \ket{b} \cdot \bra{b} U \mathcal{M}^{-1}(O) U^\dagger \ket{b}^2 \right] \right) \\
	& \equiv \lVert O \rVert^2_{\text{shadow}}. \\
\end{split}
\end{equation}
In the final line, we define the \emph{shadow norm}, $\lVert O \rVert_{\text{shadow}}$, whose square upper bounds the variance for any quantum state $\rho$.
From this upper bound on the variance, one can form estimates of $O$ using median-of-means estimation that are guaranteed to be accurate with high probability by ~\cite{huang2020predicting}.
As one can see, the shadow norm depends on the choice of unitary ensemble.

Let us now address the ensemble that $U$ is drawn from.
The hallmark application of classical shadows is to estimate fidelities of a quantum state $\rho$ with respect to many target states $\ket{\psi}$.
Naively, this task appears extraordinarly difficult: There are many possible target states $\ket{\psi}$, which may be highly-entangled and may also not share a common eigenbasis.
Nevertheless, Ref.~\cite{huang2020predicting} shows that this estimation can be done efficiently, by drawing $U$ as a random $n$-qubit Clifford unitary.
The estimation relies on two essential properties of the $n$-qubit Clifford group: (i) any Clifford unitary can be simulated classically, and (ii) the Clifford unitary group forms an exact unitary 3-design.
Property (i) guarantees that the snapshots $\hat{\rho_{U,b}}$ can be efficiently estimated and stored by a classical observer.
Property (ii) is used to show that the variances of the shadow estimates are small. 

To elaborate on the latter fact, when $U$ is drawn as an $n$-qubit random Clifford unitary, we choose the inverse map
\begin{equation} \label{eq: M definition}
	\mathcal{M}^{-1}( O ) = (2^n+1) O - \tr(O) \mathbbm{1}.
\end{equation}
The map preserves the operator trace and multiplies traceless operators by a factor $(2^n+1)$.
From this definition, the shadow norm can be bounded as follows~\cite{huang2020predicting},
\begin{equation} \label{eq: exact shadow norm}
	\tr( O_0^2 ) \leq \lVert O \rVert^2_{\text{shadow}} \leq 3 \tr( O_0^2 ),
\end{equation}
where $O_0 = O - \mathbbm{1} \cdot \tr(O)/2^n$ is the traceless part of $O$.
The computation of this bounded uses solely the fact that the Clifford unitaries form an exact unitary 3-design, which allows one to exactly evaluate the expression inside the maximum in Eq.~(\ref{eq: variance shadow norm}) since it contains 3 copies of the unitary $U$.
The shadow norm is small whenever the observable of interest is low-rank.
For example, any rank-1 projector has shadow norm
\begin{equation}
	\left\lVert \dyad{\psi} \right\rVert^2_{\text{shadow}} =  1 - 1/2^n.
\end{equation}
This allows one to accurately estimate many rank-1 expectation values from a moderate number of classical snapshots~\cite{huang2020predicting}.

Despite this promise, a key drawback of the above approach with respect to near-term implementations is the depth required to implement the random Clifford unitary $U$.
A random Clifford unitary requires linear depth in $n$ to implement from two-qubit unitary gates (for example, in 1D, or in any geometry without ancilla qubits).
However, in near-term devices, such deep circuits will be accumulate many experimental errors, which will cause the fidelity of classical shadow estimation to decay exponentially in the depth of the circuit.
For this reason, recent works have considered the possibility of lower-depth implementations of classical shadows~\cite{akhtar2023scalable,bertoni2022shallow,ippoliti2023operator}.

Let us now show how our construction of low-depth approximate unitary designs can be used to perform provably-efficient classical shadows with low-depth Clifford circuits.
To do so, we simply replace the random Clifford unitary $U$ (which forms an exact unitary 3-design) with a low-depth random Clifford circuit from our construction (which forms an $\varepsilon$-approximate unitary 3-design).
This reduces the depth of the unitary from linear in $n$ to merely logarithmic in $n$.
It remains only to analyze how the approximation error $\varepsilon$ affects the classical shadow protocol.

At a high level, the use of an approximate design will have two effects, both of which are controlled by the magnitude of the approximation error $\varepsilon$.
We note that  $\varepsilon$ can be decreased exponentially by increasing the depth of the circuit; see e.g.~Corollary~\ref{cor: upper bound design} of the main text.
First, our bound on the shadow norm, and thus the variance of the estimate, will slightly increase, by an amount proportional to $\varepsilon$.
Second, the approximation error can lead to a small \emph{bias} in the  estimation, since the map $\mathcal{M}^{-1}$ that inverts the twirl over an exact 3-design [Eq.~(\ref{eq: M definition})] will only approximately invert the twirl over an approximate 3-design [Eq.~(\ref{eq: expected value classical shadow})].
If desired, this bias can be eliminated by choosing a new map, $\tilde{\mathcal{M}}$, that exactly inverts the twirl over the approximate design.
It is not clear that such an inverse can be found efficiently, but heuristic algorithms for finding this inverse map were deviced in Refs.~\cite{akhtar2023scalable,bertoni2022shallow}.
For brevity, we do not do so here, and we instead simply show that the bias in the estimate can easily be made very small.

We summarize our results in the following lemma.
\begin{theorem} \label{thm: approximate classical shadows}
\emph{(Classical shadows with 1D log-depth Clifford circuits)}
Consider the classical shadows protocol, where the random unitary $U$ is drawn from any $\varepsilon$-approximate unitary 3-design.
Any positive observable $O$ can be estimated with variance
\begin{equation}
	\text{\emph{Var}} \left[ \tr ( \hat{\rho}_{U,b} O ) \right] \leq \lVert O \rVert_{\text{shadow}}^2 \leq 3\tr( O_0^2 ) + 10 \varepsilon \tr( O )^2,
\end{equation}
and bias,
\begin{equation}
	\left| \tr( O \hat{\rho} ) - \tr( O \rho ) \right| \leq 2 \varepsilon \tr(O),
\end{equation}
where $O_0 = O - \mathbbm{1} \cdot \tr(O) / d$ is the traceless part of $O$.
Applying Corollary~\ref{cor: upper bound design}, this allows provably accurate classical shadow estimation with Clifford circuits of depth $\mathcal{O}(\log(n/\varepsilon))$ in 1D and depth $\mathcal{O}(\log \log(n/\varepsilon))$ in all-to-all-connected circuits.
\end{theorem}
\noindent We emphasize that in the regime of interest where $O$ is low-rank, the additional terms proportional to $\varepsilon$ are sub-leading compared to the variance term that is independent of $\varepsilon$, i.e.~the variance that arises even when using a deep random Clifford unitary.
Thus, the performance of classical shadow estimation with log-depth random Clifford circuits is essentially equivalent to that with linear-depth random Clifford unitaries.

\begin{proof}
Let $\mathcal{E}$ denote the approximate unitary 3-design. We bound the bias of the estimator as follows,
\begin{equation}
\begin{split}
	\left| \tr( O \hat{\rho} ) - \tr( O \rho ) \right| & = \left| 2^n  \E_{U,b} [ \, \bra{b} U \rho U^\dagger \ket{b}  \bra{b} U \mathcal{M}^{-1} ( O ) U^\dagger \ket{b} \, ]  - \tr( O \rho ) \right| \\
         & = \left| 2^n (2^n+1)  \E_{U,b} [ \, \bra{b} U \rho U^\dagger \ket{b}  \bra{b} U O U^\dagger \ket{b} \, ] - 2^n \E_{U,b} [ \, \bra{b} U \rho U^\dagger \ket{b}] \tr(O)  - \tr( O \rho ) \right| \\
         & = \left| 2^n (2^n+1)  \tr \big( \E_{U,b} [ ( U^\dagger \dyad{b} U )^{\otimes 2} ]  \rho \otimes O \big) - \tr(\rho) \tr(O)  - \tr( O \rho ) \right| \\
         & = 2^n (2^n+1) \left|  \tr \big( \Phi_\mathcal{E}( \E_b [\dyad{b}^{\otimes 2}] )  \rho \otimes O \big) - \tr \big( \Phi_H( \E_b [\dyad{b}^{\otimes 2}] )  \rho \otimes O \big) \right| \\
         & \leq \varepsilon 2^n (2^n+1) \cdot \tr \big( \Phi_H( \E_b [\dyad{b}^{\otimes 2}] )  \rho \otimes O \big) \\
         & = \varepsilon \left( \tr(\rho) \tr(O) + \tr( O \rho ) \right) \\
         & \leq 2 \varepsilon \tr(O). \\
\end{split}
\end{equation}
In going from the third to fourth line, and the fifth to sixth line, we use the identity $2^n (2^n+1) \tr \big( \Phi_H ( \E_b [\dyad{b}^{\otimes 2}] ) \rho \otimes O \big) = \tr(\rho) \tr(O) + \tr(O\rho)$.
This identity must hold in order for the bias to be zero when using an exact unitary 2-design, i.e.~when $\mathcal{E} = H$.
In going from the fourth to fifth line, we apply the definition of an $\varepsilon$-approximate unitary 2-design.

We now turn to the variance, which is bounded by the shadow norm via Eq.~(\ref{eq: variance shadow norm}).
Note that the variance of our estimate of $O$ is the same as that of $O_0$, since they differ only by an identity component.
We have
\begin{equation}
\begin{split}
	\text{Var} \left[ \tr ( \hat{\rho}_{U,b} O ) \right] & = \text{Var} \left[ \tr ( \hat{\rho}_{U,b} O_0 ) \right] \\
    & \leq \left\lVert O_0 \right\rVert_{\text{shadow}}^2  \\
	& = \max_\sigma \left( 2^n \E_{U,b} \left[ \bra{b} U \sigma U^\dagger \ket{b}  \bra{b} U \mathcal{M}^{-1}(O_0)  U^\dagger \ket{b}^2 \right] \right) \\
	& = 2^n (2^n+1)^2  \max_\sigma \tr( \E_{U,b} \left[ (U \dyad{b} U^\dagger)^{\otimes 3} \right] \sigma \otimes O_0 \otimes O_0 ).
\end{split}
\end{equation}
To bound this expression, we replace $O_0 = O - \mathbbm{1} \cdot \tr(O)/2^n$ and use the fact that $\mathcal{E}$ is an $\varepsilon$-approximate 3-design. 
We have
\begin{align}
& 2^n (2^n+1)^2 \tr( \E_{U,b} \left[ (U \dyad{b} U^\dagger)^{\otimes 3} \right] \sigma \otimes O_0 \otimes O_0 ) \\
& = 2^n (2^n+1)^2 \bigg[ \tr( \E_{U,b} \left[ (U \dyad{b} U^\dagger)^{\otimes 3} \right] \sigma \otimes O \otimes O ) -   \frac{2 \tr(O)}{2^n}  \tr( \E_{U,b} \left[ (U \dyad{b} U^\dagger)^{\otimes 2} \right] \sigma \otimes O ) \nonumber \\
   & \,\,\,\,\,\,\,\,\,\,\,\,\,\,\, + \frac{\tr(O)^2}{2^{2n}}  \tr( \E_{U,b} \left[ (U \dyad{b} U^\dagger) \right] \sigma ) \bigg] \\
& \leq 3 \tr( O_0^2 ) +  \varepsilon \, 2^n (2^n+1)^2 \bigg[ \tr( \Phi_H( \E_b[ \dyad{b}^{\otimes 3} ] ) \sigma \otimes O \otimes O )  + \frac{2 \tr(O)}{2^n} \tr( \Phi_H( \E_b[ \dyad{b}^{\otimes 2} ] ) \sigma \otimes O ) \bigg] \nonumber \\
& \leq 3 \tr( O_0^2 ) +  \varepsilon \, 2^n (2^n+1)^2 \bigg[ \tr( O )^2 \lVert \Phi_H( \E_b[ \dyad{b}^{\otimes 3} ] ) \rVert_\infty  + \frac{2 \tr(O)^2}{2^n} \lVert \Phi_H( \E_b[ \dyad{b}^{\otimes 2} ] ) \rVert_\infty \bigg] \\
& \leq 3 \tr( O_0^2 ) +  \varepsilon 2^n (2^n+1)^2 \left( \frac{6 \tr(O)^2}{2^n(2^n+1)(2^n+2)} + \frac{4 \tr(O)^2}{2^{2n}(2^n+1)} \right) \\
& \leq 3 \tr( O_0^2 ) +  10 \varepsilon \tr(O)^2.
\end{align}
The first term in the third line corresponds to the variance when $U$ is drawn from an exact unitary 3-design.
The latter terms, proportional to $\varepsilon$, correspond to the absolute values of the first two terms in the second line when $U$ is drawn from an exact unitary 3-design.
From the definition of approximate designs, these upper bound the errors in the approximation. 
We note that the third term in the second line involves the twirl over an exact 1-design, and so incurs no error.
In going from the third to fourth line, we upper bound each trace via Holder's inequality, and use that $\lVert \sigma \rVert_1 = 1$ and $\lVert O \rVert_1 = \tr(O)$.
From the fourth to fifth line, we upper bound the spectral norm of the Haar twirls by using $\Phi_H( \dyad{b}^{\otimes 3} ) = \frac{1}{2^n(2^n+1)(2^n+2)} \sum_{\tau \in S_3} \tau$ and $\Phi_H( \dyad{b}^{\otimes 2} ) = \frac{1}{2^n(2^n+1)} \sum_{\tau \in S_2} \tau$ and noting that each permutation operator has spectral norm 1. Note that there are $3! = 6$ permutation operators in the first sum and $2! = 2$ permutation operators in the second sum.
This concludes our proof.
\end{proof}

\subsection{Proof of Corollary~\ref{cor: topological}: Quantum hardness of recognizing topological order}\label{app: topological order}

Topological order is a property of many-body quantum states that feature exotic patterns of entanglement between local degrees of freedom~\cite{wen2004quantum}.
In principle, there are many different inequivalent classes of topological order, such as trivial order, toric code topological order, double-semion topological order, D4 topological order, etc.
Sparked by recent experimental progress~\cite{clark2020observation,semeghini2021probing, satzinger2021realizing,leonard2023realization,iqbal2024topological,iqbal2024non}, a flurry of works have asked: Given access to a quantum state $\rho$, how can one efficiently recognize the topological order of the state~\cite{jiang2012identifying,cong2019quantum,rodriguez2019identifying,cian2021many,herrmann2022realizing,huang2022provably,cian2022extracting,lake2022exact,bouland2023public,cong2024enhancing}?
However, despite a range of heuristic approaches, rigorous theoretical results on this question are lacking.

Our construction of low-depth pseudorandom unitaries implies that recognizing topological order is, in fact, quantum computationally hard.
To show this, we use a defining property of topological order: invariance under the application of low-depth geometrically-local unitary circuits~\cite{chen2010local, wen2017colloquium, zeng2019quantum}.
This property allows us to create pseudorandom states within any class of topological order, by applying a low-depth pseudorandom unitary circuit $U$ to a representative state $\ket{\psi}$ from the class.
Since two pseudorandom states with different topological order cannot be distinguished from one another by any polynomial-time quantum algorithm, it follows that recognizing topological order is hard.

We remark that, if the representative state $\ket{\psi}$ is the ground state of a Hamiltonian $H$, one can view the low-depth pseudorandom unitary circuit as transforming the \emph{entire} Hamiltonian and all of its eigenstates as well.
That is, we map $H \rightarrow U H U^\dagger$ and $\ket{\psi_\lambda} \rightarrow U \ket{\psi_\lambda}$ for every eigenstate $\ket{\psi_\lambda}$ of $H$.
If $H$ is geometrically $r$-local, this produces a geometrically $\mathcal{O}(r\xi)$-local Hamiltonian.
The important point is that the pseudorandom states that we create, $U \ket{\psi}$, are ground states of geometrically-local Hamiltonians, $U H U^\dagger$, which have an identical ground and excited state structure to the original Hamiltonian $H$.
We discuss this further in the context of several precise definitions of topological order later on.

Let us now provide a short formal proof of Corollary~\ref{cor: topological} in the main text.
The corollary concerns the simplest possible relevant task: distinguishing whether a state has trivial order (as does the product state) or non-trivial topological order (as does the toric code state).
The corollary can also, trivially, be extended to the task of distinguishing any two classes of topological order, provided the existence of a representative state $\ket{\psi}$ from each class.

\begin{proof}[Proof of Corollary~\ref{cor: topological}]
We proceed by contradiction. Consider a binary classification task in which one is given many copies of a random quantum state, which is drawn from one of two ensembles with equal probability.
The first ensemble, $\mathcal{E}_{\text{triv}}$, is obtained by applying a low-depth pseudorandom unitary to the zero state, $\mathcal{E}_{\text{triv}} = \{ U \ket{0^n} | \, U \sim \mathcal{E} \}$.
Here, $\mathcal{E}$ is the pseudorandom two-layer brickwork ensemble in Corollary~\ref{cor:pseudorandom unitaries} with $\xi = \omega(\log n)$.
The second ensemble, $\mathcal{E}_{\text{topo}}$, is obtained similarly,  applying a low-depth pseudorandom unitary to the toric code ground state.
The goal of the task is to distinguish which ensemble the state was drawn from.

By definition, all of the states in $\mathcal{E}_{\text{triv}}$ have trivial topological order, while all of the states in $\mathcal{E}_{\text{topo}}$ have non-trivial topological order.
Thus, if one could recognize topological order, the task would be solved.
However, since the applied unitaries are pseudorandom, no polynomial-time quantum computation can distinguish either ensemble from the Haar-random ensemble of states. 
This implies that the two ensembles also cannot be distinguished from each other.
Hence topological order cannot be recognized in polynomial time. 
\end{proof}

We conclude with a brief discussion of several of the most common definitions of non-trivial topological order, and how they apply to the pseudorandom topologically-ordered quantum states that we construct.
\begin{itemize}
    \item \textbf{Ground state degeneracy \& Local invisibility: } Hamiltonians with non-trivial topological order have degenerate ground states on certain geometries~\cite{wen2004quantum,zeng2019quantum}. A key property of topological order is that this degeneracy is invisible to local operators~\cite{bravyi2006lieb,bravyi2010topological,hastings2010locality}. We say that a ground state manifold has $(l,\epsilon)$-topological order if the matrix $M_{ij} = \bra{\psi_j} O \ket{\psi_i}$ is $\epsilon$-close to a multiple of the identity for any operator $O$ of diameter $l$~\cite{hastings2010locality}. Typically, one wants this condition to hold for some $l$ scaling with the diameter of entire $n$-qubit system, $l = \Theta(\sqrt{n})$.
    
    The toric code Hamiltonian has ground state degeneracy 4 when placed on a torus. The Hamiltonian from our construction does as well, since it has an identical spectrum.
    
    The toric code ground states have $(\sqrt{n}/2,0)$-topological order. The pseudorandom toric code ground states that we create have $(\sqrt{n}/2-4\sqrt{2 \xi},0)$-topological order. This is identical to that of the toric code states up to a small sub-leading term,  $\mathcal{O}(\sqrt{\omega(\log n)})$.
    
    \item \textbf{Far from product state under local circuits: } As previously noted, another common definition of topological order says that a state has non-trivial topological order if it cannot be mapped to a product state by any low-depth unitary circuit. This condition is typically combined with a further restriction, that the state should be the ground state of a geometrically-local Hamiltonian~\cite{chen2010local, wen2017colloquium, zeng2019quantum}.

    The toric code state cannot be mapped to a product state by any local circuit of depth less than $\sqrt{n}/4$. Similarly, the pseudorandom toric code states that we create cannot be mapped to a product state by any local circuit of depth less than $\sqrt{n}/4 -2\sqrt{2 \xi}$. (This follows from simple light-cone arguments.) Again, the depth is identical between the two cases up to a small sub-leading term.

    Moreover, the states we create are ground states of geometrically $\mathcal{O}(\xi)$-local Hamiltonians.

    \item \textbf{Fractionalized quasiparticles: } A final, more physical criteria for topological order says that a Hamiltonian whose ground state has topological order should feature fractionalized quasiparticle excitations~\cite{wen2004quantum,zeng2019quantum}.
    In 2D, these quasiparticle excitations are generated by one-dimensional ``string'' operators $W_{i \rightarrow j}$, which create excitations with $\mathcal{O}(1)$ energy at their ends $i$ and $j$, but do not excite the Hamiltonian in their bulk.
    The excitations are fractionalized in the sense that an individual excitation cannot be created or destroyed in isolation. For example, in the toric code Hamiltonian, excitations can only be created and destroyed in pairs.

    The toric code Hamiltonian features string operators, $W^{TC}_{i \rightarrow j}$, with width $1$ in the transverse direction.
    The pseudorandom toric code Hamiltonians that we create feature ``fattened'' string operators, $U W^{TC}_{i \rightarrow j} U^\dagger$, with width $2\xi$ in the transverse direction.
    At all length scales greater than $2\xi$, the two sets of string operators behave identically.
    
\end{itemize}

\subsection{Quantum advantage in learning low-complexity physical systems}\label{app: learning advantages}

The low-depth construction of pseudorandom unitaries given in Corollary~\ref{cor:pseudorandom unitaries} implies that several large separations \cite{chen2021exponential, huang2021quantum} between learning with and without an external quantum memory (or, said differently, between classical and quantum agents) hold even for geometrically-local physical systems with low circuit complexity. Here, we consider circuit complexity to be defined in terms of the minimum depth among all geometrically-local circuits that generate the physical system.

\subsubsection{Distinguishing random pure states from maximally mixed states}

A prototypical example of such a learning separation is the task of distinguishing a Haar-random pure state $\ket{\psi}$ from the maximally mixed state $\mathbbm{1}/2^n$ \cite{chen2021exponential, huang2021quantum}.
With access to a quantum memory and therefore entangled measurements, a quantum agent can efficiently solve this problem by performing a SWAP test on $\mathcal{O}(1)$ copies to estimate the purity, $\tr(\rho^2)$, up to constant accuracy.
At the same time, Refs.~\cite{chen2021exponential, huang2021quantum} proved that a classical agent, i.e.~one with access only to single-copy measurements, requires $\Theta(2^{n/2})$ measurements on $\rho$.
By replacing the Haar-random pure state $\ket{\psi}$ with a pseudorandom state generated by our geometrically-local $\poly \log(n)$-depth circuits secure against $\poly(n)$-time adversaries (Corollary~\ref{cor:pseudorandom unitaries}), we have that a classical agent requires a superpolynomial time to solve this distinguishing task while quantum agents remain computationally efficient.
This establishes a superpolynomial quantum computational advantage in learning low-complexity physical systems.
We note that Ref.~\cite{huang2021quantum} has also considered using pseudorandom states~\cite{ji2018pseudorandom} to instantiate Haar-random states in a realistic setting.
However, before our work, all known constructions of pseudorandom states required circuit depth at least linear in $n$ on geometrically-local circuit architectures.
Our low-depth PRUs show that the quantum advantage already appears in physical systems with exponentially lower geometrically-local circuit depth.

\subsubsection{Distinguishing random unitary channels from completely depolarizing channels}

Another task with an exponential separation between classical and quantum learning agents is that of distinguishing the unitary channel, $U(\cdot)U^{\dagger}$, for an $n$-qubit Haar random unitary $U$ from an $n$-qubit completely depolarizing channel.
Refs.~\cite{huang2021quantum,chen2021exponential,aharonov2021quantum} proved that $2^{\Omega(n)}$ experiments involving the unknown channel are required by any classical agent.
Here, each experiment conducted by the classical agent involves state preparation, a single evolution under the unknown channel, and measurement on the output state.
In contrast, a quantum agent can coherently retrieve and store the output state of the unknown channel in a quantum memory in each experiment.
Again, by utilizing SWAP tests, quantum agents can distinguish the two cases using $\mathcal{O}(1)$ experiments involving the unknown channel.
By instantiating the Haar-random unitary channel with the geometrically-local $\poly \log(n)$-depth circuits secure against $\poly(n)$-time adversaries in Corollary~\ref{cor:pseudorandom unitaries}, one finds that there is also a superpolynomial quantum computational advantage for distinguishing between a random $\poly \log n$ geometrically-local quantum circuit and a completely depolarizing channel.

\subsubsection{Distinguishing entangled states from unentangled states}

A similar argument shows a superpolynomial quantum computational advantage in learning the entanglement structure of a low-complexity physical system.
More concretely, let us partition $n$ qubits into two subsets $A$ and $B$ such that $|A|=|B|=n/2$.
We consider the following two settings: (1) The unknown $n$-qubit state $|\psi\rangle$ is entangled across the two subsystems $A$ and $B$, or (2) $|\psi\rangle= |\psi_A\rangle\otimes |\psi_B\rangle$ are unentangled across $A$ and $B$.
With access to an external quantum memory, a quantum agent can simply discard subsystem $B$ and perform a SWAP test on the reduced density matrix, $\mathrm{Tr}_A[|\psi\rangle\langle\psi|]$.
Meanwhile, when $\ket{\psi}, \ket{\psi_A}, \ket{\psi_B}$ are Haar-random, a classical agent that can perform any single-copy measurement cannot efficiently distinguish $|\psi_A\rangle\otimes |\psi_B\rangle$ from $\mathbbm{1}/2^{n}=\mathbbm{1}_A/2^{n/2}\otimes \mathbbm{1}_B/2^{n/2}$.
Otherwise, a classical agent could efficiently distinguish a Haar-random state, $|\psi_A\rangle$, from the maximally mixed state, $\mathbbm{1}_A/2^{n/2}$, by preparing another state, $\ket{\psi_B}$, which we have seen in the previous section is not true.
Since $\mathbbm{1}/2^n$ and $|\psi\rangle$ are indistinguishable by classical agents with a subexponential number of measurements, so are (1) and (2).
Again, our low-depth PRU construction from Corollary~\ref{cor:pseudorandom unitaries} allows us to instantiate Haar-random states $\ket{\psi}, \ket{\psi_A}, \ket{\psi_B}$ using pseudorandom states generated from circuits of depth $\poly \log n$ in 1D.
Hence, there is a superpolynomial quantum advantage in distinguishing between entangled and unentangled low-complexity physical systems.

\subsection{The power of time-reversal in learning}\label{app: time-reversal}

In conventional quantum dynamics experiments, one characterizes properties of a unitary $U$ through experiments that apply $U$ to a quantum state.
These include experiments that apply $U$ a single time and then measure the state, as well as experiments that apply $U$ several times with controlled operations in between [Fig.~\ref{fig: setup}(b)].
More recently, an array of quantum experiments have begun to break this paradigm~\cite{baum1985multiple,choi2017dynamical,garttner2017measuring,davis2019photon,mi2021information,braumuller2022probing,sanchez2021emergent,dominguez2021decoherence,colombo2022time,li2023improving}.
These experiments have the power to apply both a unitary $U$ and its inverse $U^\dagger$, and thus reverse time for the system under study.
Time-reversal experiments are by now widespread in many quantum technologies, including superconducting~\cite{mi2021information,braumuller2022probing}, trapped ion~\cite{garttner2017measuring}, nuclear magnetic resonance~\cite{baum1985multiple,sanchez2021emergent,dominguez2021decoherence}, and cold atom~\cite{colombo2022time,li2023improving} platforms.
In nearly every case, time-reversal is achieved solely through knowledge of the \emph{type} of interaction present in the system (e.g.~dipolar interactions, or a Raman-induced transition), with no dependence on the connectivity or microscopic Hamiltonian (see within Ref.~\cite{schuster2023learning} for a review).

Here, we rigorously investigate the power of such experiments~\cite{cotler2023information}.
We consider two classes of quantum experiments to characterize properties of a unitary $U$.
The first class, \emph{time-forward} quantum experiments, encapsulate any experiment that queries $U$, possibly interspersed with polynomial-time quantum circuits.
The second class, \emph{time-reversal} quantum experiments, corresponds to any experiment that queries both $U$ and $U^\dagger$.
These correspond to the ``time-ordered'' and ``out-of-time-order'' experiments introduced in Ref.~\cite{cotler2023information}, respectively.

We now show that time-reversal quantum experiments can learn physically-relevant features of quantum dynamics exponentially more efficiently than time-forward experiments.
To do so, we formalize the example introduced in the main text as follows.
We consider the task of distinguishing between two unitary circuits.
In the first, the unitary $U_{2D}$ is any local 2D circuit of depth $d$, applied to a $\sqrt{n} \times \sqrt{n}$ array of $n$ qubits.
In the second, the unitary $U_{2D,LR}$ is any local 2D circuit augmented with long-range couplings of strength $\theta$.
Specifically, we take
\begin{equation} \label{eq: U 2D LR}
    U_{2D,LR} = U_{2D}' U_{LR},
\end{equation}
where $U_{2D}'$ is any local 2D circuit, and
\begin{equation}
    U_{LR} = \prod_{(i,j(i))} e^{-i \theta Z_i Z_{j(i)}}
\end{equation}
implements a long-range $ZZ$-rotation of strength $\theta$ between each qubit $i$ and a single additional qubit $j(i)$ a distance at least $\sqrt{n}/2$ away from $i$.
We choose this simple long-range interaction merely to have a specific example which is easy to analyze; the precise nature of the couplings will not qualitatively change our result.

We prove the following super-polynomial separation between time-reversal and time-forward experiments.
\begin{theorem}
    Consider the task of distinguishing a 2D local unitary circuit of depth $d$, $U_{2D}$, from a 2D local unitary circuit of depth $d$ augmented by long-range couplings, $U_{2D,LR}$.
    A time-forward experiment cannot solve this task in polynomial-time whenever $d = \Omega(\emph{poly} \log n)$.
    Meanwhile, a time-reversal experiment can solve this task for any unitary circuits in time $\mathcal{O}(\theta^{-2})$ whenever $d < \sqrt{n}/2$.
\end{theorem}

\begin{proof}
    The first statement of the theorem follows directly from our results.
    We suppose that the 2D local unitary circuit $U_{2D}$ is drawn from the 2D brickwork ensemble in Fig.~\ref{fig: proof overview}(c), where each small random unitary is pseudorandom.
    From Corollary~\ref{cor:pseudorandom unitaries}, this implies that $U_{2D}$ is pseudorandom and thus cannot be distinguished from a Haar-random unitary in polynomial time.
    Similarly, we suppose that $U_{2D}'$ in Eq.~(\ref{eq: U 2D LR}) is drawn from the same ensemble and thus is also pseudorandom.
    Since $U_{LR}$ is fixed and has circuit depth one, this implies that $U_{2D,LR}$ is also pseudorandom.
    Thus, $U_{2D,LR}$ cannot be distinguished from a Haar-random unitary in polynomial time, and hence $U_{2D}$ and $U_{2D,LR}$ cannot be distinguished from one another in polynomial time.

    To establish the second statement of the theorem, we consider the following experimental protocol~\cite{cotler2023information,schuster2023learning}: One first prepares a random stabilizer product state $\ket{v}$, then applies the time-reversed unitary $U^\dagger$, then perturbs a single qubit $i$ via a spin flip $X_i$, then applies the unitary $U$, and then measures each qubit in the same basis as $\ket{v}$.
    Each run of the protocol produces a single output bitstring drawn from the probability distribution
    \begin{equation}
        q_{v,U}(x) = \big| \! \bra{x} v \, U X_i U^\dagger v^\dagger \ket{0^n} \! \big|^2,
    \end{equation}
    where $x \in \{0,1\}^n$ and we denote $\ket{v} = v^\dagger \ket{0^n}$, where $v$ is a tensor product of random single-qubit Clifford unitaries.
    We assume the protocol run $M$ times, each time with an independently random $\ket{v}$.
    
    When $U = U_{2D}$ is a local 2D circuit, the operator $U X_i U^\dagger$ can only have support within a distance $d$ of the qubit $i$.
    This implies that the output bitstrings will always be $0$ on sites that are at least a distance $d$ from $i$.

    We will now show that when $U = U_{2D,LR}$ contains long-range couplings, the output bitstrings contain at least one $1$ on sites far from $i$ with high probability.
    Thus, one can distinguish $U_{2D}$ from $U_{2D,LR}$ simply by viewing whether a single faraway $1$ was received.
    This follows because the operator $U X_i U^\dagger$ will inevitably gain support on faraway qubits.
    In particular, we have
    \begin{equation}
        U X_i U^\dagger = U_{2D} U_{LR} X_i U_{LR}^\dagger U_{2D}^\dagger = U_{2D} \left( \cos \theta \, X_i + \sin \theta \, Y_i Z_{j(i)} \right) U_{2D}^\dagger.
    \end{equation}
    Since $U_{2D}$ is a local 2D circuit, the operators $U_{2D} X_i U_{2D}^\dagger$ and $U_{2D} Y_i U_{2D}^\dagger$ only have support within distance $d$ of $i$, and the operator $U_{2D} Z_{j(i)} U_{2D}^\dagger$ only has support within distance $d$ of $j(i)$.
    By assumption, we have $d < \sqrt{n}/2$ and $\text{dist}(i,j) = \sqrt{n}/2$, so these two regions of support are non-overlapping.

    To establish our claim, let us analyze the probability distribution $q_{v,U}(x)$ on the faraway qubits by expanding $U X_i U^\dagger$ in the Pauli basis,
    \begin{equation}
        U X_i U^\dagger = \sum_P c_P P,
    \end{equation}
    where $\sum_P |c_P|^2 = 1$.
    Applying standard formulas for the twirl over a tensor product of single-qubit Clifford unitaries~\cite{elben2023randomized}, we have
    \begin{equation}
        \mathbbm{E}_v \left[ q_{v,U}(x) \right] = \sum_P | c_P |^2 \cdot \prod_i \bigg( \delta_{x_i,0} \cdot \left( \delta_{P_i = \mathbbm{1}}  + \frac{1}{3} \delta_{P_i \neq \mathbbm{1}} \right) + \delta_{x_i,1} \cdot \frac{2}{3} \delta_{P_i \neq \mathbbm{1}} \bigg),
    \end{equation}
    where $P_i \in \{ \mathbbm{1},X,Y,Z\}$ denotes the support of $P$ on qubit $i$.
    Intuitively, a Pauli operator $P$ cannot flip bit $i$ if it has identity support on $i$; hence, in this case, we receive the output $x_i = 0$.
    On the other hand, if $P$ has non-identity support on $i$, then $P$ anti-commutes with the Clifford measurement basis on $i$ with probability $2/3$; hence, the output is flipped to $x_i=1$ with probability $2/3$ and left at $x_i=0$ with probability $1/3$.
    The total expected output distribution is a mixture of the distributions obtained for each individual Pauli operator $P$.

    We can now compute the probability that even a single $1$ appears faraway from $i$.
    In a single run of the experiment, this event occurs with probability
    \begin{equation}
        P(\text{any faraway 1's}) = \sum_P | c_P |^2 \cdot \left( 1 - \left( \frac{1}{3} \right)^{w[P_f]} \right) \geq \frac{2}{3} \sum_P | c_P |^2 \cdot  \delta_{w[P_f] > 0},
    \end{equation}
    where $P_f$ denotes the restriction of $P$ to qubits farther than distance $\xi$ from $i$, and $w[P_f]$ denotes the number of non-identity elements of $P_f$.
    On the right-hand side, we lower bound the sum by the probability that $P_f$ has any non-identity support.

    We finally arrive at our desired result.
    The probability that a Pauli $P$ drawn $U X_i U^\dagger$ will have support faraway from $i$ is simply given by $\sin^2 \theta$,
    \begin{equation}
        \sum_P | c_P |^2 \cdot \delta_{w[P_f] > 0} = \sin^2 \theta,
    \end{equation}
    since $U_{2D} X_i U_{2D}^\dagger$, $U_{2D} Y_i U_{2D}^\dagger$ never have support far from $i$ and $U_{2D} Z_{j(i)} U_{2D}^\dagger$ always has support far from $i$.
    Thus, each run of the protocol produces at least one faraway $1$ with probability $(2/3) \sin^2(\theta)$.
    Repeated $M$ times, the experiment will output at least one faraway $1$ with probability $1-\delta \geq 1-( 1- (2/3) \sin^2(\theta) )^M$.
    Thus, to distinguish the two unitaries with success probability $1-\delta$, it suffices to repeat the protocol $M = \mathcal{O}(\theta^{-2} \log(1/\delta))$ times.
    Setting $\delta^{-1} = \mathcal{O}(1)$ gives our stated result.
\end{proof}

\subsection{Output distributions of shallow quantum circuits} \label{app: output distributions}

We consider the output distribution when a random quantum circuit $U$ is applied to the zero state and measured in the computational basis,
\begin{equation}
    p_U(x)=|\langle x|U|0^n\rangle|^2,
\end{equation}
where $p_U(x)$ is a probability and $x \in \{0,1\}^n$.
Sampling from precisely such output distributions is the foundation of many quantum supremacy experiments~\cite{arute2019quantum}.

The output distributions of local random circuits are interesting objects.
In order for an output distribution to be hard to sample from, it must be sufficiently far in distance from the uniform distribution.
Otherwise, one could simply sample from the uniform distribution to mimic the output.
In keeping with this intuition, the output distributions of local random circuits are known to be far-from-uniform in total variation distance with probability $1-\mathrm{negl}(n)$ after depth $\mathcal{O}(n)$~\cite{aaronson2016complexity,nietner2023average}.
At the same time, the bitstrings outputted by random quantum circuits tend to \emph{feel} similar to uniformly random bitstrings.
For example, this underlies a curious property of sampling-based quantum supremacy experiments, in which all known methods to verify that one has sampled from the correct output distribution require a computationally-hard classical simulation of the circuit of interest.
Without such a simulation, the received bitstrings simply appear random.

In the following, we prove that both of these intuitions are accurate, even for local random circuits of very low depth.

\begin{theorem}[Output distributions of shallow local random quantum circuits] \label{thm:dist_1D_distributions}
The output distributions of local random circuits on any geometry with depth $\Omega(\log n)$ are far-from-uniform with probability approaching one.
At the same time, under the cryptographic assumption that no quantum algorithm can solve LWE in subexponential time, no $\mathrm{poly}(n)$-time quantum algorithm can distinguish the output distribution of local random circuits with depth $\emph{poly} \log n$ from the uniform distribution, while using only the sampled bitstrings.
\end{theorem}
\begin{proof}
We will utilize our construction of low-depth approximate unitary designs to prove the first statement of the theorem, and our construction of low-depth pseudorandom unitaries to prove the second statement. 
If desired, both statements can be proven for the \emph{same}  random circuit ensemble by defining a local random circuit $U = U_{\text{pseudo}} U_{\text{design}}$ where $U_{\text{pseudo}}$ is a $\text{poly} \log n$ depth pseudorandom unitary and $U_{\text{design}}$ is a $\log n$ depth design.
This circuit ensemble has depth $\text{poly} \log n$ and is both pseudorandom and an approximate unitary $k$-design for any $k = \mathcal{O}(1)$.

The first statement of the theorem is proven in Proposition~\ref{prop:far-from-uniform} below.
The statement holds for any approximate unitary $8$-design with relative error.
From Corollary~\ref{cor: upper bound design}, this includes shallow random quantum circuits of depth $\omega(\log n)$.

The second statement of the theorem follows from Corollary~\ref{cor:pseudorandom unitaries} and Lemma~\ref{lem:dist-Haar-unif}.
The lemma shows that the output distributions of Haar-random circuits are computationally indistinguishable from uniform, and is proven below.
We consider the pseudorandom unitary circuits with $\mathrm{poly}\log (n)$ depth constructed in Appendix~\ref{app: proof of main PRUs}.
The statement follows because no polynomial-time quantum algorithm can distinguish such pseudorandom circuits from Haar-random unitaries, and no algorithm can distinguish the output distributions of Haar-random unitaries from uniform.
Thus, no polynomial-time quantum algorithm can distinguish the output distributions of $\text{poly} \log n$ depth random circuits from uniform.
\end{proof}

Let us begin by proving Lemma~\ref{lem:dist-Haar-unif}, which lower bounds the sample complexity needed to distinguish the output distribution of a Haar-random unitary from the uniform distribution.

\begin{lemma}[Distinguishing the output distribution of a Haar-random unitary from uniform] \label{lem:dist-Haar-unif}
Given $N$ samples $x_1, \ldots, x_N \in \{0, 1\}^n$ drawn from either the uniform distribution over $n$ bits (Unif) or the output distribution after measuring $U \ket{0^n}$ in the computational basis for an unknown $n$-qubit unitary $U$ drawn from the Haar measure (Haar).
For any algorithm $\mathcal{A}$ taking in $N$ samples and outputting $\{0, 1\}$ to distinguish the two settings, we have
\begin{equation}
    \left| \Pr_{x_i \sim \mathrm{Unif}}[\mathcal{A} = 1] - \Pr_{x_i \sim \mathrm{Haar}}[\mathcal{A} = 1] \right| \leq \frac{N^2}{2^{n-1}}.
\end{equation}
When $N = \mathrm{poly}(n)$, no algorithm can distinguish the two settings with more than negligible probability.
\end{lemma}

In our proof, we will use the following lemma from Ref.~\cite{chen2021exponential} for summing over all $N!$ permutations.

\begin{lemma}[Lemma 5.12 in \cite{chen2021exponential}]\label{lem:maincombo}
For any collection of pure states $\ket{\psi_1},\ldots,\ket{\psi_N}$, \begin{equation}
\sum_{\pi\in S_N}\Tr\left(\pi\bigotimes^N_{i=1}\ketbra{\psi_i}{\psi_i}\right) \ge 1. \label{eq:maincombo}
    \end{equation}
\end{lemma}
\begin{proof}[Proof of Lemma~\ref{lem:dist-Haar-unif}]
    The maximum success probability of any algorithm for distinguishing between the uniform distribution and the output distribution of a Haar-random unitary is determined by the total variation distance between the following two distributions,
    \begin{align}
        p_{\mathrm{Unif}}(x_1, \ldots, x_N) &= \left(\frac{1}{2^n}\right)^N,\\
        p_{\mathrm{Haar}}(x_1, \ldots, x_N) &= \E_{U \sim \mathcal{E}_H}\left[ \, \prod_{i=1}^N \bra{x_i} U \ketbra{0^n}{0^n} U^\dagger \ket{x_i} \, \right],
    \end{align}
    where $\mathcal{E}_H$ is the Haar random unitary ensemble.
    In particular, the optimal success probability is equal to the total variation distance,
    \begin{equation}
        \max_{\mathcal{A}} \left| \Pr_{x_i \sim \mathrm{Unif}}[\mathcal{A} = 1] - \Pr_{x_i \sim \mathrm{Haar}}[\mathcal{A} = 1] \right| = \sum_{x_1, \ldots, x_N} \left| p_{\mathrm{Unif}}(x_1, \ldots, x_N) - p_{\mathrm{Haar}}(x_1, \ldots, x_N) \right|.
    \end{equation}
    Using the reformulation of total variation distance, we have $\sum_{z} \left| p_{1}(z) - p_{2}(z) \right| = 2 \sum_{z} \max( p_{1}(z) - p_{2}(z), 0)$ for any distributions $p_1, p_2$. Hence, the optimal success probability can be written as
    \begin{align}
        2 \sum_{x_1, \ldots, x_N \in \{0, 1\}^n} \max\left( \left(\frac{1}{2^n}\right)^N - \E_{\ket{\psi}} \left[ \prod_{i=1}^N \braket{x_i}{\psi}\!\braket{\psi}{x_i} \right],  0\right),
    \end{align}
    where $\E_{U}$ is over the Haar measure and $\E_{\ket{\psi}}$ is over the uniform measure on pure $n$-qubit states.
    We rewrite the term in the $\max(\cdot, 0)$ using the fact that $\E[ \ketbra{\psi}{\psi}^{\otimes N}] = \frac{1}{2^n(2^n+1)\cdots(2^n+N-1)} \sum_{\pi \in S_N} \pi$,
    \begin{align}
        \E_{\ket{\psi}} \left[ \prod_{i=1}^N \braket{x_i}{\psi}\!\braket{\psi}{x_i} \right] = \frac{1}{2^n(2^n+1)\cdots(2^n+N-1)} \sum_{\pi \in S_N} \Tr\left( \pi \bigotimes_{i=1}^N \ketbra{x_i}{x_i} \right).
    \end{align}
    Using Lemma~\ref{lem:maincombo}, we have
    \begin{equation}
        \E_{\ket{\psi}} \left[ \prod_{i=1}^N \braket{x_i}{\psi}\!\braket{\psi}{x_i} \right] \geq \frac{1}{2^n(2^n+1)\cdots(2^n+N-1)}.
    \end{equation}
    Together, this implies that
    \begin{equation}
        \max\left( \left(\frac{1}{2^n}\right)^N - \E_{\ket{\psi}} \left[ \prod_{i=1}^N \braket{x_i}{\psi}\!\braket{\psi}{x_i} \right],  0\right) \leq \frac{1}{2^{nN}} \left(1 - \frac{1}{ \left(1 + \frac{1}{2^n} \right) \ldots \left(1 + \frac{N-1}{2^n}\right)} \right).
    \end{equation}
    Therefore, the success probability is upper bounded by
    \begin{align}
    2 \sum_{x_1, \ldots, x_N \in \{0, 1\}^n}\frac{1}{2^{nN}} \left(1 - \frac{1}{ \left(1 + \frac{1}{2^n} \right) \ldots \left(1 + \frac{N-1}{2^n}\right)} \right) &\leq 
    2\left(1 - \frac{1}{ \left(1 + \frac{N}{2^n}\right)^{N}} \right) \leq \frac{N^2}{2^{n-1}}.
    \end{align}
    This concludes the proof.  
\end{proof}

We now return to the first statement of the theorem.
Building upon the methods in Ref.~\cite{nietner2023average}, we prove that the output distributions of the two-layer brickwork ensemble with $\xi = \omega(\log n )$ are far-from-uniform with high probability over the circuit ensemble.
\begin{proposition}[Far from uniformity]\label{prop:far-from-uniform}
    When $U$ is drawn from an $\varepsilon$-approximate unitary $8$-design with relative error, the output distributions $P_U$ satisfy
    \begin{equation}\label{eq:farfromuniform}
        \mathrm{Pr}[\mathrm{TV}(P_U,(2^{-n})_{x\in\{0,1\}^n})]\geq \Omega(1)]\geq 1-O(\varepsilon+2^{-n}),
    \end{equation}
    where $\mathrm{TV}$ denotes the total variation distance defined by $\mathrm{TV}(v,w)=\frac12 ||v-w||_1$.
\end{proposition}
\noindent We note that $\varepsilon=\mathrm{negl}(n)$ can be achieved with $\omega(\log n)$ depth quantum circuits.

We will prove the far-from-uniform property as in Ref.~\cite{nietner2023average}, which makes use of a more general concentration inequality for $k$-norms of output distributions:
\begin{theorem}\label{theorem:relativeerror_p-norm_concentration}
    Let $\mathcal{E}$ be an $\varepsilon$-approximate unitary $2k$-design with relative error. Then,
    \begin{equation}\label{eq:relativeerror_p-norm_concentration}
        \Pr_{U\sim \mathcal{E}}\left[\left|\sum_{x\in\{0,1\}^n} P_U(x)^k-\frac{k!}{2^{(k-1)n}}\right|\geq \frac{ak!}{2^{(k-1)n}} \right]\leq O(\varepsilon+k^22^{-n})\times a^{-2}.
    \end{equation}
\end{theorem}
 \noindent The proof consists of a straightforward moment calculation, where Eq.~\eqref{eq:relativeerror_p-norm_concentration} is obtained by applying Chebyshev's inequality for the random variable $X= \sum_{x\in\{0,1\}^n} P_U(x)^k$.
  We also use the following lemma (see e.g.~\cite{nietner2023average}):
  \begin{lemma}\label{lem:probability_moments}
    For any positive integer $t$ and for $U\sim \mathcal{E}_H$ sampled from the Haar ensemble, the moments of the output probability values satisfy
    \begin{equation}
    \E_{U\sim \mathcal{E}_H}\bigg[\prod_{i=1}^{\ell} |P_U(x_i)|^{\lambda_i}\bigg] = \frac{\prod_{i=1}^{\ell} (\lambda_i)!}{\prod_{i=0}^{k-1}(D+i)}\,.
    \end{equation}
    where the $x_i$'s are distinct and $\lambda\vdash k$ is an integer partition of $k$, i.e.\ $\lambda = (\lambda_1,\ldots,\lambda_\ell)$.
\end{lemma}
\begin{proof}[Proof of Theorem~\ref{theorem:relativeerror_p-norm_concentration}]
We calculate:
\begin{align}
    \begin{split}
        \mathrm{Var}_{U\sim \mathcal{E}}(X)&=\E_{U\sim \mathcal{E}}X^2-\left(\E_{ U\sim \mathcal{E}}X\right)^2\\
        &\leq (1+\varepsilon)\E_{U \sim \mathcal{E}_H}X^2 -(1-\varepsilon)^2\left(\E_{U \sim 
 \mathcal{E}_H}X\right)^2\\
        &=(1+\varepsilon)\sum_{x,y\in\{0,1\}^n}\E P_U(x)^k P_U(y)^k-(1-\varepsilon)^2\frac{k!^2}{\prod_{i=0}^{k-1}(D+i)^2}\\
        &= (1+\varepsilon)\left(\sum_{x\in\{0,1\}^n}\E P_U(x)^{2k}+\sum_{x\neq y\in\{0,1\}^n}\E P_U(x)^k P_U(y)^k\right)-(1-\varepsilon)^2\frac{2^{2n} k!^2}{\prod_{i=0}^{k-1}(D+i)^2}\\
        &=(1+\varepsilon)\left(2^{n}\frac{k!}{\prod_{i=0}^{2k-1}(D+i)}+2^{n}(2^{n}-1)\frac{k!^2}{\prod_{i=0}^{2k-1}(D+i)}\right)-(1-\varepsilon)^2\frac{2^{2n} k!^2}{\prod_{i=0}^{k-1}(D+i)^2}\\
        &\leq \frac{k!^2}{2^{2n(k-1)}}\left(O(\varepsilon)+O(k^22^{-n})\right).
    \end{split}
\end{align}
We can now apply Chebyshev's inequality stated as
\begin{equation}
    \mathrm{Pr}_{U\sim\mathcal{E}}[|X-\E X|\geq a']\leq \frac{\mathrm{Var}X}{a'^2}
\end{equation}
with $a'=a\E X$.
\end{proof}

\begin{proof}[Proof of Proposition~\ref{prop:far-from-uniform}]
    The proof proceeds similarly to the proof of Theorem 10 in Ref.~\cite{nietner2023average}, so we only briefly describe it here.
    We apply Berger's inequality~\cite{berger1997fourth}
    \begin{equation}\label{eq:berger}
        ||\bullet||_1\geq \frac{||\bullet||^3_2}{||\bullet||^2_4}
    \end{equation}
    for $\bullet =P_U-2^{-n}$.
    We then apply Eq.~\eqref{eq:relativeerror_p-norm_concentration} to prove concentration of both the numerator and denominator in Eq.~\eqref{eq:berger}.
    Concentration by at most a constant relative error then holds with probability $\mathcal{O}(\varepsilon+2^{-n})$.
    A union bound over the numerator and denominator yields Eq.~\eqref{eq:farfromuniform}.
\end{proof}


\begin{thebibliography}{100}
\bibitem{fn2}
Oftentimes, definitions of topological order are stated colloquially in terms
  of constant-depth circuits, i.e.~$\ell = \mathcal{O}(1)$. However, more
  precise definitions are nearly always robust up to any depth that is
  sub-extensive in the system diameter, e.g.~$\ell = o(\sqrt{n})$ in 2D. We
  refer to Appendix~\ref{app: topological order} for a detailed discussion.
\bibitem{fn3}
For smaller $|B|$, for example at $k=2$, one can apply $V_{ABC}$ to a product
  state and perform a SWAP test on either $A$ or $C$, which will detect that
  information has remained localized and not scrambled to the larger system.
\end{thebibliography}
\end{document}